\documentclass{scrartcl}

\usepackage[T1]{fontenc}
\usepackage[ansinew]{inputenc}
\usepackage{amsmath, amsthm, amssymb}
\usepackage{dsfont}
\usepackage{graphicx}
\usepackage{subfig}
\usepackage{bm}

\newcommand{\norm}[1]{\left\| #1 \right\|}  
\newcommand{\scprd}[1]{\left\langle #1 \right\rangle}  
\newcommand{\ddt}[1]{\frac{d} {dt} #1}   
\newcommand{\ddtl}[1]{\frac{d^l}{dt^l} #1}  
\newcommand{\N}{\mathbb{N}}  

\newcommand{\Q}{\mathbb{Q}}
\newcommand{\R}{\mathbb{R}}
\newcommand{\C}{\mathbb{C}}
\newcommand{\rg}{\operatorname{rg}}
\newcommand{\rk}{\operatorname{rk}}
\newcommand{\supp}{\operatorname{supp}}
\newcommand{\spn}{\operatorname{span}} 
\newcommand{\dist}{\operatorname{dist}}

\DeclareMathOperator*{\esssup}{ess-sup}  
\newcommand{\eps}{\varepsilon}

\newcommand{\ol}{\overline}

\renewcommand{\Re}{\operatorname{Re}}
\renewcommand{\Im}{\operatorname{Im}}

\numberwithin{equation}{section}

\newtheorem{thm}{Theorem}[section]
\newtheorem{cor}[thm]{Corollary}
\newtheorem{prop}[thm]{Proposition}
\newtheorem{lm}[thm]{Lemma}
\newtheorem{cond}[thm]{Condition}

\theoremstyle{definition} \newtheorem{ex}[thm]{Example}
\theoremstyle{definition}

\title{Adiabatic theorems with and without spectral gap condition for non-semisimple spectral values}

\author{Jochen Schmid\\  
\small Fachbereich Mathematik, Universität Stuttgart, D-70569 Stuttgart, Germany\\
\small jochen.schmid@mathematik.uni-stuttgart.de}    

\date{}

\begin{document}

\maketitle

\begin{abstract}
\small{ \noindent 
We establish adiabatic theorems with and without spectral gap condition for general 
operators $A(t): D(A(t)) \subset X \to X$ with possibly time-dependent domains in a Banach space $X$. 
We first prove adiabatic theorems with uniform and non-uniform spectral gap condition (including a slightly extended adiabatic theorem of higher order). In these adiabatic theorems the considered spectral subsets $\sigma(t)$ have only to be compact -- in particular, they need not consist of eigenvalues. 
We then prove an adiabatic theorem without spectral gap condition 
for not necessarily (weakly) semisimple 
eigenvalues: in essence, it is only required there that the considered spectral subsets $\sigma(t) = \{ \lambda(t) \}$ consist of eigenvalues $\lambda(t) \in \partial \sigma(A(t))$
and 
that there exist projections $P(t)$ reducing $A(t)$ 
such that $A(t)|_{P(t)D(A(t))}-\lambda(t)$ is nilpotent and $A(t)|_{(1-P(t))D(A(t))}-\lambda(t)$ is injective with dense range in $(1-P(t))X$ for almost every~$t$. 
In all these theorems, the regularity conditions imposed on  $t \mapsto A(t)$, $\sigma(t)$, $P(t)$ are 
fairly mild.
We explore the strength of the presented adiabatic theorems in numerous 
examples. 
And finally, we apply the adiabatic theorems for time-dependent domains to obtain -- in a very simple way -- adiabatic theorems for operators $A(t)$ defined by symmetric sesquilinear forms.} 
\end{abstract}

{ \small \noindent \emph{2010 Mathematics Subject Classification:} 34E15, 34G10, 35Q41, 47D06
} 

\section{Introduction}


Adiabatic theory for general -- as opposed to skew self-adjoint -- operators is a quite young area of research which was initiated in~\cite{NenciuRasche92} and 
developed further 
in~\cite{AbouSalem07}, \cite{Joye07}, \cite{AvronGraf11}. 
It is concerned -- in general 
terms -- with densely defined linear operators $A(t): D(A(t)) \subset X \to X$ in a Banach space $X$ over $\C$, subsets $\sigma(t)$ of $\sigma(A(t))$ and bounded projections $P(t)$ in $X$ (for $t \in I := [0,1]$) such that
\begin{itemize}
\item $A(t)$ is closed for every $t \in I$ and the initial value problems
\begin{gather*}
x' = A(\eps s) x \quad (s \in [0,\frac 1 \eps]) \quad \text{or}  \quad x' = \frac 1 \eps A(t) x \quad (t \in [0,1]) 
\end{gather*}
with initial conditions $x(s_0) = y$ or $x(t_0) = y$ (and initial times $s_0 \in [0,\frac 1 \eps)$, $t_0 \in [0,1)$) 
are well-posed on the spaces $D(A(t))$ for every value of the slowness parameter $\eps \in (0,\infty)$  
(\textbf{Condition~1.1}).
\item  $\sigma(t)$ is a compact subset of $\sigma(A(t))$ for every $t \in  I$ 
(\textbf{Condition~1.2}).
\item $P(t)$ commutes with $A(t)$ and $A(t)|_{P(t)D(A(t))}$ resp.~$A(t)|_{(1-P(t))D(A(t))}$, in some natural sense (specified below),
is spectrally related to $\sigma(t)$ as much resp.~as little as possible for every $t \in I$ except possibly for some few $t$ (\textbf{Condition~1.3}). 
\end{itemize}


What one 
wants to know first of all in adiabatic theory is the following: when -- under which additional conditions on $A$, $\sigma$ and $P$ -- does the evolution $U_{\eps}$ generated by the operators $\frac 1 \eps A(t)$ (Condition~1.1) approximately follow the spectral subspaces $P(t)X$ related to the spectral subsets $\sigma(t)$ of $A(t)$ (Condition~1.3) as the slowness parameter $\eps$ tends to $0$?
Shorter (and more precisely): under which conditions is it true that
\begin{align}  \label{eq: aussage des adsatzes}
(1-P(t)) U_{\eps}(t) P(0) \longrightarrow 0 \quad (\eps \searrow 0)
\end{align}
(with respect to a certain operator topology) 
for all $t \in I$? 
\emph{Adiabatic theorems} are, by definition, theorems that give 
such conditions. 
And this way of speaking is, etymologically, quite appropriate since the word ``adiabatos'' simply means ``admitting no transitions'' and, after all, adiabatic theorems just state that for small $\eps$ there will be almost no more transitions from the space $P(0)X$ to the spaces $(1-P(t))X$. 
We distinguish adiabatic theorems \emph{with spectral gap condition (uniform or non-uniform)} and adiabatic theorems \emph{without spectral gap condition} depending on whether $\sigma(t)$ is isolated in $\sigma(A(t))$ for every $t \in I$ (uniformly or non-uniformly) or not (Section~2.1). 
\bigskip

%





In this paper, we are going to prove adiabatic theorems with and without spectral gap condition for not necessarily (weakly) semisimple spectral values $\lambda(t) \in \sigma(t)$ of general -- not necessarily skew self-adjoint --
linear operators $A(t)$ with time-independent domains (Section~3 and~4) or time-dependent domains (Section~5 and~6). 
We thereby carry further the adiabatic theory of~\cite{AbouSalem07}, \cite{Joye07}, \cite{AvronGraf11}, for instance, and develop a general and systematic adiabatic theory with a special emphasis on the case without spectral gap. 
%
In more detail, the contents 
of the present paper can be described 
as follows.
\smallskip


In the 
adiabatic theorems with (uniform or non-uniform) spectral gap condition (Section~3 and~5.1), 
the relation between $P(t)$ and $A(t)$, $\sigma(t)$ -- vaguely described in Condition~1.3 above -- 
will 
be assumed to be as follows: 
for all but countably many $t \in I$, $P(t)$ is a bounded projection 
commuting with $A(t)$, such that $P(t)D(A(t)) = P(t)X$ and 
\begin{align*}
\sigma(A(t)|_{P(t)D(A(t))}) = \sigma(t)  \text{ \, whereas \, }  \sigma(A(t)|_{(1-P(t))D(A(t))}) = \sigma(A(t)) \setminus \sigma(t).
\end{align*}
(We will, more briefly, call such projections \emph{associated with $A(t)$ and $\sigma(t)$}.) 
Such projections always exist in the case with spectral gap and they are uniquely determined by $A(t)$ and $\sigma(t)$. 
Indeed, they are, of course, just given by
\begin{align*}
P(t) = \frac{1}{2 \pi i} \int_{\gamma_t} (z-A(t))^{-1} \, dz,
\end{align*}
where $\gamma_t$ 
is a cycle in $\rho(A(t))$ 
with indices 
$\operatorname{n}(\gamma_t, \sigma(t)) = 1$ and $\operatorname{n}(\gamma_t, \sigma(A(t)) \setminus \sigma(t)) = 1$ (Riesz projection). 
In particular, in the special situation where $\sigma(t) = \{ \lambda(t) \}$ and $\lambda(t)$ is a pole of the resolvent of $A(t)$, these projections satisfy
\begin{align} \label{eq: gl 1.2}
P(t)X = \ker(A(t)-\lambda(t))^{k} \quad \text{and} \quad  (1-P(t))X = \rg(A(t)-\lambda(t))^{k}
\end{align}    
for all $k \in \N$ with $k \ge m(t)$, where $m(t)$ is 
the order of the pole. 
%
In~\cite{NenciuRasche92}, \cite{AbouSalem07}, \cite{Joye07}, \cite{AvronGraf11} -- which, so far, 
are 
the only works proving 
adiabatic theorems with spectral gap for not necessarily skew self-adjoint operators $A(t)$ -- 
the abovementioned special situation of singletons $\sigma(t) = \{ \lambda(t) \}$ consisting of poles 
is studied: 
\cite{NenciuRasche92}, \cite{AbouSalem07}, \cite{AvronGraf11} deal with semisimple eigenvalues $\lambda(t)$ 
-- which means 
that the order $m(t)$ of the pole in~\eqref{eq: gl 1.2} is equal to $1$ for every $t \in I$ --
and \cite{Joye07}, under analyticity conditions, deals with 
not necessarily semisimple eigenvalues $\lambda(t)$ of finite 
algebraic multiplicity $m_0$. 
In the present paper 
the spectral subsets $\sigma(t)$ are also allowed to consist of 
essential singularities $\lambda(t)$ of the resolvent of $A(t)$ (Example~\ref{ex: lambda wesentl sing}). 
\smallskip


In the 
adiabatic theorems without spectral gap condition (Section~4 and~5.2), it will be assumed that $\sigma(t) = \{\lambda(t)\}$ for eigenvalues $\lambda(t) \in \partial \sigma(A(t))$, and the relation between $P(t)$ and $A(t)$, $\sigma(t)$ -- vaguely described in Condition~1.3 above -- 
will, fairly 
naturally, 
be assumed to be as follows: 
for almost every $t \in I$, 
$P(t)$ is a bounded projection 
commuting with $A(t)$, such that $P(t)D(A(t)) = P(t)X$ and 
\begin{gather*}
A(t)|_{P(t)D(A(t))} - \lambda(t) \text{\: is nilpotent whereas \:} A(t)|_{(1-P(t))D(A(t))} - \lambda(t) \text{\: is injective} \\
\text{with dense range in } (1-P(t))X.
\end{gather*}
(We will, more briefly, call such projections \emph{weakly associated with $A(t)$ and $\lambda(t)$}.)
Such projections do not always exist in the case without spectral gap 
-- but in case of existence they are uniquely determined by $A(t)$ and $\lambda(t)$. 
%
%
%
Indeed (by the first theorem of Section~2.1), such projections satisfy
\begin{align} \label{eq: gl 1.3}
P(t)X = \ker(A(t)-\lambda(t))^{k} \quad \text{and} \quad  (1-P(t))X = \overline{\rg(A(t)-\lambda(t))^{k}}
\end{align}    
for all $k \in \N$ with $k \ge m(t)$, where $m(t)$ is the order of nilpotence. Apart from regularity conditions 
and some resolvent estimate, 
the 
assumption 
that $P(t)$ be weakly associated with $A(t)$ and $\lambda(t) \in \partial \sigma(A(t))$ is already everything that has to be required in the adiabatic theorems without spectral gap condition of the present paper.
%
In~\cite{AvronGraf11} and~\cite{dipl} -- which, so far, are 
the only works establishing adiabatic theorems without spectral gap for not necessarily skew self-adjoint operators $A(t)$ -- 
the more special situation is studied 
where the $\lambda(t)$ are weakly semisimple eigenvalues 
in the sense that the order $m(t)$ of nilpotence in~\eqref{eq: gl 1.3} is equal to $1$ for almost every $t \in I$.
In numerous other papers, adiabatic theorems without spectral gap for skew self-adjoint operators $A(t)$ have been established 
-- both for eigenvalues $\lambda(t)$ 
(as in \cite{AvronElgart99}, \cite{Bornemann98}, \cite{Teufel01}, \cite{Teufel03}, \cite{AbouSalemFröhlich05}) 
and for resonances $\lambda(t)$ (as in~\cite{AbouSalemFröhlich07}, \cite{ElgartHagedorn11}).
%
%
%
%
%
%
\smallskip

As has already been indicated, 
we will develop 
here -- in contrast to the existing literature~-- adiabatic theorems 
for general 
operators $A(t)$ both in the case of time-independent domains (Section~3 and~4) \emph{and} in the case of time-dependent domains (Section~5), 
which two kinds of adiabatic theorems are 
related as follows: if in the presented adiabatic theorems for time-independent domains the regularity conditions are strenghtened, then these theorems 
become special cases of the respective adiabatic theorems for time-dependent domains.
We organize 
the theory in such a way that -- apart from 
some few modifications -- the proofs carry over from the case of time-independent to the case of time-dependent domains. 
Section~5.3 is devoted to an 
extension of the adiabatic theorem of higher order (also called superadiabatic theorem) from~\cite{JoyePfister93} or~\cite{Nenciu93} 
to the case of time-dependent domains, which extension -- albeit a bit technical -- is not difficult given the preliminaries of Section~2.
%
\smallskip


And finally, in Section~6, we apply the general adiabatic theorems for time-dependent domains of Section~5 to obtain -- in a very simple way -- adiabatic theorems with and without spectral gap condition (Section~6.2 and~6.3) for the special 
case where 
the $A(t) = i A_{a(t)}$ are 
skew self-adjoint operators defined by 
symmetric sesquilinear forms $a(t)$ having a time-independent form domain. In particular, the adiabatic theorem from~\cite{Bornemann98} 
is a special case of the theorem in Section~6.3, 
but we 
obtain it in a completely different and also simpler way. 
\smallskip


All presented adiabatic theorems are furnished with fairly mild regularity conditions on $t \mapsto A(t), \sigma(t), P(t)$. Indeed, in the adiabatic theorems for time-independent domains, for instance, it suffices to require of $A$ that $t \mapsto A(t)$ be $W^{1,\infty}_*$-regular (Section~2.1) and $(M,0)$-stable (Section~2.2), which two things are satisfied if, for instance, $t \mapsto A(t)x$ is continuously differentiable for all $x \in D = D(A(t))$ and $A(t)$ for every $t \in I$ generates a contraction semigroup in $X$. In particular, 
these regularity conditions are more general 
than those 
of the previously known adiabatic theorems. 
See the discussion 
after Theorem~\ref{thm: Kato85} and also see the examples (especially 
Example~\ref{ex: A nur W^{1,infty}-reg und nur (M,0)-stabil}). 
\smallskip


We complement 
the presented adiabatic theorems with examples in order to explore their scope 
and to demonstrate that some selected hypotheses 
are 
essential. 
All these examples will be of a simple standard structure explained 
in Section~2.4.
\bigskip



A few words on the interdependence of the various sections of this paper seem appropriate: the sections about the case with spectral gap (Section~3.1, 3.2, 5.1, 6.2 (interdependent in that order)) are independent of the sections about the case without spectral gap (Section~4.1, 4.2, 5.2, 6.3 (interdependent in that order)), and Section~5.3 can be read independently of all sections other than Section~2. 
Section~2 provides the most important preliminaries, most importantly, those on spectral theory 
(Section~2.1), well-posedness and evolution systems (Section~2.2), and adiabatic evolution systems (Section~2.3), but also those 
related to the regularity conditions of our adiabatic theorems. At  first reading, however, one may 
well ignore the less common notions of $W^{m,\infty}_*$-regularity or $(M,0)$-stability 
(explained in Section~2.1 and 2.2)
and replace them by the simpler notions of $m$ times strong continuous differentiability and contraction semigroups.
\smallskip


In the entire paper, $X$, $Y$, $Z$ denote Banach spaces over $\C$ and NOT, SOT, WOT stand for the norm, strong, weak operator topology 
of the Banach space $L(X,Y)$ of bounded linear operators from $X$ to $Y$ or -- unless explicitly stated otherwise -- of the Banach space $L(X) = L(X,X)$.
We will abbreviate
\begin{align*}
I := [0,1] \quad \text{and} \quad \Delta := \{ (s,t) \in I^2: s \le t\},
\end{align*}
and for evolution systems $U$ defined on $\Delta$ we will 
write $U(t) := U(t,0)$ for all $t \in I$, while $U_{\eps}$ for $\eps \in (0,\infty)$ will always denote the evolution system for $\frac 1 \eps A$ on $D(A(t))$ provided it exists.
As far as notational conventions on general spectral theory 
are concerned (in particular, concerning the not completely universal notion 
of continuous and residual spectrum), we follow the standard textbooks~\cite{DunfordSchwartz}, \cite{TaylorLay80} or~\cite{Yosida80}. 
And finally, the notation employed in the examples will be explained in Section~2.4. 

\section{Some preliminaries}

In this section, we provide the most important preliminaries for 
the adiabatic theorems of the subsequent sections. 
%
Section~2.1 introduces the central 
spectral-theoretical notions of associatedness and weak associatedness 
including the notion of weak semisimplicity, 
properly defines uniform and non-uniform spectral gaps, 
explains the notion of $W^{m,\infty}_*$-regularity of operator-valued maps (including 
its well-behavedness under the formation of products and inverses),
and finally furnishes continually used lemmas centering on one-sided differentiability. 
%
Section~2.2 then recalls the concept of well-posedness and evolution systems for non-autonomous linear evolution equations as well as the notion of $(M,0)$-stability due to Kato.
%
And finally, Section~2.3 introduces the terminology of adiabatic evolution systems while Section~2.4 explains the standard structure of the examples of this paper.

\subsection{Spectral theory, uniform and non-uniform spectral gaps, regularity properties of operator-valued functions, and one-sided differentiability}

As is clear from the setting described in Section~1, a first important point in general adiabatic theory is to precisely define what it means for projections $P(t)$ to be spectrally related to $A(t)$ and $\sigma(t)$ in the sense of Condition~1.3.
%
In order to do so 
we introduce the following notion of associatedness (which is completely canonical) 
and weak associatedness (which --~for non-normal, or at least, non-spectral operators 
is not 
canonical). 
%
%
%
Suppose $A: D(A) \subset X \to X$ is a densely defined closed linear map with $\rho(A) \ne \emptyset$, $\sigma \ne \emptyset$ is a compact isolated subset of $\sigma(A)$, $\lambda$ a not necessarily isolated spectral value of $A$, and $P$ a bounded projection in $X$.
We then say, following~\cite{TaylorLay80}, that 
\emph{$P$~is associated with $A$ and $\sigma$} if and only if 
$P$ commutes with $A$ (in short: $P A \subset A P$), $P D(A) = P X$ and
\begin{align*}
\sigma(A|_{PD(A)}) = \sigma  \text{ \, whereas \, }  \sigma(A|_{(1-P)D(A)}) = \sigma(A) \setminus \sigma.
\end{align*}
We say that \emph{$P$~is weakly associated with $A$ and $\lambda$} if and only if 
$P$ commutes with $A$, $P D(A) = P X$ and 
\begin{gather*}
A|_{PD(A)} - \lambda \text{\: is nilpotent whereas \:} A|_{(1-P)D(A)} - \lambda \text{\: is injective and} \\
\text{has dense range in } (1-P)X.
\end{gather*}
If above the order of nilpotence is at most $m$, we will often, 
more precisely, speak of $P$ as being \emph{weakly $m$-associated with $A$ and $\lambda$}. 
Also, we call $\lambda$ a \emph{weakly semisimple eigenvalue of $A$} if and only if $\lambda$ is an eigenvalue and there is a projection $P$ weakly $1$-associated with $A$ and $\lambda$.
%
%
%
%
\smallskip

We point out that, for 
spectral values $\lambda$ of a densely defined operator $A$ with $\rho(A) \ne \emptyset$ and bounded projections $P$ commuting with $A$, it is fairly natural to take the vaguely described spectral relatedness of Condition~1.3 
to mean precisely that $P$ is weakly associated with $A$ and $\lambda$.
%
Indeed, it is more than natural to take this condition 
to mean at least that 
$A|_{P D(A)}$ is bounded with $\sigma(A|_{PD(A)}) = \{ \lambda \}$ (or in other words, that $P D(A) = PX$ and $A|_{PD(A)}-\lambda$ is quasinilpotent) 
and that $\lambda \notin \sigma_p(A|_{(1-P)D(A)})$. 
(It is important to notice here that, if $\lambda$ is non-isolated in $\sigma(A)$, then it must belong to $\sigma(A|_{(1-P)D(A)})$ by the closedness of spectra.)
And then it is natural to furhter require 
that $A|_{PD(A)}-\lambda$ be even nilpotent (instead of 
only quasinilpotent) and that $\lambda$ belong to the continuous (instead of the residual) spectrum of $A|_{(1-P)D(A)}$, which finally is nothing but the weak associatedness of $P$ with $A$ and $\lambda$. 
\smallskip

%
%
%
%
It is well-known that if $A$ is densely defined with $\rho(A) \ne \emptyset$ and $\sigma \ne \emptyset$ is compact and isolated in $\sigma(A)$, then there exists a unique projection $P$ associated with $A$ and $\sigma$, namely 
\begin{align*}
P := \frac{1}{2 \pi i} \int_{\gamma} (z-A)^{-1} \, dz
\end{align*}
(occasionally 
called the \emph{Riesz projection of $A$ onto $\sigma$}), where $\gamma$ is a cycle in $\rho(A)$ with indices $\operatorname{n}(\gamma, \sigma) = 1$ and $\operatorname{n}(\gamma, \sigma(A) \setminus \sigma) = 0$. (See, for instance, Theorem~2.14 and Proposition~2.15 of~\cite{dipl} or the standard textbook~\cite{GohbergGoldbergKaashoek} for detailed proofs of these well-known facts.) 
Also, if $\sigma = \{ \lambda \}$ consists of a pole of $(\,.\,-A)^{-1}$ 
of order $m$ and if $P$ is associated with $A$ and $\sigma$, then
\begin{align*}
PX = \ker(A-\lambda)^k \quad \text{and} \quad (1-P)X = \rg(A-\lambda)^k 
\end{align*}
for all $k \in \N$ with $k \ge m$ (Theorem~5.8-A of~\cite{Taylor58}). 
%
%
We now prove a similar theorem for weak associatedness 
revealing, in particular, that on the one hand, for a given operator $A$ and a spectral value $\lambda$ of $A$ no weakly associated projection need 
exist in general and that, on the other hand, if there exists any weakly associated projection then it is already unique. 
(See the third remark following the theorem.)  
This theorem complements 
Corollary~2.2 of~\cite{Lay70}, which covers only 
the case of isolated spectral values. 
It will be crucial 
in the presented adiabatic theorems without spectral gap condition. 

\begin{thm} \label{thm: typ mögl für PX und (1-P)X}
Suppose $A: D(A) \subset X \to X$ is a densely defined closed linear map with $\rho(A) \ne \emptyset$, 
$\lambda \in \sigma(A)$, and $P$ is a bounded projection in $X$. 
If $P$ is weakly $m$-associated with $A$ and $\lambda$, 
then
\begin{align*}
PX = \ker(A-\lambda)^k \quad \text{and} \quad (1-P)X = \overline{ \rg(A-\lambda)^k }
\end{align*}
for all $k \in \N$ with $k \ge m$
and, moreover, $\lambda$ is a pole of $(\,.\,-A)^{-1}$ if and only if $\rg(A-\lambda)^k$ is closed for some (all) $k \ge m$, in which case $P$ is the projection associated with $A$ and~$\lambda$. 
\end{thm}

\begin{proof}
We may clearly assume that $\lambda = 0$ 
because $P$, being weakly associated with $A$ and $\lambda$, is also weakly associated with $A-\lambda$ and $0$.
Set $M := PX$ and $N := (1-P)X$. 
We first show that $M = \ker A^k$ for all $k \ge m$. Since $A|_{PX} = A|_{PD}$ is nilpotent of order $m$, $A^k|_{PX} = ( A|_{PX} )^k = 0$ and hence $M = PX \subset \ker A^k$ for all $k \ge m$. And since $A|_{(1-P)D(A)}$ is injective, 
\begin{align*}
A^k|_{(1-P)D(A^k)} = ( A|_{(1-P)D(A)} )^k
\end{align*}
is injective as well and hence $\ker A^k \subset PX = M$ for all $k \in \N$.
We now show that $N = \overline{\rg A^k}$ for all $k \ge m$. As $PX = \ker A^k$ for $k \ge m$, we have
\begin{align*}
\rg A^k = A^k PD(A^k) + A^k (1-P)D(A^k) = (1-P) A^k D(A^k) \subset (1-P)X = N
\end{align*} 
and therefore $\overline{\rg A^k } \subset N$ for all $k \ge m$. It remains to show that the reverse inclusion $N \subset \overline{ \rg A^k }$ holds true for all $k \in \N$ and this will be done by induction over $k$. Since $A|_{(1-P) D(A)}$ has dense range in $(1-P)X = N$, the desired inclusion is clearly satisfied for $k = 1$. Suppose now that $N \subset \overline{ \rg A^k }$ is satisfied for some arbitrary $k \in \N$. Since 
\begin{align*}
\rg A|_{(1-P)D(A)} = A (1-P)D(A) = A (z_0 - A)^{-1} N
\end{align*}
and since $A (z_0-A)^{-1}$ is a bounded operator for every $z_0 \in \rho(A)$, it then follows by the induction hypothesis that $A (z_0-A)^{-1} N \subset \overline{ \rg A^{k+1} }$ and hence 
\begin{align*}
N = \overline{ \rg A|_{(1-P)D(A)} } \subset \overline{ \rg A^{k+1} },
\end{align*}
as desired.
%
In order to see the characterization of the pole property for $\lambda = 0$ 
(under the weak associatedness hypothesis), 
we have only to observe that, by what has just been shown, the ascent of $A$ is less than or equal to $m$ 
and to use Theorem~5.8-A and Theorem~5.8-D of~\cite{Taylor58} or Theorem~V.10.1 and Theorem~V.10.2 of~\cite{TaylorLay80}.
%
\end{proof}

%
Some remarks, 
answering in particular the existence and uniqueness question for weak associatedness, are in order.
\smallskip

1. We also have the following simple converse of the above theorem: if one has a direct sum decomposition
\begin{align*}
X = \ker(A-\lambda)^m \oplus \overline{ \rg(A-\lambda)^m }
\end{align*}
for some $m \in \N$ (where $A: D(A) \subset X \to X$ is densly defined with $\rho(A) \ne \emptyset$ and $\lambda \in \sigma(A)$) and 
and if the projection $P$ corresponding to the above decomposition of $X$ commutes with $A$, 
then 
$P$ is weakly $m$-associated with $A$ and $\lambda$.
(It should be noticed here that $P$ is automatically bounded because $(A-\lambda)^m$ is closed 
by $\rho(A) \ne \emptyset$.)
%
%
%
%
%
\smallskip

2. It is easy to see -- with the help of Volterra operators and shift operators (Section~2.4) as building blocks -- that the nilpotence, injectivity, and dense range requirements (encapsulated in the weak associatedness) 
are all 
essential for the conclusion of the above theorem. 
And -- as opposed to the case where $A$ is normal -- isolatedness of $\lambda$ in $\sigma(A)$ alone does not imply the closedness of the spaces $\rg (A-\lambda)^k$ 
in the theorem above. (Just take $A := \operatorname{diag}(0,V)$ on $X := L^2(I) \times L^2(I)$ with the Volterra operator $V$ defined by $(Vf)(t) := \int_0^t f(\tau) \,d\tau$, $\lambda := 0$, and $P$ the orthogonal projection onto $L^2(I) \times 0$.)
\smallskip

3. If $A: D(A) \subset X \to X$ is densely defined with $\rho(A) \ne \emptyset$ and $\lambda \in \sigma(A)$, then the above theorem implies the following two 
facts about weak associatedness. 
(i) In general, there exists 
no projection 
weakly associated with $A$ and $\lambda$. 
(Indeed, for $A := \operatorname{diag}(0,S)$ on $X := \ell^2(\N) \times \ell^2(\N)$ and $\lambda := 0$, where $S$ is the right shift operator on $\ell^2(\N)$, one has 
\begin{align*}
\ker(A-\lambda)^k + \overline{\rg(A-\lambda)^k} \ne X
\end{align*}
for all $k \in \N$. Also see the remarks after Example~\ref{ex: A_2(t) nicht diagb} and Example~\ref{ex: ablop}.) 
%
(ii) If, however, there exists a projection weakly associated with $A$ and $\lambda$, then there is only one such projection. 
(Indeed, if $P$ and $Q$ are two projections weakly associated with $A$ and $\lambda$ and if the orders of nilpotence of $A|_{PD(A)}-\lambda$ and $A|_{QD(A)}-\lambda$ are $m$ and~$n$ respectively, then 
\begin{gather*}
PX = \ker(A-\lambda)^m = \ker(A-\lambda)^n = QX, \\ (1-P)X = \overline{\rg(A-\lambda)^m} = \overline{\rg(A-\lambda)^n} = (1-Q)X
\end{gather*}
and therefore $P = Q$.) 
\smallskip

A rather well-understood general class of operators $A$ for which the existence of weakly associated projections is always guaranteed 
is furnished by the class of 
\emph{spectral operators of scalar type}  (Definition~XVIII.2.12 in~\cite{DunfordSchwartz}) 
or bounded \emph{spectral operators of finite type} (Definition~XV.5.3 in~\cite{DunfordSchwartz}). 
(Spectral operators of scalar type comprise, 
for instance, the normal operators 
or 
the generic one-dimensional periodic Schrödinger operators (Remark~8.7 of~\cite{GesztesyTkachenko09}).)
In both cases the weakly associated projections are just given by 
$P^{A}( \{ \lambda \} )$ where $P^{A}$ denotes the spectral measure -- or, in the terminology of~\cite{DunfordSchwartz}, the spectral resolution -- 
of the spectral operator $A$ (Definition~XV.2.5 and~XVIII.2.1 
of~\cite{DunfordSchwartz}). 
(In order to see this, use Theorem~XVIII.2.11 of~\cite{DunfordSchwartz} in the former case and in the latter case notice that
\begin{gather*}
A|_{PX} - \lambda \text{\: is quasinilpotent whereas \:} A|_{(1-P)X} - \lambda \text{\: is injective and} \\
\text{has dense range in } (1-P)X
\end{gather*}
by Corollary~XV.8.5 of~\cite{DunfordSchwartz} 
and that the quasinilpotence above is even a nilpotence by the finite type assumption.) 
We point out, however, that adiabatic theory is 
by no means confined to spectral operators $A(t)$ (Example~\ref{ex: A_2(t) nicht diagb}).
\smallskip

%
%
%
%
%
%
%
%
Associatedness and weak associatedness, for isolated 
spectral values $\lambda$ of an 
operator~$A$, 
are related as follows. 
If $\lambda$ is a pole of $(\,.\,-A)^{-1}$, then associatedness and weak associatedness coincide: a projection $P$ is then associated with $A$ and $\lambda$ if and only if it is weakly associated with $A$ and $\lambda$ (by Theorem~\ref{thm: typ mögl für PX und (1-P)X} and the first remark following it). 
If, however, $\lambda$ is 
an essential singularity of $(\,.\,-A)^{-1}$, then associatedness and weak associatedness have nothing to do with each other: 
a projection $P$ associated with $A$ and $\lambda$ can then not possibly be weakly associated with $A$ and $\lambda$, and vice versa. (Indeed, if a projection $P$ is both associated and weakly $m$-associated with $A$ and $\lambda$, then 
\begin{align*}
z \mapsto (z-A)^{-1} &= (z-A)^{-1}P + (z-A)^{-1}(1-P) \\
&= \sum_{k=0}^{m-1} \frac{ (A|_{PD(A)}-\lambda)^k }{ (z-\lambda)^{k+1} } \, P 
+ \big( z-A|_{(1-P)D(A)} \big)^{-1} (1-P) 
\end{align*}   
has a pole of order $m$ at $\lambda$.) 
%
%
%
Additionally, since semisimple spectral values of $A$ are, by definition, poles of $(\,.\,-A)^{-1}$ of order $1$ (and, in particular, eigenvalues),    
it follows 
that semisimplicity and weak semisimplicity coincide for isolated 
points $\lambda$ of $\sigma(A)$ that are poles. 
\smallskip

We close this paragraph on spectral theory by noting that in reflexive spaces weak associatedness carries over to the dual maps -- provided that some core condition is satisfied, which is the case for semigroup generators, for instance (Proposition~II.1.8 of~\cite{EngelNagel}).
In the presented adiabatic theorem without spectral gap condition for reflexive spaces, this will be important. 
Associatedness carries over to dual maps as well, of course (Section~III.6.6 of~\cite{KatoPerturbation80}) 
-- but this will not be needed in the sequel.

\begin{prop} \label{prop: schwache assoziiertheit, dual}
Suppose $A: D(A) \subset X \to X$ is a densely defined closed linear map in the reflexive space $X$ such that $\rho(A) \ne \emptyset$ and $D(A^k)$ is a core for $A$ for all $k \in \N$. If $P$ is weakly $m$-associated with $A$ and $\lambda \in \sigma(A)$, then $P^*$ is weakly $m$-associated with $A^*$ and $\lambda$. 
\end{prop}

\begin{proof}
We begin by showing -- by induction over $k \in \N$ -- the preparatory statement that
\begin{align} \label{eq: indbeh 1}
(A^k)^* = (A^*)^k 
\end{align}
for all $k \in \N$, which might also be of independent interest (notice that $D(A^k)$ 
being a core for $A$ is dense in $X$, so that $(A^k)^*$ is really well-defined).
Clearly, \eqref{eq: indbeh 1} is true for $k = 1$ and, assuming that it is true for some arbitrary $k \in \N$, we now show that $(A^{k+1})^* = (A^*)^{k+1}$ holds true as well.
It is easy to see that $(A^*)^{k+1} \subset (A^{k+1})^*$ and it remains to see that $D((A^{k+1})^*) \subset D((A^*)^{k+1})$. So let $x^* \in D((A^{k+1})^*)$. We show that
\begin{align} \label{eq: zwbeh für k-indschritt}
x^* \in D((A^k)^*) \quad \text{and} \quad (A^k)^* x^* \in D(A^*),
\end{align}
from which it then follows -- by the induction hypothesis -- that $x^* \in D((A^*)^{k+1})$ as desired.
In order to prove that $x^* \in D((A^k)^*)$ we show that 
\begin{align*} 
x^* \in D((A^l)^*)
\end{align*}
for all $l \in \{ 1, \dots, k \}$ -- by induction over $l \in \{ 1, \dots, k \}$ and by working with suitable powers of $(A^*-z_0)^{-1} = ((A-z_0)^{-1})^*$, 
where $z_0$ is an arbitrary point of $\rho(A^*) = \rho(A) \ne \emptyset$ (Theorem~III.5.30 of~\cite{KatoPerturbation80}). 
In the base step of the 
induction, notice that for all $y \in D(A)$
\begin{align*}
&\big\langle (A^*-z_0)^{-k} (A^{k+1})^* x^*, y \big\rangle = \big \langle x^*, A^{k+1} (A-z_0)^{-k} y \big \rangle \\
&\qquad \qquad \qquad = \big \langle x^*, (A-z_0) y \big \rangle + \sum_{i = 0}^k \binom{k+1}{i} z_0^{k+1-i} \big \langle (A^*-z_0)^{-k+i} x^*, y \big \rangle, 
\end{align*}
from which it follows that $x^* \in D((A-z_0)^*) = D(A^*)$. 
In the inductive step, 
assume that $x^* \in D(A^*), \dots, D((A^l)^*)$ for some arbitrary $l \in \{1, \dots, k-1 \}$. Since for all $y \in D(A^{l+1})$
\begin{align*}
&\big \langle (A^*-z_0)^{-(k-l)} (A^{k+1})^* x^*, y \big \rangle = \big \langle x^*, A^{k+1} (A-z_0)^{-(k-l)} y \big \rangle \\
&\qquad \qquad = \big \langle x^*, (A-z_0)^{l+1} y \big \rangle + \sum_{i=k-l+1}^{k} \binom{k+1}{i} z_0^{k+1-i} \big \langle x^*, (A-z_0)^{-(k-l)+i} y \big
\rangle \\ 
&\qquad \qquad \quad + \sum_{i=0}^{k-l} \binom{k+1}{i} z_0^{k+1-i} \big \langle (A^*-z_0)^{-(k-l)+i} x^*,  y \big \rangle,
\end{align*}
it follows by the induction hypothesis of the $l$-induction and by applying the binomial formula to $(A-z_0)^{-(k-l)+i} y$ for $i \in \{ k-l+1, \dots, k+1 \}$ that $x^* \in D((A^{l+1})^*)$.
So the $l$-induction is finished and it remains to show that $(A^k)^* x^* \in D(A^*)$. Since $D(A^{k+1})$ by assumption is a core for $A$, there is for every $y \in D(A)$ a sequence $(y_n)$ in $D(A^{k+1})$ such that 
\begin{align*}
\big \langle (A^k)^* x^*, A y \big \rangle = \lim_{n \to \infty} \big \langle (A^k)^* x^*, A y_n \big \rangle = \lim_{n \to \infty} \big \langle x^*, A^{k+1} y_n \big \rangle 
= \big \langle (A^{k+1})^* x^*, y \big \rangle.
\end{align*}
It follows that $(A^k)^* x^* \in D(A^*)$ and this yields -- together with the induction hypothesis of the $k$-induction -- that $x^* \in D((A^*)^{k+1})$, which finally ends the proof the preparatory statement~\eqref{eq: indbeh 1}.
\smallskip

After this preparation we can now move on to the main part of the proof where we assume, without loss of generality, that $\lambda = 0$ and exploit 
the first remark after Theorem~\ref{thm: typ mögl für PX und (1-P)X} to show that $P^*$ is weakly $m$-associated with $A^*$ and $\lambda = 0$.
$A^*$ is densely defined (due to the reflexivity of $X$ (Theorem~III.5.29 of~\cite{KatoPerturbation80})) 
with $\rho(A^*) = \rho(A) \ne \emptyset$ (Theorem~III.5.30 of~\cite{KatoPerturbation80}) and 
\begin{align*}
P^* A^* \subset (AP)^* \subset (PA)^* = A^* P^*
\end{align*}
because $AP \supset PA$.
Since $(A^m)^* = (A^*)^m$ by~\eqref{eq: indbeh 1} and since $PX = \ker A^m$ and $(1-P)X = \overline{ \rg A^m }$ (by Theorem~\ref{thm: typ mögl für PX und (1-P)X}), we further have 
\begin{gather*}
P^*X^* = \ker(1-P)^* = ((1-P)X)^{\perp} = (\overline{\rg A^m})^{\perp} = \ker(A^m)^* = \ker (A^*)^m \\
\text{and} \\
(1-P^*)X^* = \ker P^* = (P X)^{\perp} = (\ker A^m)^{\perp} = ( \ker (A^m)^{**} )_{\perp} = \overline{\rg (A^m)^* } = \overline{\rg (A^*)^m },
\end{gather*}
where in the fourth equality of the second line the closedness of $A^m$ (following from $\rho(A) \ne \emptyset$) and the 
reflexivity of $X$ have been used. (In the above relations, we denote by $U^{\perp} := \{ x^* \in Z^*: \scprd{x^*, U^*} = 0 \}$ and $V_{\perp} := \{ x \in Z: \scprd{V, x} = 0 \}$ the annihilators of subsets $U$ and $V$ of a normed space $Z$ and its dual $Z^*$,  respectively.)
It is now clear from the first remark after Theorem~\ref{thm: typ mögl für PX und (1-P)X} that $P^*$ is weakly $m$-associated with $A^*$ and $\lambda = 0$ 
and we are done.
\end{proof}
\bigskip

We continue by properly defining what exactly we mean by uniform and non-uniform spectral gaps. 
Suppose that $A(t): D(A(t)) \subset X \to X$, for every $t$ in some compact interval $J$, is a densely defined closed linear map and that $\sigma(t)$ is a compact subset of $\sigma(A(t))$ for every $t \in J$. We then speak of a \emph{uniform resp.~non-uniform spectral gap for $A$ and $\sigma$} if and only if $\sigma(t)$ is isolated in $\sigma(A(t))$ for every $t \in J$ (in short: there is a \emph{spectral gap for $A$ and $\sigma$}) 
and $\sigma(\,.\,)$ is resp.~is not uniformly isolated in $\sigma(A(\,.\,))$. (Uniform isolatedness 
means that there is a $t$-independent constant $r_0 > 0$ such that
\begin{align*}
\overline{U}_{r_0}(\sigma(t)) \cap \sigma(A(t)) := \{ z \in \C: \dist(z,\sigma(t)) \le r_0 \} \cap \sigma(A(t)) = \sigma(t)
\end{align*}
for every $t \in I$, of course.)
Also, we say that \emph{$\sigma(\,.\,)$ falls into $\sigma(A(\,.\,)) \setminus \sigma(\,.\,)$ at the point $t_0 \in J$} if and only if there is a sequence $(t_n)$ in $J$ 
converging to $t_0$ such that
\begin{align*}
\dist(\sigma(t_n), \sigma(A(t_n)) \setminus \sigma(t_n)) \longrightarrow 0 \quad (n \to \infty).
\end{align*} 
It is clear that the set of points at which $\sigma(\,.\,)$ falls into $\sigma(A(\,.\,)) \setminus \sigma(\,.\,)$ is closed. 
Also, it follows by the compactness of $J$ that 
a spectral gap for $A$ and $\sigma$ is uniform if and only if $\sigma(\,.\,)$ at no point falls into $\sigma(A(\,.\,)) \setminus \sigma(\,.\,)$. 
The following proposition gives a criterion (in terms of some mild regularity conditions on $t \mapsto A(t)$, $\sigma(t)$, $P(t)$) 
for a spectral gap for $A$ and $\sigma$ to be even uniform.  
It is of some interest in the third remark at the beginning of Section~3.3 (and is crucial 
in the applied example to neutron transport theory presented in~\cite{dipl}). 
We refer to Section~IV.2.4 and Theorem~IV.2.25 of~\cite{KatoPerturbation80} for a definition and a characterization of \emph{convergence} (and hence, continuity) \emph{in the generalized sense} and to Section~IV.3 of~\cite{KatoPerturbation80} for the definition of \emph{upper and lower semicontinuity} of set-valued functions $t \mapsto \sigma(t)$. 

\begin{prop} \label{prop: zshg isoliert und glm isoliert}
Suppose that $A(t): D(A(t)) \subset X \to X$ is a closed linear map for every $t$ in a compact interval $J$ and that $t \mapsto A(t)$ is continuous in the generalized sense. Suppose further that $\sigma(t)$ for every $t \in J$ is a compact and isolated subset of $\sigma(A(t))$ such that $\sigma(\,.\,)$ falls into $\sigma(A(\,.\,)) \setminus \sigma(\,.\,)$ at $t_0 \in J$, and let $t \mapsto \sigma(t)$ be upper semicontinuous at $t_0$. Finally, for every $t \in J$ let $P(t)$ be the 
projection associated with $A(t)$ and $\sigma(t)$. Then $t \mapsto P(t)$ is discontinuous at $t_0$ and 
\begin{align*}
\limsup_{n \to \infty} \bigl( \rk P(t_n) \bigr) \le \rk P(t_0) - 1
\end{align*}
for all sequences $(t_n)$ such that $t_n \longrightarrow t_0$ and $\dist ( \sigma(t_n), \sigma(A(t_n)) \setminus \sigma(t_n) ) \longrightarrow 0$. 
\end{prop}

See~\cite{dipl} (Proposition~5.3) for a proof.
Clearly, one also has the following converse of the above proposition: if $t \mapsto A(t)$ is continuous in the generalized sense as above and $t \mapsto \sigma(t)$ is even \emph{continuous} (that is, 
upper and lower semicontinuous) then uniform isolatedness of $\sigma(\,.\,)$ in $\sigma(A(\,.\,))$ implies that $t \mapsto P(t)$ is continuous. (Use Theorem~IV.3.15 of~\cite{KatoPerturbation80}.)
\\

We now turn to regularity properties of vector-valued and, in particular, operator-valued functions. We start by briefly recalling those facts on vector-valued Sobolev spaces that will be needed in the sequel (see~\cite{ArendtBatty00} and~\cite{Amann95} for more detailed expositions). We follow the notational conventions of~\cite{ArendtBatty00}. In particular, \emph{($\mu$-)measurability} of a $Y$-valued map on a complete measure space $(X_0, \mathcal{A}, \mu)$ will not only mean that this map is $\mathcal{A}$-measurable  but also that it is $\mu$-almost separably-valued, whereas the notion of \emph{($\mu$-)strong measurability} will be reserved for operator-valued maps that are pointwise $\mu$-measurable. 
Suppose $J$ is a non-trivial interval and $p \in [1,\infty) \cup \{ \infty \}$. Then $W^{1,p}(J,X)$ is defined to consist of those (equivalence classes of) $p$-integrable functions $f: J \to X$ for which there is a $p$-integrable function $g: J \to X$ (called a weak derivative of $f$) such that 
\begin{align*}
\int_J f(t) \varphi'(t) \, dt = - \int_J g(t) \varphi(t) \, dt
\end{align*}
for all $\varphi \in C_c^{\infty}(J°,\C)$. As usual, \emph{$p$-integrability} of a function $f: J \to X$ (with $p \in [1,\infty) \cup \{ \infty \}$) means that $f$ is 
measurable and $\norm{f}_p:= (\int_{J} \norm{f(\tau)}^p \, d\tau)^{1/p}  < \infty$ if $p \in [1,\infty)$ or $\norm{f}_{p} := \esssup_{t \in J} \norm{f(t)} < \infty$ if $p = \infty$. If $f$ is in $W^{1,p}(J,X)$ and $g_1$, $g_2$ are two weak derivatives of $f$, then $g_1 = g_2$ almost everywhere, so that up to almost everywhere equality there is only one weak derivative of $f$ which is denoted by $\partial f$. It is well-known that $W^{1,p}(J,X)$ is a Banach space w.r.t.~the norm $\norm{\,.\,}_{1,p}$ whith $\norm{f}_{1,p} := \norm{f}_p + \norm{ \partial f }_p$ for $f \in W^{1,p}(J,X)$. It is also well-known that the space $W^{1,\infty}(J,X)$ -- just like 
$W^{1,p}(J,X)$ for $p \in [1,\infty)$ -- can be characterized by means of 
indefinite integrals: $W^{1,\infty}(J,X)$ consists of those (equivalence classes of) $\infty$-integrable functions $f: J \to X$ for which there is an $\infty$-integrable function $g$ such that for some (and hence every) $t_0 \in J$ 
\begin{align*}
f(t) = f(t_0) + \int_{t_0}^t g(\tau) \,d\tau \text{\, for all } t \in J,
\end{align*}
or, equivalently (by Lebesgue's differentiation theorem), $W^{1,\infty}(J,X)$ consists of (equivalence classes of) $\infty$-integrable functions $f: J \to X$ which are differentiable almost everywhere and whose (pointwise) derivative $f'$ is $\infty$-integrable such that for some (and hence every) $t_0 \in J$ 
\begin{align*}
f(t) = f(t_0) + \int_{t_0}^t f'(\tau) \,d\tau  \text{\, for all } t \in J.
\end{align*}
Additionally, the pointwise derivative $f'$ of an $f \in W^{1,\infty}(J,X)$ equals the weak derivative $\partial f$ almost everywhere.
It follows from this characterization that, in case $X$ is reflexive (or more generally, satisfies the Radon--Nikod\'{y}m property), $W^{1,\infty}(J,X)$ consists exactly of the (equivalence classes of) Lipschitz continuous functions (where one inclusion is completely trivial and independent of the Radon--Nikod\'{y}m property, of course).
\smallskip


We now move on to define -- following the introduction of Kato's work~\cite{Kato85} -- $W^{m,\infty}_*$-\emph{regularity} which shall be used 
in all our adiabatic theorems with time-independent domains. An operator-valued function $J \ni t \mapsto A(t) \in L(X,Y)$ on a compact interval $J$ is said to belong to $W^{0,\infty}_*(J, L(X,Y)) = L^{\infty}_*(J,L(X,Y))$ if and only if $t \mapsto A(t)$ is strongly measurable and $t \mapsto \norm{A(t)}$ is essentially bounded. And $t \mapsto A(t)$ is said to belong to $W^{1,\infty}_*(J, L(X,Y))$ if and only if there is $B \in L^{\infty}_*(J,L(X,Y))$ (called a $W^{1,\infty}_*$-\emph{derivative} of $A$) such that for some (and hence every) $t_0 \in J$ 
\begin{align*}
A(t)x = A(t_0)x + \int_{t_0}^t B(\tau)x \,d\tau \text{\, for all } t \in J \text{ and } x \in X.
\end{align*}
$W^{m,\infty}_*(J,L(X,Y))$ for arbitrary $m \in \N$ is defined recursively, of course.
\smallskip

We point out that the $W^{m,\infty}_*$-spaces (unlike the $W^{m,\infty}$-spaces), by definition, consist of functions (of operators) rather than equivalence classes of such functions.  
It is obvious from the characterization of $W^{1,\infty}(J,Y)$ by way of indefinite integrals that, if $t \mapsto A(t)$ is in $W^{1,\infty}_*(J, L(X,Y))$, then $t \mapsto A(t)x$ is (the continuous representative of an element) in $W^{1,\infty}(J,Y)$. 
In particular, $W^{1,\infty}_*$-regularity implies Lipschitz continuity w.r.t.~the norm topology, and furthermore, $W^{1,\infty}_*$-regularity can be thought of as being not much more than Lipschitz continuity (in view of the above remarks in conjunction with the Radon--Nikod\'{y}m property). 
A simple and important criterion for $W^{1,\infty}_*$-regularity is furnished 
by the following proposition.

\begin{prop} \label{prop: WOT-stet db impl W^{1,infty}-reg}
Suppose $J \ni t \mapsto A(t) \in L(X,Y)$ is SOT- or WOT-continuously differentiable, where $J$ is a compact interval.
Then $t \mapsto A(t)$ is in $W^{1,\infty}_*(J, L(X,Y))$.
\end{prop}

\begin{proof}
 It is well-known that a weakly continuous map $J \to Y$ is almost separably valued, whence $t \mapsto A'(t)x$ is measurable 
(by Pettis' characterization of measurability). With the help of the Hahn--Banach theorem the conclusion readily follows. 
\end{proof}

It follows from Lebesgue's differentiation theorem that $W^{1,\infty}_*$-derivatives are essentially unique, more precisely: if $t \mapsto A(t)$ is in $W^{1,\infty}_*(J,L(X,Y))$ and $B_1$, $B_2$ are two $W^{1,\infty}_*$-derivatives of $A$, then one has for every $x \in X$ that $B_1(t)x = B_2(t)x$ for almost every $t \in J$. It should be emphasized that this last condition does \emph{not} imply that $B_1(t) = B_2(t)$ for almost every $t \in J$. (Indeed, take $J:= [0,1]$, $X := \ell^2(J)$ and define 
\begin{align*}
A(t):= 0 \text{\,  as well as \,} B_1(t)x := \scprd{ e_t, x} e_t \text{\, and \,} B_2(t)x := 0 
\end{align*}
for $t \in J$ and $x \in X$, where $e_t(s) := \delta_{s \, t}$.
Then, for every $x \in X$, $B_1(t)x$ is different from $0$ for at most countably many $t \in J$, and it follows that $B_1$ and $B_2$ both are $W^{1,\infty}_*$-derivatives of $A$, but $B_1(t) \ne B_2(t)$ for every $t \in J$.) 
\smallskip

In the presented adiabatic theorems for time-independent domains (Section~\ref{sect: adsätze mit sl} and \ref{sect: adsätze ohne sl}), we will make much use of the following lemma stating that $W^{1,\infty}_*$-regularity carries over to products and inverses. It is noted in the introduction of~\cite{Kato85} for separable spaces. We prove it here since it is not proved in~\cite{Kato85} and, more importantly, since it is not a priori clear 
-- (almost) separability being crucial for measurability -- 
that the separability assumption of~\cite{Kato85} is actually not needed 
for this lemma. 
An analogue of this lemma for SOT- and WOT-continuous differentiability is well-known (and easily proved with the help of the theorem of Banach--Steinhaus). 

\begin{lm} \label{lm: prod- und inversenregel}
(i) Suppose that $t \mapsto A(t)$ is in $W^{1,\infty}_*(J,L(X,Y))$ and that $t \mapsto B(t)$ is in $W^{1,\infty}_*(J,L(Y,Z))$ where $J = [a,b]$. Then $t \mapsto B(t)A(t)$ is in $W^{1,\infty}_*(J,L(X,Z))$ and $t \mapsto B'(t)A(t) + B(t)A'(t)$ is a $W^{1,\infty}_*$-derivative of $B A$ for every $W^{1,\infty}_*$-derivative $A'$, $B'$ of $A$ or $B$, respectively. \\
\noindent (ii) Suppose that $t \mapsto A(t)$ is in $W^{1,\infty}_*(J,L(X,Y))$ and that $A(t)$ is bijective onto $Y$ 
for every $t \in J$. 
Then $t \mapsto A(t)^{-1}$ is in $W^{1,\infty}_*(J,L(Y,X))$ and $t \mapsto - A(t)^{-1} A'(t) A(t)^{-1}$ is a $W^{1,\infty}_*$-derivative of $A^{-1}$ for every $W^{1,\infty}_*$-derivative $A'$ of $A$.
\end{lm}

\begin{proof}
We begin with some general preparatory considerations. Whenever an operator-valued map $t \mapsto C(t)$ is in $W^{1,\infty}_*(J,L(X,Y))$, 
we shall write $\tilde{C}$ for the map on $\tilde{J}:= [a-1,b+1]$ obtained from $C$ by trivially extending it by $C(a)$ to the left and by $C(b)$ to the right. 
Also we shall write 
\begin{align*}
\tilde{C}_n(t)x := \bigl(j_{\frac{1}{n}} * \tilde{C}(\,.\,)x \bigr)(t) = \int_{ \tilde{J} } j_{\frac{1}{n}}(t-r) \tilde{C}(r)x \,dr
\end{align*}
for all $t \in \tilde{J}$ and all $x \in X$. It is clear that $t \mapsto \tilde{C}(t)$ is in $W^{1,\infty}_*(\tilde{J},L(X,Y))$. 
In particular, $t \mapsto \tilde{C}(t)x$ is the continuous representative of an element of $W^{1,\infty}(\tilde{J},Y)$ and hence of $W^{1,1}(\tilde{J},Y)$ for every $x \in X$, 
from which we conclude -- using the usual facts on mollification (which are proved in exactly the same way as in the scalar case) -- that for every $x \in X$
\begin{align*}
\tilde{C}_n(\,.\,)x \big|_J      &\xrightarrow[{\norm{\,.\,}_{\infty}}]{}  \tilde{C}(\,.\,)x \big|_J = C(\,.\,)x \quad (n \to \infty), 
\end{align*}
%
\begin{gather*}
\big\| \tilde{C}_n(t)x \big\| \le \sup_{\tau \in \tilde{J}} \big\|  \tilde{C}(\tau)   \big\| \, \norm{x}  \quad (t \in J, n \in \N), \\
\big\| \tilde{C}_n'(t)x \big\| = \big\| \bigl(j_{\frac{1}{n}} * \tilde{C}'(\,.\,)x \bigr)(t)    \big\| \le \esssup_{\tau \in \tilde{J}} \big\|  \tilde{C}'(\tau)   \big\| \, \norm{x} \quad (t \in J, n \in \N),
\end{gather*}
\begin{gather*}
\tilde{C}_n(\,.\,)x \big|_J      \xrightarrow[{\norm{\,.\,}_{1,1}}]{}  \tilde{C}(\,.\,)x \big|_J = C(\,.\,)x \quad (n \to \infty),  
\end{gather*}
where $\tilde{C}'$ denotes an arbitrary $W^{1,\infty}_*$-derivative of $\tilde{C}$. 
\smallskip

(i) We fix arbitrary $W^{1,\infty}_*$-derivatives $A'$, $B'$ of $A$ and $B$ and prove that $t \mapsto B'(t)A(t) + B(t)A'(t)$ is in $L^{\infty}_*(J, L(X,Z))$ and that 
\begin{align*}
B(t)A(t)x = B(a)A(a)x + \int_a^t B'(\tau)A(\tau)x + B(\tau)A'(\tau)x \, d\tau
\end{align*}
for every $t \in J$ and $x \in X$. It is easy to see that $t \mapsto B'(t)A(t) + B(t)A'(t)$ is indeed in $L^{\infty}_*(J, L(X,Z))$ (by virtue of Lemma~A~4 of~\cite{Kato73} stating that strong measurability of operator-valued functions carries over to products). 
Additionally, it is clear from the SOT-continuous differentiability of $t \mapsto \tilde{A}_n(t)$ and $t \mapsto \tilde{B}_n(t)$ that 
\begin{align*}
\tilde{B}_n(t) \tilde{A}_n(t)x = \tilde{B}_n(a) \tilde{A}_n(a)x + \int_a^t \tilde{B}_n'(\tau) \tilde{A}_n(\tau)x + \tilde{B}_n(\tau) \tilde{A}_n'(\tau)x \, d\tau
\end{align*}
for every $t \in J$, 
 $x \in X$, and $n \in \N$. 
We now fix $t \in J$ and $x \in X$ and choose $\tilde{a}$, $\tilde{b}$ such that $\norm{A(\tau)}, \norm{A'(\tau)} \le \tilde{a}$ and $\norm{B(\tau)}, \norm{B'(\tau)} \le \tilde{b}$ for almost every $\tau \in J$. 
In virtue of the preparatory considerations above, we have 
\begin{gather*}  \label{eq: gl 1, prod- und inversenregel}
\int_a^t \tilde{B}_n(\tau) \tilde{A}_n'(\tau)x \, d\tau - \int_a^t B(\tau) A'(\tau)x \, d\tau \notag \\
= \int_a^t \tilde{B}_n(\tau) \bigl( \tilde{A}_n'(\tau)x - A'(\tau)x \bigr) \, d\tau + \int_a^t \bigl( \tilde{B}_n(\tau) - B(\tau) \bigr) A'(\tau)x \, d\tau
\longrightarrow 0 
\end{gather*}
as $n \to \infty$.
And since for every $\varepsilon > 0$ there is a partition $a=t_0 < t_1 < \dots < t_m = t$ of $[a,t]$ 
such that for all $i \in \{1, \dots, m\}$ one has $\sup_{\tau \in [t_{i-1},t_i]} \norm{A(\tau)x - A(t_i)x} < \eps / 2 \tilde{b}$ 
we see 
that the norm of
\begin{gather*}  \label{eq: gl 2, prod- und inversenregel}
\int_a^t \tilde{B}_n'(\tau) \tilde{A}_n(\tau)x \,d\tau - \int_a^t B'(\tau)A(\tau)x \, d\tau \notag \\
= \int_a^t \tilde{B}_n'(\tau) \bigl( \tilde{A}_n(\tau)x - A(\tau)x \bigr) \, d\tau 
+ \sum_{i=1}^m \int_{t_{i-1}}^{t_i} \bigl( \tilde{B}_n'(\tau) - B'(\tau) \bigr) \bigl( A(\tau)x - A(t_i)x \bigr) \, d\tau \notag \\
+ \sum_{i=1}^m \int_{t_{i-1}}^{t_i} \bigl( \tilde{B}_n'(\tau) - B'(\tau) \bigr) A(t_i)x \, d\tau
\end{gather*} 
is less than $3 \eps$ for sufficiently large $n$. Assertion~(i) thus follows.
\smallskip

(ii) We fix an arbitrary $W^{1,\infty}_*$-derivative $A'$ of $A$ and show that $t \mapsto A(t)^{-1} A'(t) A(t)^{-1}$ is in $L^{\infty}_*(J, L(Y,X))$ and that 
\begin{align*}
A(t)^{-1}y = A(a)^{-1}y  - \int_a^t  A(\tau)^{-1} A'(\tau) A(\tau)^{-1}y \, d\tau
\end{align*}
for every $t \in J$ and $y \in Y$.
As above, it follows from Lemma~A~4 of~\cite{Kato73} that $t \mapsto A(t)^{-1} A'(t) A(t)^{-1}$ is indeed in $L^{\infty}_*(J, L(X,Z))$. 
Since 
$\sup_{t \in J} \big\| \tilde{A}_n(t) - A(t) \big\| \longrightarrow 0$ as $n \to \infty$, there is an $n_0 \in \N$ such that $\tilde{A}_n(t)$ is invertible (bijective onto $Y$) for all $t \in J$ and $n \ge n_0$ and such that $\big\| \tilde{A}_n(t)^{-1} \big\|  \le \tilde{c} < \infty$ for all $t \in J$ and $n \ge n_0$. 
It follows from the SOT-continuous differentiability of $t \mapsto \tilde{A}_n(t)$ that $t \mapsto \tilde{A}_n(t)^{-1}$ 
is SOT-continuously differentiable as well, whence
\begin{align*}
\tilde{A}_n(t)^{-1}y = \tilde{A}_n(a)^{-1}y  - \int_a^t  \tilde{A}_n(\tau)^{-1} \tilde{A}_n'(\tau) \tilde{A}_n(\tau)^{-1}y \, d\tau
\end{align*}
for every $t \in J$, $y \in Y$ and $n \ge n_0$. 
We now fix $t \in J$ and $y \in Y$. As $\sup_{\tau \in J} \big\| \tilde{A}_n(\tau)^{-1} - A(\tau)^{-1} \big\| \longrightarrow 0$ as $n \to \infty$,
we obtain
\begin{gather*}  \label{eq: gl 3, prod- und inversenregel}
\int_a^t \tilde{A}_n(\tau)^{-1} \tilde{A}_n'(\tau) \tilde{A}_n(\tau)^{-1}y \, d\tau - \int_a^t A(\tau)^{-1} A'(\tau) A(\tau)^{-1}y \, d\tau \notag \\
= \int_a^t \tilde{A}_n(\tau)^{-1} \bigl(  \tilde{A}_n'(\tau) \tilde{A}_n(\tau)^{-1}y - A'(\tau) A(\tau)^{-1}y \bigr) \, d\tau \notag \\
+ \int_a^t \bigl( \tilde{A}_n(\tau)^{-1} - A(\tau)^{-1} \bigr) A'(\tau) A(\tau)^{-1}y \, d\tau 
\longrightarrow 0 
\end{gather*}
as $n \to \infty$, where the first term is treated as in the proof of~(i). 
Assertion~(ii) thus follows.
\end{proof}


We shall need the following simple product rule very often: it will always be used for estimating the difference of two evolution systems and for establishing adiabaticity of evolution systems. And furthermore, it will take the role of Lemma~\ref{lm: prod- und inversenregel} in the adiabatic theorems for time-dependent domains (Section~\ref{sect: adsätze für zeitabh domains}).


\begin{lm} \label{lm: prodregel rechtsseit db}
Suppose $C(t)$ is a bounded linear map in $X$ for every $t \in J = [a,b]$, let $t_0 \in [a,b)$, and let $Y_{t_0}$ be a dense subspace of $X$. Suppose that $t \mapsto C(t)y$ is right differentiable at $t_0$ for all $y \in Y_{t_0}$ and that the map $f: J \to X$ is right differentiable at $t_0$ and $f(t_0) \in Y_{t_0}$. Suppose finally that $\sup_{t \in J_{t_0}} \norm{ C(t) } < \infty$ for a neighbourhood $J_{t_0}$ of $t_0$.
Then $t \mapsto C(t)f(t)$ is right differentiable at $t_0$ with right derivative 
\begin{align*}
\partial_+ ( C(\,.\,) f(\,.\,) )(t_0) = \partial_+ C(t_0) f(t_0) + C(t_0) \partial_+ f(t_0).
\end{align*}
\end{lm}

\begin{proof}
We have 
\begin{align*}
&\frac{C(t_0 + h)f(t_0+h) - C(t_0)f(t_0)}{h} \\
&\qquad \qquad = C(t_0+h) \frac{f(t_0+h)-f(t_0)}{h} - \frac{C(t_0+h)f(t_0) - C(t_0)f(t_0)}{h}
\end{align*}
for positive and sufficiently small $h$. Since $\sup_{t \in J_{t_0}} \norm{ C(t) } < \infty$ and $Y_{t_0}$ is dense, we easily get that $C(t_0+h) \longrightarrow C(t_0)$ as $h \searrow 0$ w.r.t.~SOT, and the desired conclusion follows. 
\end{proof}

We shall also need the following lemma on the relation between right differentiability and the class $W^{1,\infty}$. 
It will be used very often -- especially in Section~\ref{sect: adsätze für zeitabh domains} -- 
in conjunction 
with the lemma above: Lemma~\ref{lm: prodregel rechtsseit db} will yield right differentiability of a given product and Lemma~\ref{lm: rechtsseit db und W^{1,infty}}, which is a variant of Corollary~2.1.2 of~\cite{Pazy83}, will then yield an integral representation 
for this product.

\begin{lm} \label{lm: rechtsseit db und W^{1,infty}} 
Suppose $f:J \to X$ is a continuous, right differentiable map on a compact interval $J = [a,b]$ such that the right derivative $\partial_+ f: [a,b) \to X$ is bounded. Then $f$ is in $W^{1,\infty}(J,X)$ and
\begin{align*}
f(t) = f(t_0) + \int_{t_0}^t \partial_+ f(\tau) \, d\tau
\end{align*}
for all $t_0, t \in J$. In particular, if $\partial_+ f$ is even continuous and continuously extendable to the right endpoint $b$, then $f$ is continuously differentiable.
\end{lm}

\begin{proof}
Since $\partial_+f$ is measurable (as the pointwise limit of a sequence of difference quotients) and $\partial_+ f$ is bounded, we have only to show that 
\begin{align*}
\int_{(a,b)} f(t) \varphi'(t) \, dt = - \int_{(a,b)} \partial_+ f(t) \varphi(t) \, dt \text{ \,\, for all } \varphi \in C_c^{\infty}((a,b), \C)
\end{align*}
in order to get $f \in W^{1,\infty}(J,X)$ (from which, in turn, the asserted integral representation follows by the continuity of $f$). 
So, let $\varphi \in C_c^{\infty}((a,b), \C)$ and denote by $\tilde{\varphi}$ and $\tilde{f}$ the zero extension of $\varphi$ and $f$ to the whole real line. Then 
\begin{align*}
\int_{(a,b)} f(t) \varphi'(t) \, dt &= \lim_{h \searrow 0} \int_{\R} \tilde{f}(t) \frac{ \tilde{\varphi}(t-h) - \tilde{\varphi}(t) }{-h} \, dt \\
&= - \lim_{h \searrow 0} \int_{\R} \frac{ \tilde{f}(t+h) - \tilde{f}(t) }{h} \tilde{\varphi}(t) \, dt 
= \int_{(a,b)} \partial_+ f(t) \varphi(t) \, dt,
\end{align*}
since $\supp \varphi \subset [a+\delta, b-\delta]$ for some $\delta > 0$ and since
\begin{align*}
\norm{ \frac{ f(t+h) - f(t) }{h} \varphi(t) } \le \sup_{\tau \in (a,b)} \norm{ \partial_+ f(\tau) } \norm{\varphi}_{\infty} < \infty
\end{align*}
for all $t \in [a+\delta, b-\delta]$ and $h \in (0,\delta)$ (which mean value estimate can be derived from the continuity and right differentiability of $f$ in a similar way as Lemma~III.1.36 of~\cite{KatoPerturbation80}). 
\end{proof}

\subsection{Well-posedness and evolution systems}


We briefly recall the fundamental notion of evolution systems for $A$ on $D(A(t))$ where $A$ is a family of linear operators. (See Definition~VI.9.2 in~\cite{EngelNagel} which also covers the more general case of evolution systems for $A$ on certain subspaces $Y_t$ of $D(A(t))$. We will only work with the special case $Y_t = D(A(t))$ since it is only in this case that an adiabatic theory with reasonably practical assumptions can be developed.) 
Suppose that $A(t): D(A(t)) \subset X \to X$ is a linear operator for every $t \in J = [a,b]$ and that $U(t,s)$ is a bounded linear operator in $X$ for every $(s,t) \in \Delta_J := \{ (s,t) \in J^2: s \le t \}$. 
Then $U$ is called an \emph{evolution system for $A$ on (the spaces) $D(A(t))$} if and only if the following holds true: 
\begin{itemize}
\item [(i)] $[s,b] \ni t \mapsto U(t,s)y$ is a continuously differentiable solution to the initial value problem 
$x' = A(t)x, \; x(s) = y$
for $y \in D(A(s))$ and $s \in [a,b)$ 
\item[(ii)] $U(t,s) U(s,r) = U(t,r)$ for all $(r,s), (s,t) \in \Delta_J$ 
and $\Delta_J \ni (s,t) \mapsto U(t,s)x$ is continuous for all $x \in X$.
\end{itemize}

We refer to~\cite{EngelNagel} (Definition~VI.9.1) for the definition of \emph{well-posedness of the initial value problems corresponding to $A$ on $D(A(t))$}. 
It is well-known that there is a one-to-one correspondence between well-posedness and evolution systems (Proposition~VI.9.3 of~\cite{EngelNagel}): if $A$ is a family of linear operators $A(t): D(A(t)) \subset X \to X$, then the initial value problems corresponding to $A$ are well-posed on $D(A(t))$ if and only if there is an evolution system $U$ for $A$ on $D(A(t))$.
In particular, it follows (by the uniqueness requirement in the definition of well-posedness) that if there is any evolution system for a given family $A$ on $D(A(t))$, then it is already unique.
In order to see the above-mentioned correspondence one essentially has only to combine Lemma~\ref{lm: prodregel rechtsseit db} and Lemma~\ref{lm: rechtsseit db und W^{1,infty}} with the following simple lemma, which will always be used when the difference of two evolution systems has to be dealt with. 

\begin{lm} \label{lm: zeitentw rechtsseit db}
Suppose $A(t): D(A(t)) \subset X \to X$ for every $t \in J = [a,b]$ is a densely defined linear map and suppose there is an evolution system $U$ for $A$ on $D(A(t))$. 
Then, for every $s_0 \in [a,t)$ and every $x_0 \in D(A(s_0))$, the map $[a,t] \ni s \mapsto U(t,s)x_0$ is right differentiable at $s_0$ with right derivative $-U(t,s_0) A(s_0) x_0$. In particular, if $D(A(t)) = D$ for all $t \in J$ and $s \mapsto A(s)x$ is continuous for all $x \in D$, then $[a,t] \ni s \mapsto U(t,s)x$ is continuously differentiable 
for all $x \in D$. 
\end{lm}

\begin{proof}
Since $U(t,s)U(s,r) = U(t,r)$ for $(r,s), (s,t) \in \Delta_J$ and since $\Delta_J \ni (s,t) \mapsto U(t,s)$ is SOT-continuous, we obtain for every $s_0 \in [a,t)$ and $x_0 \in D(A(s_0))$ that
\begin{align*}
\frac{ U(t,s_0+h)x_0 - U(t,s_0)x_0 }{h} = -U(t,s_0+h) \frac{ U(s_0+h,s_0)x_0 - x_0 }{h} 
\longrightarrow -U(t,s_0)A(s_0)x_0 
\end{align*}
as $h \searrow 0$ from which the assertions follow (remember Lemma~\ref{lm: rechtsseit db und W^{1,infty}}).
\end{proof}

We also briefly recall the notion of $(M,\omega)$-stability from~\cite{Kato70}: a family $A$ of linear operators $A(t): D(A(t)) \subset X \to X$ (where $t \in J$) is called \emph{$(M,\omega)$-stable} (for some $M \in [1,\infty)$ and $\omega \in \R$) if and only if $A(t)$ generates a strongly continuous semigroup on $X$ for every $t \in J$ and 
\begin{align*}
\norm{   e^{A(t_n) s_n}  \, \dotsm \, e^{A(t_1) s_1}   } \le M e^{\omega (s_1 + \, \dotsb \, + s_n) } 
\end{align*} 
for all $s_1, \dots, s_n \in [0,\infty)$ and all $t_1, \dots, t_n \in J$ satisfying $t_1 \le \dotsb \le t_n$ (with arbitrary~$n \in \N$). 
Clearly, a family $A$ of linear operators in $X$ is $(1,0)$-stable if and only if each member $A(t)$ of the family generates a contraction semigroup on $X$. It should be remarked 
that there are very simple examples -- relevant to adiabatic theory -- of $(M,0)$-stable families that fail to be $(1,0)$-stable (Example~\ref{ex: A nur W^{1,infty}-reg und nur (M,0)-stabil}).
When it comes to estimating perturbed evolution systems in Section~\ref{sect: adsätze mit sl} and~\ref{sect: adsätze ohne sl}, 
the following important fact (well-known from~\cite{Kato70}) will always -- and tacitly -- be used: 
if $A$ is an $(M,\omega)$-stable family of linear operators $A(t): D(A(t)) \subset X \to X$ for $t \in J$, $B(t)$ is a bounded operator in $X$ for $t \in J$ and $b := \sup_{t \in J} \norm{ B(t) }$ is finite, then $A + B$ is $(M, \omega + M b)$-stable. 
In our examples the following lemma will be important. 

\begin{lm} \label{lm: (M,w)-stabilität und ähnl.trf.}
Suppose $A_0$ is an $(M_0,\omega_0)$-stable family of 
operators $A_0(t): D(A_0(t)) \subset X \to X$ for $t \in J$ and $R(t): X \to X$ for every $t \in J$ is a bijective bounded operator such that $t \mapsto R(t)$ is in $W^{1,\infty}_*(J, L(X))$. Then the family $A$ with $A(t) := R(t)^{-1} A_0(t) R(t)$ is $(M,\omega)$-stable for some $M \in [1,\infty)$ and $\omega = \omega_0$.
\end{lm}

\begin{proof}
We may assume that $\omega_0 = 0$, since $(\tilde{M},\tilde{\omega})$-stability of a family $\tilde{A}$ is equivalent to the $(\tilde{M},0)$-stability of $\tilde{A}-\tilde{\omega}$.
Set $\norm{x}_t := d \, e^{-M_0 c t} \, \norm{ R(t) x }_{0 \, t}$ for $x \in X$ and $t \in J$, where $c := \esssup_{(s,t) \in J^2} \norm{ R'(t) R(s)^{-1} }$ 
and $d := \sup_{t \in J}  e^{M_0 c t} \norm{ R(t)^{-1} }$   
and the $\norm{\,.\,}_{0 \, t}$ are norms on $X$ associated with $A_0$ according to Proposition~1.3 of~\cite{Nickel00}. Then -- as can be gathered from the proof of Theorem~4.2 in~\cite{Kisynski63} -- 
the norms $\norm{\,.\,}_t$ satisfy the conditions (a), (b), (c) of Proposition~1.3 in~\cite{Nickel00} for the family $A$ with a certain $M \in [1,\infty)$ 
and therefore $A$ is $(M,0)$-stable, as desired.
\end{proof}

We now turn to sufficient conditions for the well-posedness of the initial value problems corresponding to a given family $A$ of linear operators $A(t)$. 
We will make very much use of the following condition.

\begin{cond} \label{cond: reg 1}
$A(t): D \subset X \to X$ for every $t \in I$ is a densely defined closed 
linear map such that $A$ is $(M,\omega)$-stable for some $M \in [1, \infty)$ and $\omega \in \R$ and such that $t \mapsto A(t)$ is in $W^{1,\infty}_*(I,L(Y,X))$, where $Y$ is the space $D$ endowed with the graph norm of $A(0)$. 
\end{cond}

As was noted (in Proposition~\ref{prop: WOT-stet db impl W^{1,infty}-reg}) above, the $W^{1,\infty}_*$-regularity requirement of Condition~\ref{cond: reg 1} is fulfilled if, for instance, $t \mapsto A(t)$ is SOT- or WOT-continuously differentiable. 
It follows from a famous theorem of Kato (Theorem~1 of~\cite{Kato73}) -- and, for seperable spaces, is explicitly remarked in Section~1 of~\cite{Kato85} 
-- that Condition~\ref{cond: reg 1} guarantees well-posedness. 

\begin{thm}[Kato] \label{thm: Kato85}
 Suppose $A(t): D \subset X \to X$ for every $t \in I$ is a linear map such that Condition~\ref{cond: reg 1} is satisfied. Then there is a unique evolution system $U$ for $A$ on $D$, 
and the following estimate holds true: 
\begin{align*}
\norm{U(t,s)} \le M e^{\omega (t-s)} \,\, \text{for all } (s,t) \in \Delta. 
\end{align*}
\end{thm}

Condition~\ref{cond: reg 1} does not only guarantee well-posedness, but it is also essentially everything we have to require of $A$ in the adiabatic theorems of Section~\ref{sect: adsätze mit sl} and~\ref{sect: adsätze ohne sl} for time-independent domains: indeed, we have only to add the requirement that $\omega = 0$ to Condition~\ref{cond: reg 1} to arrive at the hypotheses on $A$ of these theorems. 
In most adiabatic theorems in the literature -- for example those of ~\cite{AvronSeilerYaffe87}, \cite{AvronElgart99}, \cite{Teufel01}, \cite{Teufel03}, \cite{AbouSalemFröhlich05}, 
\cite{AbouSalem07} or~\cite{AvronGraf11} -- by contrast, the hypotheses on $A$ rest upon Yosida's theorem (Theorem~XIV.4.1 of~\cite{Yosida80}), which is reproduced, for instance, in Reed and Simon's book (Theorem~X.70 of~\cite{ReedSimon}) or Blank, Exner and Havl\'{i}\v{c}ek's book (Theorem~9.5.3 of~\cite{BlankExnerHavlicek}): in these adiabatic theorems it is required of $A$ that each $A(t)$ generate a contraction semigroup on $X$ and that an appropriate translate $A-z_0$ of $A$ satisfy the rather 
involved hypotheses of Yosida's theorem (or -- for example in the case of~\cite{AvronSeilerYaffe87} or~\cite{AvronGraf11} -- more convenient strengthenings thereof). 
It is shown in~\cite{evol} that this is the case 
if and only if $A(t)-z_0$, for every $t \in I$, is a boundedly invertible generator of a contraction semigroup on $X$ and 
\begin{align*}
t \mapsto A(t)x \text{ is continuously differentiable for all } x \in D.
\end{align*}
In particular, it follows 
that the hypotheses on $A$ of the adiabatic theorems of the present paper are more general than the respective hypotheses of the previously known adiabatic theorems -- and, of course, they are also striclty more general (which is demonstrated by the examples of Section~3 and~4). 
\smallskip

A trivial -- but nonetheless useful -- consequence of Theorem~\ref{thm: Kato85} is the following corollary establishing well-posedness for families $A$ where $A(t) = R(t)^{-1} A_0(t) R(t)$ and $D(A_0(t))$ is time-independent. A non-trivial sufficient condition for well-posedness in the case of time-dependent domains is furnished, for instance, 
by~\cite{Kisynski63} (which will be exploited in Section~\ref{sect: adsätze für A(t) von sesquilinearformen}) or by~\cite{Tanabe60}, \cite{KatoTanabe62},  \cite{FujieTanabe73}, \cite{AcquistapaceTerreni87}.

\begin{cor} \label{cor: Kato85}
Suppose $A_0$ is a family of linear maps $A_0(t): D \subset X \to X$ that satisfies Condition~\ref{cond: reg 1} and let $A(t) := R(t)^{-1} A_0(t) R(t)$ for $t \in I$, where $t \mapsto R(t)$ is in $W^{2,\infty}_*(I,L(X))$ and $R(t)$ is bijective onto $X$ for every $t \in I$. 
Then there is a unique evolution system $U$ for $A$ on $D(A(t))$. 
\end{cor}

\begin{proof}
Since $t \mapsto A_0(t) + R'(t) R(t)^{-1}$ is in $W^{1,\infty}_*(I,L(Y,X))$ by Lemma~\ref{lm: prod- und inversenregel} and since $A_0 + R' R^{-1}$ is $(M,\omega + M b)$-stable with $b := \sup_{t \in I} \norm{ R'(t) R(t)^{-1} }$, it follows from Theorem~\ref{thm: Kato85} that there is a unique evolution system $\tilde{U}_0$ for $A_0 + R' R^{-1}$ on $D$. Set $U(t,s) := R(t)^{-1} \tilde{U}_0(t,s) R(s)$ for $(s,t) \in \Delta$. Then $U$ is an evolution system for $A$ on $D(A(t))$, as is easily verified.
\end{proof}

In the adiabatic theorems with spectral gap condition of Section~\ref{sect: adsätze für zeitabh domains} (especially in the adiabatic theorem of higher order) the following well-expected perturbative proposition 
will be needed. It gives a perturbation series expansion for a perturbed evolution system if only this perturbed evolution exists. (See the classical example of Phillips (Example~6.4 of~\cite{Phillips53}) showing that the existence of the perturbed evolution really has to be required.)  

\begin{prop} \label{prop: störreihe für gestörte zeitentw}
Suppose that $A(t): D(A(t)) \subset X \to X$ is a densely defined linear map for every $t \in I$ and that $t \mapsto B(t) \in L(X)$ is WOT-continuous. Suppose further that there is an evolution system $U$ for $A$ on $D(A(t))$ and an evolution system $V$ for $A+B$ on $D(A(t))$. Then
\begin{itemize}
\item[(i)] $V(t,s) = \sum_{n=0}^{\infty} V_n(t,s)$, where $V_0(t,s) := U(t,s)$ and
\begin{align*}
V_{n+1}(t,s)x := \int_s^t U(t,\tau) B(\tau) V_n(\tau, s)x \, d\tau \text{ for } x \in X \text{ and } n \in \N \cup \{ 0 \}.
\end{align*}
\item[(ii)] If there are $M \in [1,\infty)$, $\omega \in \R$ such that $\norm{ U(t,s) } \le M e^{ \omega (t-s) }$ for $(s,t) \in \Delta$, then 
\begin{align*}
\norm{ V(t,s) } \le M e^{ (\omega + M b)(t-s) }
\end{align*}
for all $(s,t) \in \Delta$, where $b := \sup_{t \in I} \norm{ B(t)}$. 
And if, for every $(s,t) \in \Delta$, 
$U(t,s)$ is unitary  and $B(t)$ is skew symmetric, 
then $V(t,s)$ is unitary as well.
\end{itemize}
\end{prop}

\begin{proof}
Since weakly continuous maps on compact intervals are integrable (see the proof of Proposition~\ref{prop: WOT-stet db impl W^{1,infty}-reg}), it easily follows that the integrals defining the $V_n$ really exist 
and that $\tilde{V}(t,s) := \sum_{n = 0}^{\infty} V_n(t,s)$ exists uniformly in $(s,t) \in \Delta$. Also, it is easy to see -- applying Lemma~\ref{lm: prodregel rechtsseit db} and Lemma~\ref{lm: rechtsseit db und W^{1,infty}} to $[s,t] \ni \tau \mapsto U(t,\tau)V(\tau,s)x$ with $x \in D(A(s))$ -- that $V$ satisfies the same integral equation as $\tilde{V}$ from which assertion~(i) follows. Assertion~(ii) is a simple consequence of the series expansion in~(i).
\end{proof}

\subsection{Adiabatic evolutions and a trivial adiabatic theorem}

As has been explained in Section~1, the principal 
goal of adiabatic theory is to establish the convergence~\eqref{eq: aussage des adsatzes} or, in other words, to show that the evolution systems $U_{\eps}$ for $\frac 1 \eps A$ are, in some sense, approximately adiabatic w.r.t.~$P$ as $\eps \searrow 0$. We say that an evolution system for a family $A$ of linear operators $A(t): D(A(t)) \subset X \to X$ is \emph{adiabatic w.r.t.~a family $P$ of bounded projections $P(t)$ in $X$} if and only if $U(t,s)$ for every $(s,t) \in \Delta$ intertwines $P(s)$ with $P(t)$, more precisely: 
\begin{align*} 
P(t) U(t,s) = U(t,s) P(s) 
\end{align*}
for every $(s,t) \in \Delta$.
Since the pioneering work~\cite{Kato50} of Kato, the basic strategy in proving the convergence~\eqref{eq: aussage des adsatzes} has been to show that
\begin{align}
U_{\eps}(t)-V_{\eps}(t) \longrightarrow 0 \quad (\eps \searrow 0)
\end{align}
for every $t \in I$, where the $V_{\eps}$ are suitable comparison evolution systems that are adiabatic 
w.r.t.~the family $P$ of projections $P(t)$ 
related to the data $A$, $\sigma$. 
A simple 
way of obtaining 
adiabatic evolutions w.r.t.~some given family $P$ 
(independently observed by Kato in~\cite{Kato50} and Daleckii--Krein in~\cite{DaleckiiKrein50}) is described in the following important proposition.

\begin{prop}[Kato, Daleckii--Krein] \label{prop: intertwining relation}
Suppose $A(t): D(A(t)) \subset X \to X$ for every $t \in I$ is a densely defined closed linear map and $P(t)$ a bounded projection in $X$ such that $P(t)A(t) \subset A(t)P(t)$ for every $t \in I$ and $t \mapsto P(t)$ is SOT-continuously differentiable. If the evolution system $V_{\eps}$ for $\frac 1 \eps A + [P',P]$ exists on $D(A(t))$ for every $\eps \in (0,\infty)$, then $V_{\eps}$ is adiabatic w.r.t.~$P$ for every $\eps \in (0,\infty)$.
\end{prop}

\begin{proof}
Choose an arbitrary $(s,t) \in \Delta$ with $s \ne t$. 
It then follows by Lemma~\ref{lm: prodregel rechtsseit db} and Lemma~\ref{lm: zeitentw rechtsseit db} that, for every $x \in D(A(s))$, the map 
\begin{align*}
[s,t] \ni \tau \mapsto V_{\eps}(t,\tau) P(\tau) V_{\eps}(\tau,s)x
\end{align*}
is continuous and right differentiable. Since $P(\tau)$ commutes with $A(\tau)$ 
and 
\begin{align} \label{eq: PP'P=0}
P(\tau)P'(\tau)P(\tau) = 0 
\end{align}
for every $\tau \in I$ (which follows by applying $P$ from the left and the right to $P' = P'P+PP'$), it further follows that the right derivative of this map is identically $0$ and so (by Lemma~\ref{lm: rechtsseit db und W^{1,infty}}) this map is constant. 
In particular,
\begin{align*}
P(t)V_{\eps}(t,s)x - V_{\eps}(t,s)P(s)x = V_{\eps}(t,\tau) P(\tau) V_{\eps}(\tau,s)x \big|_{\tau=s}^{\tau=t} = 0,
\end{align*}
as desired. 
\end{proof}

We now briefly discuss two situations where the conclusion of the adiabatic theorem is already trivially true.

\begin{prop} \label{prop: triv adsatz}
Suppose $A(t): D(A(t)) \subset X \to X$ for every $t \in I$ is a densely defined closed linear map and $P(t)$ is a bounded projection in $X$ such that the evolution system $U_{\eps}$ exists on $D(A(t))$ for every $\eps \in (0,\infty)$ and such that $P(t)A(t) \subset A(t)P(t)$ 
for every $t \in I$ and $t \mapsto P(t)$ is SOT-continuously differentiable.
\begin{itemize}
\item[(i)] If $P' = 0$, then $U_{\eps}$ is adiabatic w.r.t.~$P$ for every $\eps \in (0,\infty)$ (in particular, the convergence~\eqref{eq: aussage des adsatzes} holds trivially), 
and the reverse implication is also true.

\item[(ii)] If there are $\gamma \in (0,\infty)$ and $M \in [1,\infty)$ such that for all $(s,t) \in \Delta$ and $\eps \in (0,\infty)$
\begin{align}  \label{eq: zweiter triv adsatz}
\norm{U_{\eps}(t,s)} \le M e^{-\frac \gamma \eps (t-s)},
\end{align}
then $\sup_{t \in I} \norm{U_{\eps}(t)-V_{\eps}(t)} = O(\eps)$ as $\eps \searrow 0$,
whenever the evolution system $V_{\eps}$ for $\frac 1 \eps A + [P',P]$ 
exists on $D(A(t))$ for every $\eps \in (0,\infty)$.
\end{itemize}
\end{prop}

\begin{proof}
(i) See, for instance, Section~IV.3.2 of~\cite{Krein71} for the reverse implication (differentiate the adiabaticity relation with respect to the variable $s$) -- the other implication is obvious from Proposition~\ref{prop: intertwining relation}.
\smallskip

(ii) Since for $x \in D(A(0))$ one has (by Lemma~\ref{lm: zeitentw rechtsseit db}, Lemma~\ref{lm: prodregel rechtsseit db}, and Lemma~\ref{lm: rechtsseit db und W^{1,infty}}) 
\begin{align*}
V_{\eps}(t)x - U_{\eps}(t)x = U_{\eps}(t,s)V_{\eps}(s)x \big|_{s=0}^{s=t} = \int_0^t U_{\eps}(t,s) [P'(s),P(s)] V_{\eps}(s)x \, ds
\end{align*}
for every $t \in I$ and $\eps \in (0,\infty)$, it follows 
with the help of Proposition~\ref{prop: störreihe für gestörte zeitentw} that
\begin{align*}
\norm{ U_{\eps}(t)-V_{\eps}(t) } \le M^2 c \, e^{Mc} \,\, t \, e^{-\frac \gamma \eps  t} 
\end{align*}
for all $t \in I$ and $\eps \in (0,\infty)$, where $c$ denotes an upper bound of $s \mapsto \norm{[P'(s),P(s)]}$. And from this the desired conclusion is obvious.
\end{proof}

Combining Proposition~\ref{prop: triv adsatz}~(ii) with Example~\ref{ex: (M,0)-stabilität wesentl, mit sl} one sees that adiabatic theory is interesting only if the evolution systems for $\frac 1 \eps A$ are \emph{only just} bounded w.r.t.~$\eps \in (0,\infty)$: if even the evolution for $\frac 1 \eps (A+\gamma)$ is bounded in $\eps \in (0,\infty)$ for some $\gamma > 0$, then adiabatic theory is trivial for $A$ (by Proposition~\ref{prop: triv adsatz}~(ii)), and if only the evolution for $T(A-\gamma)$ is bounded in $\eps \in (0,\infty)$ for some $\gamma > 0$, then adiabatic theory is generally impossible for $A$ 
(by Example~\ref{ex: (M,0)-stabilität wesentl, mit sl}).

\subsection{Standard examples}

We will complement 
the adiabatic theorems of this paper by 
examples in order to demonstrate, on the one hand, that the presented theorems are strictly more general than the previously known adiabatic theorems (positive examples) and that, on the other hand, some selected 
hypotheses 
of our theorems cannot be dispensed with (negative examples). We have made sure that in all positive examples the conclusion of the respective adiabatic theorem is not already trivially fulfilled in the sense that it does not already follow from the trivial adiabatic theorem presented above. (See Example~\ref{ex: A nur W^{1,infty}-reg und nur (M,0)-stabil} where this is once -- and for all -- explained in detail.) 
All examples will be of the following simple 
standard structure: 
\begin{itemize}
\item $X = \ell^p(I_d)$ for some $p \in [1,\infty)$ and $d \in \N \cup \{ \infty \}$ (where $I_d := \{1, \dots, d\}$ for $d \in \N$ and $I_{\infty} := \N$)
or $X = L^p(X_0)$ for some $p \in [1,\infty)$ and some measure space $(X_0, \mathcal{A}, \mu)$ 
or $X$ is a 
product of some of the aforementioned spaces (endowed with the sum norm)
\item $A(t) = R(t)^{-1} A_0(t) R(t)$, where $A_0(t): D \subset X \to X$ is a semigroup generator on $X$ with $t$-independent dense domain $D$ 
(chosen 
equal or unequal to $X$ depending on whether we are in the case of time-independent or time-dependent domains), $A_0$ satisfies Condition~\ref{cond: reg 1}, 
and $R(t) := e^{C t}$ for some bounded operator $C$. 
\end{itemize}
Condition~\ref{cond: reg 1} with $\omega = 0$ for $A_0$ ensures (by Lemma~\ref{lm: (M,w)-stabilität und ähnl.trf.} and Corollary~\ref{cor: Kato85}) that the hypotheses on $A$ of the adiabatic theorems of Sections~3 to~5 are fulfilled.
In some of our examples we will use the right or left shift operator $S_+$ and $S_-$ on $\ell^p(I_{\infty})$ defined by 
\begin{align*}
S_+(x_1, x_2, \dots) := (0, x_1, x_2, \dots) \quad \text{and} \quad S_-(x_1,x_2,x_3, \dots) := (x_2,x_3, \dots).
\end{align*}
Since $\norm{S_{\pm}} \le 1$, it follows from the theorem of Hille--Yosida that $e^{i \vartheta} S_+ - 1$ and $e^{i \vartheta} S_- - 1$ generate contraction semigroups on $\ell^p(I_{\infty})$ for $p \in [1,\infty)$ and $\vartheta \in \R$ (use a Neumann series expansion!). 
It is well-known (Example~V.4.1 and V.4.2 of~\cite{TaylorLay80}) that $\sigma(S_{\pm}) = \overline{U}_1(0)$ for all $p \in [1,\infty)$, 
the fine structure of $\sigma(S_+)$ being given by
\begin{gather*}
\sigma_p(S_+) = \emptyset, \quad \sigma_c(S_+) = \emptyset, \quad \sigma_r(S_+) = \overline{U}_1(0) \quad (p = 1)  \\
\sigma_p(S_+) = \emptyset, \quad \sigma_c(S_+) = \partial U_1(0), \quad \sigma_r(S_+) = U_1(0) \quad (p \in (1,\infty))
\end{gather*}
and the fine structure of $\sigma(S_-)$ being given by
\begin{align*}
\sigma_p(S_-) = U_1(0), \quad \sigma_c(S_-) = \partial U_1(0), \quad \sigma_r(S_-) = \emptyset \quad (p \in [1,\infty)).
\end{align*}
Additionally, we will sometimes use multiplication operators $M_f$ on $L^p(X_0)$ ($p \in [1,\infty)$) for some measurable function $f: X_0 \to \C$ and some $\sigma$-finite measure space $(X_0, \mathcal{A}, \mu)$ in which case, as is well-known (Proposition~I.4.10 of~\cite{EngelNagel}),
one has 
\begin{align*}
\sigma(M_f) = \operatorname{ess-rg} f := \{ z \in \C: \mu \big( f^{-1}(U_{\eps}(z)) \big) \ne 0 \text{ for all } \eps > 0 \big\}
\end{align*}
and, in particular 
(take $\mu$ to be the counting measure on $X_0 := I_d$),
\begin{align*}
\sigma \big( \operatorname{diag}( (\lambda_n)_{n \in I_d} ) \big) = \sigma \big( M_{ (\lambda_n)_{n \in I_d} } \big) = \overline{ \{ \lambda_n: n \in I_d \}  }.
\end{align*}

In quite some examples, we will work with families $A$ of operators $A(t)$ in $X := \ell^p(I_d)$ whose spectra $\sigma(A(t))$ 
are singletons and whose nilpotent parts depend on $t$ in the simplest possible way, 
namely via a scalar factor.

\begin{cond} \label{cond: baustein mit nicht-halbeinfachem ew}
$N \ne 0$ is a nilpotent operator in $X := \ell^p(I_d)$ (with $p \in [1,\infty)$ and $d \in \N$), $\lambda(t) \in \C$ and $\alpha(t) \in [0,\infty)$ for all $t \in I$, and there is an $r_0 > 0$ such that $-\Re \lambda(t) = |\Re \lambda(t) | \ge r_0 \alpha(t)$ for all $t \in I$.
\end{cond}

As is shown in the next lemma, this condition characterizes $(M,0)$-stability of families~$A$ 
of the simple type 
just described.

\begin{lm} \label{lm: char (M,0)-stab für einfaches A}
Suppose that $N \ne 0$ is a nilpotent operator in $X := \ell^p(I_d)$ with $p \in [1,\infty)$ and $d \in \N$ and that $A(t) = \lambda(t) + \alpha(t) N$ for every $t \in I$, where $\lambda(t) \in \C$ and $\alpha(t) \in [0,\infty)$. Then $A$ is $(M,0)$-stable for some $M \in [1,\infty)$ if and only if 
Condition~\ref{cond: baustein mit nicht-halbeinfachem ew} is satisfied.
\end{lm}

\begin{proof}
Suppose first that $A$ is $(M, 0)$-stable for some $M \in [1,\infty)$ and assume that $N = \operatorname{diag}(J_1, \dots, J_m)$ is in Jordan normal form with (decreasingly ordered) Jordan block matrices $J_1, \dots, J_m$ 
(notice that this assumption, 
by virtue of Lemma~\ref{lm: (M,w)-stabilität und ähnl.trf.}, does not restrict generality). 
We then show that $-\Re \lambda(t) = | \Re \lambda(t) | \ge \frac{1}{4 M} \, \alpha(t)$ for every $t \in I$. 
It is clear by the $(M,0)$-stability of $A$ that $\lambda(t) \in \sigma(A(t)) \subset \{ \Re z \le 0 \}$ for every $t \in I$ and that the family $\tilde{A}$ with $\tilde{A}(t) := \Re \lambda(t) + \alpha(t) N$ is $(M, 0)$-stable as well. 
If $\alpha(t) = 0$, then the desired inequality 
is trivial. If $\alpha(t) \ne 0$, then $\Re \lambda(t) < 0$ by the $(M,0)$-stability of $A$ and therefore we get 
-- computing $(\lambda - \tilde{A}(t))^{-1} e_2  
= (\frac{\alpha(t)}{(\lambda-\Re \lambda(t))^2}, \frac{1}{\lambda-\Re \lambda(t)}, 0, 0, \dots)$ for $\lambda \in (0,\infty)$, setting $\lambda := |\Re \lambda(t)|$, and using the $(M,0)$-stability of $\tilde{A}$ -- 
that
\begin{align*}
\frac{\alpha(t)}{4 \, | \Re \lambda(t) |}    \le  \norm{ |\Re \lambda(t)| \, \big( |\Re \lambda(t)| - \tilde{A}(t) \big)^{-1} \, e_2 } \le M,
\end{align*}
as desired.
Suppose conversely that there is an $r_0 > 0$ such that $-\Re \lambda(t) = |\Re \lambda(t)| \ge r_0 \alpha(t)$ for every $t \in I$. Then there is an $M = M_{r_0} \in [1,\infty)$ such that $\norm{ e^{N s} } \le M e^{r_0 \, s}$ for all $s \in [0,\infty)$ 
and thus
\begin{align*}
\norm{ e^{A(t_n)s_n} \dotsb e^{A(t_1)s_1} } 
= e^{\Re \lambda(t_n)s_n} \dotsb e^{\Re \lambda(t_1)s_1} \, \norm{ e^{N ( \alpha(t_n)s_n + \dotsb + \alpha(t_1)s_1 )} }
\le M
\end{align*}
for all $s_1, \dots, s_n \in [0,\infty)$ and all $t_1, \dots, t_n \in I$ satisfying $t_1 \le \dotsb \le t_n$ (with arbitrary~$n \in \N$), as desired.
\end{proof}


It should be noticed that Condition~\ref{cond: baustein mit nicht-halbeinfachem ew} does not already guarantee $(1,0)$-stability, however. Indeed, if for instance
\begin{align*}
A(t) := -\frac{t}{3} + t^2 N \quad \text{with} \quad 
N:= \begin{pmatrix} 0  & 1  &   &  \\
                       & 0  & 1 &  \\
                       &    & \ddots & \ddots  \\
                       &    &        & 0
    \end{pmatrix}  \text{\,\, in } X := \ell^p(I_d)
\end{align*}
($p \in [1,\infty)$ and $2 \le d \in \N$), then $A$ is $(M,0)$-stable for some $M \in [1,\infty)$ by the above lemma, 
but not $(1,0)$-stable, because $A(1) = -\frac{1}{3} + N$ is not dissipative in $\ell^p(I_d)$ 
and hence (by the theorem of Lumer--Phillips) does not generate a contraction semigroup. 
\smallskip

At some point (Example~\ref{ex: reg an P wesentl, ohne sl}) the following simple lemma will be needed which, in essence, is the reason why adiabatic theory for multiplication operators $A(t) = M_{f_t}$ is typically uninteresting. See~\cite{dipl} (Lemma~2.11 and the remark following it) 
for the 
proof.

\begin{lm}  \label{lm: wenn P stet db, dann schon konst}
Suppose that $P(t)$ for every $t \in I$ is a bounded projection in $X := L^p(X_0)$ (where $(X_0,\mathcal{A},\mu)$ is a measure space and $p \in [1,\infty)$) and that $P(t) = M_{\chi_{E_t}}$ for almost every $t \in I$, where $E_t \in \mathcal{A}$. If $t \mapsto P(t)$ is SOT-continuously differentiable, then $t \mapsto P(t)$ is already constant.
\end{lm}

\section{Adiabatic theorems with spectral gap condition for time-independent domains} \label{sect: adsätze mit sl}

After having provided the most important preliminaries in Section~2, we now prove an adiabatic theorem with uniform spectral gap condition (Section~3.1) and an adiabatic theorem with non-uniform spectral gap condition (Section~3.2) for general operators $A(t)$ with time-independent domains. In these theorems the considered spectral subsets $\sigma(t)$ are only assumed to be compact so that, even if they are singletons, they need not consist of eigenvalues: they are allowed to be singletons consisting of essential singularities of the resolvent. In~\cite{AbouSalem07}, \cite{AvronGraf11}, \cite{Joye07} the case of poles is 
treated. 

\subsection{An adiabatic theorem with uniform spectral gap condition}  \label{sect: adsatz mit glm sl}

We begin by proving an adiabatic theorem with uniform spectral gap condition by extending Abou Salem's 
proof from~\cite{AbouSalem07}, which rests upon solving a suitable commutator equation. 

\begin{thm} \label{thm: handl adsatz mit glm sl}
Suppose $A(t): D \subset X \to X$ for every $t \in I$ is a linear map such that Condition~\ref{cond: reg 1}  is satisfied with $\omega = 0$. 
Suppose further that $\sigma(t)$ for every $t \in I$ is a compact 
subset of $\sigma(A(t))$, 
that $\sigma(\,.\,)$ at no point falls into $\sigma(A(\,.\,))\setminus \sigma(\,.\,)$, and that $t \mapsto \sigma(t)$ is continuous. And finally, for every $t \in I$, let $P(t)$ be the 
projection associated with $A(t)$ and $\sigma(t)$ and suppose that $I \ni t \mapsto P(t)$ is in $W^{2,\infty}_*(I,L(X))$. Then 
\begin{align*}
\sup_{t \in I} \norm{ U_{\eps}(t) - V_{\eps}(t) } = O(\eps) \quad (\eps \searrow 0),
\end{align*}
where $V_{\eps}$ is the evolution system for $\frac 1 \eps A + [P',P]$.
\end{thm}

\begin{proof}
Since $\sigma(\,.\,)$ is uniformly isolated in $\sigma(A(\,.\,)) \setminus \sigma(\,.\,)$ and $t \mapsto \sigma(t)$ is continuous, 
there is, for every $t_0 \in I$, a non-trivial closed interval $J_{t_0} \subset I$ containing $t_0$ and a cycle $\gamma_{t_0}$ in $\rho(A(t_0))$ such that $\rg \gamma_{t_0} \subset \rho(A(t))$ and 
\begin{align*}
\operatorname{n}(\gamma_{t_0}, \sigma(t)) = 1 \quad \text{and} \quad \operatorname{n}(\gamma_{t_0}, \sigma(A(t)) \setminus \sigma(t)) = 0 
\end{align*}
for all $t \in J_{t_0}$.
We can now define 
\begin{align*}
 B(t)x := \frac{1}{2 \pi i} \int_{\gamma_{t_0}} (z-A(t))^{-1} P'(t) (z-A(t))^{-1} x \, dz
\end{align*}
for all $t \in J_{t_0}$, $t_0 \in I$ and $x \in X$. Since $\rho(A(t)) \ni z \mapsto (z-A(t))^{-1} P'(t) (z-A(t))^{-1} x$ is a holomorphic $X$-valued map (for all $x \in X$) and since the cycles $\gamma_{t_0}$ and $\gamma_{t_0'}$ are 
homologous in $\rho(A(t))$ whenever $t$ lies both in $J_{t_0}$ and in $J_{t_0'}$, the path integral exists in $X$ 
and does not depend on the special choice of $t_0 \in I$ with the property that $t \in J_{t_0}$. In other words, $t \mapsto B(t)$ is 
well-defined on $I$.
\smallskip

As a first preparatory step, we easily infer from the closedness of $A(t)$ that $B(t)X \subset D(A(t)) = D = Y$ and that
\begin{align} \label{eq: commutator equation}
 B(t) A(t) - A(t) B(t) \subset [P'(t),P(t)]
\end{align}
for all $t \in I$, which commutator equation will be essential in the main part of the proof.
As a second preparatory step, we show that $t \mapsto B(t)$ is in $W^{1,\infty}_*(I,L(X,Y))$, which is not very surprising (albeit a bit technical). It suffices to show that $J_{t_0} \ni t \mapsto B(t)$ is in $W^{1,\infty}_*(J_{t_0},L(X,Y))$ for every $t_0 \in I$. We therefore fix $t_0 \in I$. 
Since $\rho(A(t)) \ni z \mapsto (z-A(t))^{-1}$ is continuous w.r.t.~the norm of $L(X,Y)$ for every $t \in J_{t_0}$, 
we see that $B(t)$ is in $L(X,Y)$ for every $t \in J_{t_0}$. 
We also see, by virtue of Lemma~\ref{lm: prod- und inversenregel}, that for every $z \in \rg \gamma_{t_0}$ the map $t \mapsto (z-A(t))^{-1} P'(t) (z-A(t))^{-1}$ is in $W^{1,\infty}_*(J_{t_0},L(X,Y))$ and $t \mapsto C(t,z) = C_1(t,z) + C_2(t,z) + C_3(t,z)$ is a $W^{1,\infty}_*$-derivative of it, 
where
\begin{align} \label{eq: W^{1,infty}-abl des integranden}
C_1(t,z) &= (z-A(t))^{-1} A'(t) (z-A(t))^{-1} P'(t) (z-A(t))^{-1}, \notag \\
&C_2(t,z) = (z-A(t))^{-1} P''(t) (z-A(t))^{-1}, \\
C_3(t,z) &= (z-A(t))^{-1} P'(t) (z-A(t))^{-1} A'(t) (z-A(t))^{-1} \notag
\end{align}
and $A'$, $P''$ are arbitrary $W^{1,\infty}_*$-derivatives of $A$ and $P'$. 
Since $t \mapsto C(t,z)$ is strongly measurable for all $z \in \rg \gamma_{t_0}$, it follows that $t \mapsto \frac{1}{2 \pi i} \int_{\gamma_{t_0}} C(t, z) \, dz$ is strongly measurable as well (as the strong limit of Riemann sums), and since $J_{t_0} \times \rg \gamma_{t_0} \ni (t,z) \mapsto (z-A(t))^{-1}$ is continuous w.r.t.~the norm of $L(X,Y)$ and hence bounded, it follows by~\eqref{eq: W^{1,infty}-abl des integranden} that
\begin{align*} 
t \mapsto \Big\| \frac{1}{2 \pi i} \int_{\gamma_{t_0}} C(t, z) \, dz \Big\|_{X \to Y}
\end{align*}
is essentially bounded. So $t \mapsto \frac{1}{2 \pi i} \int_{\gamma_{t_0}} C(t, z) \, dz$ is in $W^{0,\infty}_*(J_{t_0}, L(X,Y))$ 
and 
one easily concludes that
\begin{align*}
 B(t)x = B(t_0)x + \int_{t_0}^t \frac{1}{2 \pi i} \int_{\gamma_{t_0}} C(\tau, z)x \, dz \, d\tau
\end{align*}
for all $t \in J_{t_0}$ and $x \in X$, as desired.
\smallskip

After these preparations we can now turn to the main part of the proof. 
We fix $x \in D$ and let $V_{\eps}$ denote the evolution system for $\frac 1 \eps A + [P',P]$ (which really exists due to Theorem~\ref{thm: Kato85}). Then $s \mapsto U_{\eps}(t,s) V_{\eps}(s)x$ is in $W^{1,\infty}([0,t],X)$ (by Lemma~\ref{lm: prodregel rechtsseit db} and Lemma~\ref{lm: rechtsseit db und W^{1,infty}}) and we get, exploiting the commutator equation~\eqref{eq: commutator equation} for $A$ and $B$, that
\begin{align*}
V_{\eps}(t)x - U_{\eps}(t)x &= U_{\eps}(t,s)V_{\eps}(s)x \big|_{s=0}^{s=t} = \int_0^t U_{\eps}(t,s) [P'(s),P(s)] V_{\eps}(s)x \, ds \\
&= \int_0^t U_{\eps}(t,s) \bigl( B(s)A(s) - A(s)B(s) \bigr) V_{\eps}(s)x \, ds
\end{align*}
for all $t \in I$. Since for every $t \in I$ the maps $s \mapsto V_{\eps}(s) \big|_{Y}$ and $s \mapsto U_{\eps}(t,s) \big|_{Y}$ are continuously differentiable on $[0,t]$ w.r.t.~SOT of $L(Y,X)$ (Lemma~\ref{lm: zeitentw rechtsseit db}) and hence belong to $W^{1,\infty}_*([0,t], L(Y,X))$, 
and since $s \mapsto B(s)$ belongs to $W^{1,\infty}_*([0,t], L(X,Y))$, 
we can further conclude, using Lemma~\ref{lm: prod- und inversenregel}, that $s \mapsto U_{\eps}(t,s) B(s) V_{\eps}(s)x$ is in $W^{1,\infty}([0,t],X)$, so that 
\begin{align*}
&V_{\eps}(t)x - U_{\eps}(t)x = \eps \int_0^t U_{\eps}(t,s) \Bigl( - \, \frac 1 \eps A(s)B(s) + B(s) \frac 1 \eps A(s) \Bigr) V_{\eps}(s)x \, ds \\
&\quad = \eps \, U_{\eps}(t,s) B(s) V_{\eps}(s)x \big|_{s=0}^{s=t} 
- \eps \int_0^t U_{\eps}(t,s) \bigl( B'(s) + B(s) [P'(s),P(s)] \bigr) V_{\eps}(s)x \, ds
\end{align*}
for all $t \in I$ and $\eps \in (0,\infty)$, where $B'$ denotes an arbitrary $W^{1,\infty}_*$-derivative of $B$. And from this, the conclusion of the theorem is obvious. 
\end{proof}

%

\subsection{An adiabatic theorem with non-uniform spectral gap condition}   \label{sect: adsatz mit nichtglm sl}

We continue by proving an adiabatic theorem with non-uniform spectral gap condition where $\sigma(\,.\,)$ falls into $\sigma(A(\,.\,)) \setminus \sigma(\,.\,)$ at countably many points that, 
in turn, accumulate at only finitely many points. 
We do so by extending Kato's 
proof from~\cite{Kato50} where finitely many eigenvalue crossings for skew self-adjoint $A(t)$ are treated.

\begin{thm} \label{thm: handl adsatz mit nichtglm sl}
Suppose $A(t): D \subset X \to X$ for every $t \in I$ is a linear map such that Condition~\ref{cond: reg 1} is satisfied with $\omega = 0$. 
Suppose further that $\sigma(t)$ for every $t \in I$ is a compact 
subset of $\sigma(A(t))$, that $\sigma(\,.\,)$ at countably many points accumulating at only finitely many points 
falls into $\sigma(A(\,.\,))\setminus \sigma(\,.\,)$, and that $I \setminus N \ni t \mapsto \sigma(t)$ is continuous, where $N$ denotes the set of 
those points where $\sigma(\,.\,)$ falls into $\sigma(A(\,.\,))\setminus \sigma(\,.\,)$. And finally, for every $t \in I \setminus N$, let $P(t)$ be the 
projection associated with $A(t)$ and $\sigma(t)$ 
and suppose that $I \setminus N \ni t \mapsto P(t)$ extends to a map (again denoted by $P$) in $W^{2,\infty}_*(I,L(X))$. Then 
\begin{align*}
\sup_{t \in I} \norm{ U_{\eps}(t) - V_{\eps}(t) } \longrightarrow 0 \quad (\eps \searrow 0),
\end{align*}
where $V_{\eps}$ is the evolution system for $\frac 1 \eps A + [P',P]$.
\end{thm}

\begin{proof}
We first prove the assertion in the case where $\sigma(\,.\,)$ at only finitely many points $t_1, \dots, t_m$ (ordered in an increasing way) 
falls into $\sigma(A(\,.\,))\setminus \sigma(\,.\,)$. So let $\eta >0$. We partition the interval $I$ as follows: 
\begin{align*}  
 I = I_{0 \, \delta} \cup J_{1 \, \delta} \cup I_{1 \, \delta} \cup \dots \cup J_{m \, \delta} \cup I_{m \, \delta},
\end{align*}
where $J_{i \, \delta}$ for $i \in \{1,\dots, m\}$ is a relatively open subinterval of $I$ containing $t_i$ of length less than $\delta$ (which will be chosen in a minute) and where $I_{0 \, \delta}$, \dots, $I_{m \, \delta}$ are the closed subintervals of $I$ lying between the subintervals
$J_{1 \, \delta}$, \dots, $J_{m \, \delta}$. In the following, we set $t_{i \, \delta}^- := \inf I_{i \, \delta}$ and $t_{i \, \delta}^+ := \sup I_{i \, \delta}$ for $i \in \{0, \dots, m\}$, and we choose $c$ so large that $\norm{P(s)}$, $\norm{P'(s)}$ and
$\norm{ [P'(s),P(s)] } \le c$ for all $s \in I$.
Since 
\begin{align*}
 \norm{ V_{\eps}(t,t_{i-1 \, \delta}^+)x - U_{\eps}(t, t_{i-1 \, \delta}^+)x } &= \norm{  \int_{t_{i-1 \, \delta}^+}^t U_{\eps}(t, s) [P'(s), P(s)] V_{\eps}(s, t_{i-1 \, \delta}^+)x \, ds  } \\
 &\le M c M e^{M c} \, \delta \norm{x}
\end{align*}
for every $t \in J_{i \, \delta}$, $x \in D$ and $\eps \in (0,\infty)$, we can achieve -- by choosing $\delta$ small enough -- that
\begin{align}  \label{eq: adsatz mit nichtglm sl 1}
\sup_{t \in J_{i \, \delta}} \norm{ V_{\eps}(t,t_{i-1 \, \delta}^+) - U_{\eps}(t, t_{i-1 \, \delta}^+) } < \frac{\eta}{ \big( 4 M^2 e^{2 Mc} \big)^m }
\end{align}
for every $\eps \in (0,\infty)$ and $i \in \{1, \dots, m\}$.
And since $\sigma(\,.\,) \big|_{I_{i \, \delta}}$ at no point falls into $\big(  \sigma(A(\,.\,))\setminus \sigma(\,.\,)  \big) \big|_{I_{i \, \delta}}$, we conclude from the above adiabatic theorem with uniform spectral gap condition 
(applied to the restricted data $A|_{I_{i \, \delta}}$, $\sigma|_{I_{i \, \delta}}$, $P|_{I_{i \, \delta}}$) 
that there is an $\eps_{\delta} \in (0, \infty)$ such that
\begin{align}   \label{eq: adsatz mit nichtglm sl 2}
 \sup_{t \in I_{i \, \delta}}  \norm{ V_{\eps}(t,t_{i \, \delta}^-) - U_{\eps}(t, t_{i \, \delta}^-) } < \frac{\eta}{ \big( 4 M^2 e^{2 Mc} \big)^m }
\end{align}
for every $\eps \in (0, \eps_{\delta})$ and $i \in \{0, \dots, m\}$.
Combining the estimates~\eqref{eq: adsatz mit nichtglm sl 1} and~\eqref{eq: adsatz mit nichtglm sl 2} and using the product property from the definition of evolution systems, we readily conclude for every $i \in \{1, \dots, m\}$ that
\begin{align*}
 \norm{V_{\eps}(t)-U_{\eps}(t)} < \frac{\eta}{ \big( 4 M^2 e^{2 Mc} \big)^{m-i} } \le \eta
\end{align*}
for all $t \in I_{i-1 \, \delta} \cup J_{i \, \delta} \cup I_{i \, \delta}$ and $\eps \in  (0,\eps_{\delta})$, and the desired conclusion follows. 
\smallskip

We now prove the assertion in the case where $\sigma(\,.\,)$ at infinitely many points accumulating at only finitely many points $t_1, \dots, t_m$ (ordered in an increasing way) falls into $\sigma(A(\,.\,))\setminus \sigma(\,.\,)$. In order to do so, we partition $I$ and choose $\delta$ as we did above.
We then obtain the estimate~\eqref{eq: adsatz mit nichtglm sl 1} as above and the estimate~\eqref{eq: adsatz mit nichtglm sl 2} 
by realizing that  $\sigma(\,.\,) \big|_{I_{i \, \delta}}$ at only finitely many points falls into $\big(  \sigma(A(\,.\,))\setminus \sigma(\,.\,)  \big) \big|_{I_{i \, \delta}}$ (so that the case just proved 
can be applied). And from these estimates 
the conclusion follows as above.
\end{proof}


It should be noticed that the hypotheses of the above adiabatic theorem allow the spectral subsets $\sigma(t)$ to be non-isolated in $\sigma(A(t))$ for all points $t \in N$, 
that is (by our definition of spectral gaps in Section~2.1), they also 
allow for some situations without spectral gap. 
It should also be noticed that, in the situation of the above theorem, one has $P(t) A(t) \subset A(t) P(t)$ for every $t \in I$ (although a priori this is clear only for $t \in I \setminus N$), which follows by a continuity argument. 
(Indeed, 
if $t_0 \in I$ then it can be approximated by a sequence $(t_n)$ in $I \setminus N$. Since $t \mapsto (A(t_0)-1)(A(t)-1)^{-1}$ is NOT-continuous (by the $W^{1,\infty}_*$-regularity of $t \mapsto A(t)$ and Lemma~\ref{lm: prod- und inversenregel}), we see that for any $x \in D$ 
\begin{align} \label{eq: P(t) vertauscht mit A(t) für alle t}
A(t_0)P(t_n)x &= (A(t_0)-1)(A(t_n)-1)^{-1} \, P(t_n) (A(t_n)-1)x + P(t_n)x \nonumber \\
&\longrightarrow P(t_0)A(t_0)x \quad (n \to \infty)
\end{align}
and therefore $P(t_0)A(t_0) \subset A(t_0)P(t_0)$ by the closedness of $A(t_0)$.)
In particular, the evolution $V_T$ appearing in the above theorem really is adiabatic w.r.t.~to $P$ by Proposition~\ref{prop: intertwining relation}, as it should be. 



\subsection{Some remarks and examples}  \label{sect: bsp, adsätze mit sl}

We begin with four remarks concerning the adiabatic theorems with uniform and non-uniform spectral gap condition alike.
\smallskip

1. In the special situation where $\sigma(t) = \{ \lambda(t) \}$ and $\lambda(t)$ 
is a pole of the resolvent map $(\,.\,-A(t))^{-1}$ of order at most $m_0 \in \N$ for all $t \in I$, the operators $B(t)$ -- used in the proof of the adiabatic theorems with spectral gap condition above to solve the commutator equation~\eqref{eq: commutator equation} -- can be brought to a form, 
namely~\eqref{eq: wegweiser}, 
which points the way to the solution of 
an appropriate (approximate) commutator equation 
in 
the adiabatic theorems without spectral gap condition below. 
%
%
Since $PP'P$, $\overline{P}P'\overline{P} = 0$ by~\eqref{eq: PP'P=0} (where $\overline{P} := 1-P$) and 
\begin{align*}
(z-A(t))^{-1}P(t) = \frac{1}{z-\lambda(t)} \Big( 1- \frac{A(t)-\lambda(t)}{z-\lambda(t)} \Big)^{-1} P(t) 
= \sum_{k=0}^{m_0-1} \frac{ (A(t)-\lambda(t))^k P(t)}{ (z-\lambda(t))^{k+1} } 
\end{align*}
for every $z \in \rho(A(t))$ by Theorem~5.8-A of~\cite{Taylor58}, we see that  
\begin{gather*}
B(t) = \sum_{k = 0}^{m_0-1} \frac{1}{2 \pi i} \int_{\gamma_t} \frac{\overline{R}(t,z)}{(z-\lambda(t))^{k+1}}  \, dz \,\, P'(t) (A(t)-\lambda(t))^k P(t) \qquad \qquad \\
\qquad \qquad \qquad + \sum_{k = 0}^{m_0-1} (A(t)-\lambda(t))^k P(t) P'(t) \,\, \frac{1}{2 \pi i} \int_{\gamma_t} \frac{\overline{R}(t,z)}{(z-\lambda(t))^{k+1}}  \, dz,
\end{gather*}
and since the reduced resolvent map $z \mapsto \overline{R}(t,z) := (z-A(t)|_{\overline{P}(t)D(A(t))})^{-1} \overline{P}(t)$ is holomorphic on $\rho(A(t)) \cup \{ \lambda(t) \}$, we further see -- using Cauchy's theorem -- that
\begin{align} \label{eq: wegweiser}
&B(t) = \sum_{k = 0}^{m_0-1} \overline{R}(t,\lambda(t))^{k+1} P'(t) (A(t)-\lambda(t))^k P(t)  \notag \\
&\qquad \qquad \qquad \qquad \qquad \quad + \sum_{k = 0}^{m_0-1} (\lambda(t)-A(t))^k P(t) P'(t) \overline{R}(t,\lambda(t))^{k+1}.
\end{align}

2.  In the even more special situation where $\sigma(t) = \{ \lambda(t) \} \subset i \R$ and $\lambda(t)$ is a pole of the resolvent map $(\, . \, - A(t))^{-1}$, the hypotheses of the above adiabatic theorem with uniform spectral gap condition become essentially -- apart from regularity conditions -- equivalent to the hypotheses of the respective adiabatic theorem (Theorem~9) of~\cite{AvronGraf11}, 
and a similar equivalence holds true for the above adiabatic theorem with non-uniform spectral gap condition. 
Indeed, if $\sigma(t)$ for every $t \in I$ is a singleton consisting of a pole $\lambda(t)$ on the imaginary axis, then the order $m(t)$ of nilpotence of $A(t)|_{P(t)D}-\lambda(t)$ 
must be equal to~$1$, since otherwise the relation
\begin{align}  \label{eq: ascent = 1 für erz beschr halbgr}
\delta \big( \lambda(t)+\delta -A(t) \big)^{-1} P(t) 
= \sum_{k=0}^{m(t)-1} \frac{  (A(t)-\lambda(t))^k }{\delta^k} \,  P(t)
\end{align}
would yield the contradiction that the right hand side of~\eqref{eq: ascent = 1 für erz beschr halbgr} explodes as $\delta \searrow 0$ whereas the left hand side of~\eqref{eq: ascent = 1 für erz beschr halbgr} remains bounded as $\delta \searrow 0$ (by virtue of the $(M,0)$-stability of $A$ and by $\lambda(t) \in i \R$). And therefore (by Theorem~5.8-A of~\cite{Taylor58}) $P(t)X = \ker(A(t)-\lambda(t))$ and $(1-P(t))X = \rg(A(t)-\lambda(t))$ 
as in~\cite{AvronGraf11}.
\smallskip

3. It is obvious from the proof of Theorem~\ref{thm: handl adsatz mit glm sl} that the assumption that $\sigma(\,.\,)$ at no point fall into $\sigma(A(\,.\,)) \setminus \sigma(\,.\,)$ and that $t \mapsto \sigma(t)$ be continuous can be replaced by the weaker -- but also less convenient -- 
requirement that for each $t_0 \in I$ there be a non-trivial closed interval $J_{t_0} \subset I$ containing $t_0$ and a cycle $\gamma_{t_0}$ such that $\rg \gamma_{t_0} \subset \rho(A(t))$ and 
\begin{align*}
\operatorname{n}(\gamma_{t_0}, \sigma(t)) = 1 \quad \text{and} \quad \operatorname{n}(\gamma_{t_0}, \sigma(A(t)) \setminus \sigma(t)) = 0 
\end{align*}
for all $t \in J_{t_0}$. 
It can be shown that this weaker requirement 
still entails the upper semicontinuity of $t \mapsto \sigma(t)$ and hence -- by Proposition~\ref{prop: zshg isoliert und glm isoliert} -- it still ensures that $\sigma(\,.\,)$ at no point falls into $\sigma(A(\,.\,)) \setminus \sigma(\,.\,)$ or, in other words, that the spectral gap is uniform. (See Corollary~5.4 of~\cite{dipl} for a proof.)
Consequently, if one adds to the thus weakened hypotheses the requirement that $t \mapsto \sigma(t)$ be lower semicontinuous (which -- by Proposition~5.6 of~\cite{dipl} -- is fulfilled if, for instance, $\sigma(t)$ is finite for every $t \in I$, and $0 \in \sigma(t)$ for all $t \in I$ or $0 \notin \sigma(t)$ for all $t \in I$), one arrives at an adiabatic theorem equivalent to the original one above (Theorem~\ref{thm: handl adsatz mit glm sl}). 
Similar remarks hold for the case of non-uniform spectral gap (Theorem~\ref{thm: handl adsatz mit nichtglm sl}).
\smallskip

4. We finally remark 
that the above adiabatic theorems -- along with the commutator equation method used in their proofs -- 
can be extended to several subsets $\sigma_1(t)$, \dots, $\sigma_r(t)$ of $\sigma(A(t))$: 
if $A$, $\sigma_j$, $P_j$ for every $j \in \{1, \dots, r\}$ satisfy the hypotheses of the above adiabatic theorem with uniform or non-uniform spectral gap and if $\sigma_j(\,.\,)$ and $\sigma_l(\,.\,)$ for all $j \ne l$ fall into each other at only countably many points accumulating at only finitely many points, then there exists an evolution system $V_{\eps}$, namely that for $\frac 1 \eps A + K$ with
\begin{align} \label{eq: K mehrere sigma_j}
K(t) := \frac 1 2 \sum_{j=1}^{r+1} [P_j'(t),P_j(t)] \quad \text{and} \quad 
P_{r+1}(t) := 1-P(t) := 1 - \sum_{j=1}^r P_j(t),
\end{align}
which on the one hand is simultaneously adiabatic w.r.t.~all the $P_j$ by~\cite{Kato50} and on the other hand well approximates the evolution system $U_{\eps}$ for $\frac 1 \eps A$ in the sense that
\begin{align*}
\sup_{t \in I} \norm{ U_{\eps}(t) - V_{\eps}(t) } \longrightarrow 0 \quad (\eps \searrow 0).
\end{align*}
In order to see this, one has only to observe 
that $B(t) := \frac 1 2 \sum_{j=1}^{r+1} B_j(t)$ with
\begin{gather}
B_j := \frac{1}{2 \pi i} \int_{\gamma_j} (z-A)^{-1} P_j' (z-A)^{-1} \, dz \quad (j \in \{1, \dots, r\}) \notag \\
B_{r+1} :=  \frac{1}{2 \pi i} \int_{\gamma} (z-A)^{-1} P' (z-A)^{-1} \, dz  \label{eq: B_{r+1}} \\ 
\text{with \,} \gamma := \gamma_1 + \dotsb + \gamma_r  \text{\, $(\gamma_j = \gamma_{j \, t}$ as in the proofs above) and } P := P_1 + \dotsb + P_r \notag
\end{gather}
solves the commutator equation $B(t)A(t) - A(t)B(t) \subset K(t)$ for all points $t$ where no crossing 
takes place (because $[P_{r+1}', P_{r+1}] = [P',P]$) and then to proceed as in the proofs of the adiabatic theorems above. See also~\cite{BrouderPanatiStoltz10}.
In the special case of skew self-adjoint operators $A(t)$ one can further refine 
the statement above: it is then possible to show -- by further adapting the commutator equation method -- 
that even the 
evolution system $\ol{V}_{\eps}$ for $\frac 1 \eps A + \ol{K}$ with
\begin{align*}
\ol{K}(t) := \frac 1 2 \Big( [(P_{r+1}^-)'(t),P_{r+1}^-(t)] + \sum_{j=1}^r [P_j'(t),P_j(t)] + [(P_{r+1}^+)'(t),P_{r+1}^+(t)] \Big)
\end{align*}
well approximates the evolution system $U_{\eps}$ for $\frac 1 \eps A$ -- notice that $\ol{V}_{\eps}$ is 
is not only adiabatic w.r.t.~$P_{r+1} = P_{r+1}^- + P_{r+1}^+$ 
but also w.r.t.~$P_{r+1}^-$ and $P_{r+1}^+$ separately, where $P_{r+1}^{\pm}(t)$ are the spectral projections of $A(t)$ corresponding to the parts $\sigma^{\pm}(t)$ of the spectrum 
which on $i \R$ are located 
below resp.~above all the compact parts $\sigma_1(t)$, \dots, $\sigma_r(t)$.
In order to see this, 
set
\begin{align*}
B_{r+1 \, n}^{\pm}(t) := \frac{1}{2 \pi i} \int_{\gamma_{n \, t}^{\pm}} (z-A(t))^{-1} (P_{r+1}^{\pm})'(t) (z-A(t))^{-1} \, dz
\end{align*}
where $\gamma_{n \, t}^{\pm}(\tau) := \pm \, \tau + c^{\pm}(t)$ for $\tau \in [-n,n]$ 
with points $c^{\pm}(t) \in i \R$ lying in the gap between $\sigma^{\pm}(t)$ and the rest of $\sigma(A(t))$ 
and depending continuously differentiably on $t$, 
and observe that (by the skew self-adjointness of $A(t)$)
\begin{gather*}
P_{r+1 \, n}^{\pm}(t)x := \frac{1}{2 \pi i} \int_{\gamma_{n \, t}^{\pm}} (z-A(t))^{-1}x \, dz \longrightarrow P_{r+1}^{\pm}(t)x - \frac 1 2 x \quad (n \to \infty) \\
\text{and} \\
\big\| B_{r+1 \, n}(t) \big\|, \norm{B_{r+1 \, n}'(t)} \le \int_{-\infty}^{\infty} \frac{c}{\dist\big( \gamma_{n \, t}^{\pm}(\tau) ,\sigma(A(t)) \big)^2 } \, d \tau \le C < \infty \quad (n \in \N, t \in I).
\end{gather*}  
A slightly less general 
general statement was first proven in~\cite{Nenciu80} by a different method than the commutator equation technique indicated above. 
\bigskip

We now move on to discuss some examples. In the first -- very simple -- example, $t \mapsto A(t)$ is only $W^{1,\infty}_*$-regular and only $(M,0)$-stable (without being SOT-continuously differentiable or $(1,0)$-stable), 
which means that this example lies outside the scope of the previously known adiabatic theorems.

\begin{ex} \label{ex: A nur W^{1,infty}-reg und nur (M,0)-stabil}
Suppose $A$, $\sigma$, $P$ with $A(t) = R(t)^{-1} A_0(t) R(t)$, $\sigma(t) = \{ \lambda(t) \}$, $P(t) = R(t)^{-1} P_0 R(t)$, and $R(t) = e^{C t}$ are given as follows in $X := \ell^p(I_2) \times \ell^p(I_1)$ (where $p \in [1,\infty)$):
\begin{align*}
A_0(t) := \begin{pmatrix} \lambda(t) + \alpha(t)N & 0 \\ 0 & \mu(t) \end{pmatrix} 
\quad \text{and} \quad
P_0 := \begin{pmatrix} 1 & 0 \\ 0 & 0 \end{pmatrix},
\end{align*} 
where $\lambda$, $\alpha$, $N$ satisfy Condition~\ref{cond: baustein mit nicht-halbeinfachem ew} 
and where $\mu(t) \in \{ \Re z \le 0 \}$ is such that $\lambda(\,.\,)$ falls into $\mu(\,.\,)$ at only countably many points accumulating at only finitely many points. Additionally, choose the family $\lambda + \alpha N$ to be not $(1,0)$-stable (which, by the discussion after Lemma~\ref{lm: char (M,0)-stab für einfaches A}, can easily be achieved), the functions $t \mapsto \lambda(t)$, $\alpha(t)$, $\mu(t)$ to be Lipschitz continuous without being continuously differentiable, and 
\begin{align*}
C := \begin{pmatrix} 0 & 0 & 0 \\ 0 & 0 & 1 \\ 0 & -1 & 0 \end{pmatrix}.
\end{align*}
Since $A_0$ by Lemma~\ref{lm: char (M,0)-stab für einfaches A} is $(M_0,0)$-stable for some $M_0 \in [1,\infty)$, the twisted family $A$ is $(M,0)$-stable for another $M \in [1,\infty)$ by Lemma~\ref{lm: (M,w)-stabilität und ähnl.trf.}. 
And since $P_0$ is obviously associated with $A_0(t)$ and $\sigma(t)$ for every $t \in I \setminus N$ 
(where $N$ denotes the set of those points in $I$ where $\lambda(\,.\,)$ falls into $\mu(\,.\,)$), the same is true for $P(t)$ and $A(t)$ instead of $P_0$, $A_0(t)$. So the hypotheses of the adiabatic theorem with non-uniform spectral gap condition (Theorem~\ref{thm: handl adsatz mit nichtglm sl}) are fulfilled and therefore 
\begin{align*}
(1-P(t)) U_{\eps}(t) P(0) \longrightarrow 0 \quad (\eps \searrow 0)
\end{align*}
uniformly in $t \in I$, 
but this does not already follow from the trivial adiabatic theorems above (Proposition~\ref{prop: triv adsatz}~(i) and~(ii)). Indeed, as $[P_0,C] \ne 0$, it follows that $P'(t) \ne 0$ for every $t \in I$. 
%
And if $\mu$ is chosen in such a way that $\mu(t_0) = 0$ for some $t_0 \in [0,1)$, then it follows that the (block diagonal!) evolution $U_{0 \, \eps}$ for $\frac 1 \eps A_0$ for no $\gamma > 0$ satisfies the estimate~\eqref{eq: zweiter triv adsatz} uniformly in $\eps \in (0,\infty)$, whence by Proposition~\ref{prop: störreihe für gestörte zeitentw} the same is true for the evolution $\tilde{U}_{0 \, \eps}$ for $\frac 1 \eps A_0 + C = \frac 1 \eps A_0 + R' R^{-1}$ and finally 
(by the proof of Corollary~\ref{cor: Kato85}) also for the evolution $U_{\eps}$ for $\frac 1 \eps A$ that we are interested in.   
$\blacktriangleleft$
\end{ex}

In the next exmaple, the spectral subsets $\sigma(t) = \{ \lambda(t) \}$ 
are singletons consisting of spectral values $\lambda(t) \in i \R$ of $A(t)$ that are not eigenvalues and, a fortiori, are not poles (but essential singularities) of $(\,.\,-A(t))^{-1}$. In particular, the adiabatic theorem with spectral gap condition from~\cite{AvronGraf11} 
cannot be applied here (also see Example~4 of~\cite{AvronGraf11}). 
We make use of the Volterra operator $V$ in $L^2([0,1])$ defined by
\begin{align*}
(V f)(t) := \int_0^t f(\tau) \, d\tau \quad (t \in I).
\end{align*}
Since $V$ is quasinilpotent and both $V$ and $V^*$ are injective, it follows that $\sigma(V) = \{0\} = \sigma_c(V)$.

\begin{ex} \label{ex: lambda wesentl sing}
Suppose $A$, $\sigma$, $P$ with $A(t) = R(t)^{-1} A_0(t) R(t)$, $\sigma(t) = \{ \lambda(t) \} := \{0\}$, $P(t) = R(t)^{-1} P_0 R(t)$, and $R(t) = e^{C t}$ are given as follows in $X := L^2(I) \times L^2(I)$:
\begin{align*}
A_0(t) := \begin{pmatrix} -V & 0 \\ 0 & a(t) + M_f \end{pmatrix} 
\quad \text{and} \quad
P_0 := \begin{pmatrix} 1 & 0 \\ 0 & 0 \end{pmatrix},
\end{align*} 
where $V$ is the Volterra operator defined above and where $f: I \to \C$ is a measurable function with $\operatorname{ess-rg}f = [-1,0]$. Additionally, suppose the function $t \mapsto a(t) \in (-\infty,0]$ is Lipschitz continuous and falls into $0$ at only countably many points accumulating at only finitely many points, and 
\begin{align*}
C := \begin{pmatrix} 0 & 1 \\ -1 & 0 \end{pmatrix}.
\end{align*}
Since $-V$ and $a(t)+M_f$ are dissipative in $L^2(I)$, the family $A_0$ is $(1,0)$-stable and, by the unitarity of the rotation operators $e^{C t} = R(t)$, the same goes for $A$. Also, since $P_0$ commutes with $A_0(t)$ and 
\begin{gather*}
\sigma(A_0(t)|_{P_0 X}) = \sigma(-V) = \{0\} = \sigma(t), \\ 
\sigma(A_0(t)|_{(1-P_0)X}) = \sigma(a(t)+M_f) = a(t)+[-1,0] = \sigma(A_0(t))\setminus \sigma(t)
\end{gather*} 
for every $t \in I \setminus N$, 
$P_0$ is associated with $A_0(t)$ and $\sigma(t)$ for every $t \in I \setminus N$ and hence the same holds true for $P(t)$ and $A(t)$ instead of $P_0$, $A_0(t)$. All other hypotheses of Theorem~\ref{thm: handl adsatz mit nichtglm sl} are clear.  
$\blacktriangleleft$
\end{ex}

We 
finally give a simple example 
showing that the conclusion of the above adiabatic theorems 
will, in general, fail if the evolution systems $U_{\eps}$ for $\frac 1 \eps A$ are not bounded in $\eps$ (and hence $A$ is not $(M,0)$-stable).

\begin{ex} \label{ex: (M,0)-stabilität wesentl, mit sl}
Suppose $A$, $\sigma$, $P$ with $A(t) := R(t)^{-1} A_0(t) R(t)$, $\sigma(t) := \{\lambda(t)\}$ and $P(t) := R(t)^{-1} P_0 R(t)$ are given as follows in $X := \ell^2(I_2)$:
\begin{align*}
A_0(t) := \begin{pmatrix} \lambda(t) & 0 \\ 0 & 0 \end{pmatrix}, \quad  P_0 := \begin{pmatrix} 1 & 0 \\ 0 & 1 \end{pmatrix}, \quad  R(t) := e^{C t}  \text{ \, with \, } C := 2 \pi \begin{pmatrix} 0  & 1 \\ -1 & 0 \end{pmatrix},
\end{align*}
and $t \mapsto \lambda(t) \in [0,\infty)$ is Lipschitz continuous such that $\lambda(\,.\,)$ at only countably many points accumulating at only finitely many points falls into $0$. 
Then all the hypotheses of Theorem~\ref{thm: handl adsatz mit nichtglm sl} 
are fullfilled with the sole exception that $A$ is not $(M,0)$-stable (because $\sigma(A(t)) = \{ 0, \lambda(t) \}$ is contained in the closed left half-plane only for countably many $t \in I$) 
and, in fact, the conclusion of this theorems 
fails. Indeed, since
\begin{align*}
R(t) = e^{C t} = \begin{pmatrix} \cos(2 \pi t) & \sin(2\pi t) \\ -\sin(2\pi t) & \cos(2\pi t) \end{pmatrix},
\end{align*}
we see that
\begin{align*}
A(t) = R(t)^{-1} A_0(t) R(t) = \lambda(t) \begin{pmatrix} \cos^2(2\pi t) & \cos(2\pi t) \sin(2 \pi t) \\ \cos(2\pi t) \sin(2 \pi t)  & \sin^2(2\pi t) \end{pmatrix}
\end{align*}
is a positive linear map (in the lattice sense) for all $t \in [0,t_0]$ with $t_0 := \frac{1}{4}$. And since $1-P(t_0) = P_0$, we see (by the series expansion for $U_{\eps}$) 
that
\begin{align*}
&\norm{(1-P(t_0))U_{\eps}(t_0)P(0) e_1} = \big| \scprd{ e_1, U_{\eps}(t_0) e_1} \big| = \scprd{ e_1, U_{\eps}(t_0) e_1} \\
&\qquad \qquad \ge 1 + \frac 1 \eps \int_0^{t_0} \scprd{e_1, A(\tau) e_1} \, d\tau = 1 + \frac 1 \eps \int_0^{t_0} \lambda(\tau) \cos^2(2\pi \tau) \, d\tau,
\end{align*} 
which right hand side does not converge to $0$ as $\eps \searrow 0$, as desired. 
$\blacktriangleleft$
\end{ex}

An example with non-diagonalizable $A(t)$ and $\sigma(A(t)) = \{0,i\}$ showing as well 
that the conclusion of the above adiabatic theorems 
will generally fail if the family $A$ is not $(M,0)$-stable 
can be found in Joye's paper~\cite{Joye07} at the end of Section~1. A generic 
version of this (non-generic)  
example is given by the following data: 
$A(t) := R(t)^{-1} A_0(t) R(t)$, $\lambda(t) := 0$, $P(t) := R(t)^{-1} P_0 R(t)$  in $X := \ell^2(I_3)$, where
\begin{align*}
A_0(t) = A_0 := \begin{pmatrix} 0 & i & 0 \\ 0 & 0 & 0 \\ 0 & 0 & i \end{pmatrix} 
\quad \text{and} \quad 
R(t):= e^{C t} \quad \text{with} \quad C:= \begin{pmatrix} 0 & 0 & 0 \\ i k & 0 & k \\ 1 & -1 & 0 \end{pmatrix}
\end{align*}
for a parameter $k \in (-\infty, 0)$ and where $P_0$ is the orthogonal projection onto $\spn\{ e_1, e_2 \}$.

\section{Adiabatic theorems without spectral gap condition for time-independent domains} \label{sect: adsätze ohne sl}

After having established 
general adiabatic theorems with spectral gap condition in Section~3, 
we can 
now prove 
an adiabatic theorem without spectral gap condition for general operators $A(t)$ with not necessarily weakly semisimple spectral values $\lambda(t)$: in Section~4.1 it appears in a qualitative version 
and in Section~4.2 in a quantitatively refined version, 
which, in turn, is applied to the special case of spectral operators of scalar type. 
%
We thereby generalize the recent adiabatic theorems without spectral gap condition of Avron, Fraas, Graf, Grech from~\cite{AvronGraf11} and of Schmid from~\cite{dipl}, which theorems -- although independently obtained -- are essentially the same (save for some regularity subtleties). 
In these theorems -- which so far are the only ones to cover 
not necessarily skew self-adjoint operators $A(t)$ 
in the case without spectral gap -- 
the considered eigenvalues $\lambda(t)$ are required 
to be weakly semisimple. Since, however, the eigenvalues of general operators are generally not weakly semisimple (Section~4.3 provides simple examples for this), it is natural to ask whether one can do without the requirement of weak semisimplicity (or, in other words, weak $1$-associatedness). And the theorems below show that one actually can: 
indeed, in essence, it suffices to require 
just weak associatedness -- which, at the beginning of Section~2.1, has been explained to be a fairly natural assumption.

\subsection{A qualitative adiabatic theorem without spectral gap condition} 

We begin with two lemmas that will be crucial in the proofs of the presented adiabatic theorems without spectral gap condition.
%

\begin{lm} \label{lm: lm 1 zum erw adsatz ohne sl}
Suppose that $A: D(A) \subset X \to X$ is a densely defined closed linear map 
and that $\lambda \in \sigma(A)$ and $\delta_0 \in (0,\infty)$ and $\vartheta_0 \in \R$ such that $\lambda + \delta e^{i \vartheta_0} \in \rho(A)$ for all $\delta \in (0,\delta_0]$. Suppose finally that $P$ is a bounded projection in $X$ such that $P A \subset A P$ and
\begin{align*}
(1-P)X \subset \overline{\rg(A-\lambda)^{m_0}}
\end{align*}
for some $m_0 \in \N$, and that there is $M_0 \in (0,\infty)$ such that
\begin{align*}
\norm{ \big( \lambda + \delta e^{i \vartheta_0} - A \big)^{-1} (1-P) } \le \frac{M_0}{\delta}
\end{align*}
for all $\delta \in (0, \delta_0]$. 
Then $\delta \big( \lambda + \delta e^{i \vartheta_0} - A \big)^{-1} (1-P)x \longrightarrow 0$ as $\delta \searrow 0$ for all $x \in X$. 
\end{lm}

\begin{proof}
If $x \in \rg(A-\lambda)^{m_0}$, then $x = (\lambda - A)^{m_0}x_0$ for some $x_0 \in D(A^{m_0})$ and, by the assumed resolvent estimate,
\begin{align*} 
\delta \big( \lambda + \delta e^{i \vartheta_0} - A \big)^{-1} \ol{P}x 
=& \,\, \delta \big( \lambda + \delta e^{i \vartheta_0} - A \big)^{-1} \ol{P} \big(-\delta e^{i \vartheta_0} \big)^{m_0}x_0 \notag \\
&+ \delta \sum_{k=1}^{m_0} \binom{m_0}{k} \big( \lambda + \delta e^{i \vartheta_0} - A \big)^{k-1} \big(-\delta e^{i \vartheta_0} \big)^{m_0-k} \ol{P} x_0
\longrightarrow 0
\end{align*}
as $\delta \searrow 0$, where of course $\ol{P} := 1-P$.
And if $x \in X$, then $\ol{x} := \ol{P}x$ can be approximated arbitrarily well by elements $y$ of $\rg(A-\lambda)^{m_0}$ and therefore
\begin{align*}
\delta \big( \lambda + \delta e^{i \vartheta_0} - A \big)^{-1} \ol{P}x = \delta \big( \lambda + \delta e^{i \vartheta_0} - A \big)^{-1} \overline{P} (\overline{x}-y) +  \delta \big( \lambda + \delta e^{i \vartheta_0} - A \big)^{-1} \overline{P} y 
\end{align*}
can be made arbitrarily small for $\delta$ small enough by the assumed resolvent estimate and by what has just been shown.
\end{proof}

\begin{lm} \label{lm: lm 2 zum erw adsatz ohne sl}
Suppose that $A(t): D(A(t)) \subset X \to X$ for every $t \in I$ is a densely defined closed linear map, that $\lambda(t) \in \sigma(A(t))$ and $\delta_0 \in (0,\infty)$ and $\vartheta(t) \in \R$ such that $\lambda(t) + \delta e^{i \vartheta(t)} \in \rho(A(t))$ for all $\delta \in (0,\delta_0]$ and $t \in I$ and such that 
\begin{align*}
t \mapsto \big( \lambda(t) + \delta e^{i \vartheta(t)} - A(t) \big)^{-1}
\end{align*}
is in $W^{1,\infty}_*(I,L(X))$ (resp.~SOT-continuously differentiable) 
and $t \mapsto \lambda(t)$ as well as $t \mapsto e^{i \vartheta(t)}$ is Lipschitz continuous (resp.~continuously differentiable).
Suppose further that $P(t)$ for every $t \in I$ is a bounded projection in $X$ such that $P(t)A(t) \subset A(t)P(t)$ for every $t \in I$ and $P(t)X \subset \ker(A(t)-\lambda(t))^{m_0}$ for every $t \in I$ (where $m_0 \in \N$) and $t \mapsto P(t)$ is SOT-continuously differentiable. 
Then 
\begin{align*}
t \mapsto (A(t)-\lambda(t))P(t)
\end{align*}
is in $W^{1,\infty}_*(I,L(X))$ (resp.~SOT-continuously differentiable) and, in particular, for every $\eps \in (0,\infty)$ the evolution system $V_{0\,\eps}$ for $\frac 1 \eps AP + [P',P]$ exists on $X$ and is adiabatic w.r.t.~$P$. 
If, in addition, the evolution system $U_{\eps}$ for $\frac 1 \eps A$ exists on $D(A(t))$ 
and if there is an $M \in (0,\infty)$ such that $\norm{U_{\eps}(t,s)} \le M$ for all $(s,t) \in \Delta$ and $\eps \in (0,\infty)$, then
\begin{align*}
\norm{ V_{0\,\eps}(t,s) P(s) } \le Mc \, e^{Mc(t-s)} 
\end{align*}
for all $(s,t) \in \Delta$, where $c$ is an upper bound of $t \mapsto \norm{P(t)}, \norm{ P'(t) }$. 
\end{lm}

\begin{proof}
Since $P(t)X \subset \ker(A(t)-\lambda(t))^{m_0} \subset D(A(t)^{m_0})$ 
for every $t \in I$, 
we see that $(A(t)-\lambda(t))P(t)$ is a bounded linear map in $X$ for every $t \in I$ and that 
\begin{align}  \label{eq: stet von A P, lm 2 zum erw adsatz ohne sl}
t \mapsto \,\, &(A(t)-\lambda(t))P(t) = (A(t)-\lambda(t)) S_{\delta}(t) \,\, S_{\delta}(t)^{m_0-1} \big( A(t)-\lambda(t)-\delta e^{i \vartheta(t)} \big)^{m_0} P(t) \notag \\
&= \big( 1 + \delta e^{i \vartheta(t)} S_{\delta}(t) \big) \, \sum_{k=0}^{m_0-1} \binom{m_0}{k}  \big( -\delta e^{i \vartheta(t)} \big)^{m_0-k} \cdot \notag \\ 
&\qquad \qquad \qquad \qquad \qquad \qquad \qquad \cdot S_{\delta}(t)^{m_0-1 - k} \, \big( 1 + \delta e^{i \vartheta(t)} S_{\delta}(t) \big)^k P(t) 
\end{align}
is in $W^{1,\infty}_*(I,L(X))$ (resp.~SOT-continuously differentiable) by Lemma~\ref{lm: prod- und inversenregel}, because
\begin{align*}
t \mapsto S_{\delta}(t) := \big( A(t)-\lambda(t)-\delta e^{i \vartheta(t)} \big)^{-1}
\end{align*}
is of that regularity by Lemma~\ref{lm: prod- und inversenregel}. In particular, $t \mapsto A(t)P(t)$ is SOT-continuous and therefore 
the evolution system $V_{\eps}$ for $\frac 1 \eps AP + [P',P]$ exists on $X$ and (by virtue of Proposition~\ref{prop: intertwining relation}) is adiabatic w.r.t.~$P$ for every $\eps \in (0,\infty)$. 
Suppose finally that the additional assumption concerning the evolution system $U_{\eps}$ for $\frac 1 \eps A$ is satisfied. Since for all $x \in X$ and $(s,t) \in \Delta$ the map $[s,t] \ni \tau \mapsto U_{\eps}(t,\tau)V_{0\,\eps}(\tau,s)P(s)x$ is continuous and right differentiable (Lemma~\ref{lm: prodregel rechtsseit db}) with bounded (even continuous)
right derivative
\begin{align*}
\tau \mapsto & \, \, U_{\eps}(t,\tau) \Big( \frac 1 \eps A(\tau)P(\tau)- \frac 1 \eps A(\tau) + [P'(\tau),P(\tau)] \Big) V_{0\,\eps}(\tau,s)P(s)x \\
&= U_{\eps}(t,\tau) P'(\tau) V_{0\,\eps}(\tau,s)P(s)x 
\end{align*}
(where in the last equality the adiabaticity of $V_{\eps}$ w.r.t.~$P$ and~\eqref{eq: PP'P=0} have been used), it follows from Lemma~\ref{lm: rechtsseit db und W^{1,infty}} that
\begin{align}  \label{eq: intgl, lm 2 zum erw adsatz ohne sl}
V_{0\,\eps}(t,s)P(s)x - U_{\eps}(t,s)P(s)x &= U_{\eps}(t,\tau) V_{0\,\eps}(\tau,s)P(s)x \big|_{\tau=s}^{\tau=t} \notag \\
&= \int_s^t U_{\eps}(t,\tau) P'(\tau) V_{0\,\eps}(\tau,s)P(s)x \, d\tau
\end{align} 
for all $(s,t) \in \Delta$ and $x \in X$. And this integral equation, by the Gronwall inequality, yields the desired estimate for $V_{0\,\eps}(t,s)P(s)$.
\end{proof}

With these lemmas at hand, we can now prove the announced 
general adiabatic theorem without spectral gap condition for not necessarily weakly semisimple eigenvalues. 
Similarly to the works~\cite{AvronElgart99} of Avron and Elgart and~\cite{Teufel01} of Teufel 
its proof rests upon solving a suitable approximate commutator equation. 
In this undertaking the insights gained in~Section~3, especially formula~\eqref{eq: wegweiser}, will prove 
indispensable. 
(Alternatively, part~(i) of the theorem could also -- 
less elegantly -- 
be based upon 
a suitable iterated partial integration argument, but part~(ii) could not.)

\begin{thm} \label{thm: erw adsatz ohne sl}
Suppose $A(t): D \subset X \to X$ for every $t \in I$ is a linear map such that Condition~\ref{cond: reg 1} is satisfied with $\omega = 0$. Suppose further that $\lambda(t)$ for every $t \in I$ is an eigenvalue of $A(t)$, and that there are numbers $\delta_0 \in (0,\infty)$ and $\vartheta(t) \in \R$ such that $\lambda(t) + \delta e^{i \vartheta(t)} \in \rho(A(t))$ for all $\delta \in (0,\delta_0]$ and $t \in I$ and such that $t \mapsto \lambda(t)$ and $t \mapsto e^{i \vartheta(t)}$ are Lipschitz continuous. 
Suppose finally that $P(t)$ for every $t \in I$ is a bounded projection in $X$ commuting with $A(t)$ 
such that $P(t)$ for almost every $t \in I$ is weakly associated with $A(t)$ and $\lambda(t)$, suppose there is an $M_0 \in (0,\infty)$ such that 
\begin{align*}
\norm{ \big( \lambda(t) + \delta e^{i \vartheta(t)} - A(t) \big)^{-1} (1-P(t)) } \le \frac{M_0}{\delta} 
\end{align*}
for all $\delta \in (0, \delta_0]$ and $t \in I$, let $\rk P(0) < \infty$ and suppose that $t \mapsto P(t)$ is SOT-continuously differentiable. 
%
\begin{itemize}
\item[(i)] If $X$ is arbitrary (not necessarily reflexive), then
\begin{align*} 
\sup_{t \in I} \norm{ \big( U_{\eps}(t) - V_{0 \, \eps}(t) \big) P(0) } \longrightarrow 0 \quad (\eps \searrow 0),
\end{align*}
where $V_{0 \, \eps}$ is the evolution system for $\frac 1 \eps  A P + [P',P]$ on $X$ for every $\eps \in (0,\infty)$. 
\item[(ii)] If $X$ is reflexive and $t \mapsto P(t)$ is norm continuously differentiable, then
\begin{align*}
\sup_{t \in I} \norm{ U_{\eps}(t) - V_{\eps}(t) } \longrightarrow 0 \quad (\eps \searrow 0),
\end{align*}
whenever the evolution system $V_{\eps}$ for $\frac 1 \eps A + [P',P]$ exists on $D$ for every $\eps \in (0, \infty)$.
\end{itemize}
\end{thm}

\begin{proof}
We begin with some preparations which will be 
used in the proof of both assertion~(i) and assertion~(ii).
As a first preparatory step, we show that $t \mapsto P(t)$ is in $W^{1,\infty}_*(I,L(X,Y))$ and that there is an $m_0 \in \N$ such that $P(t)X \subset \ker(A(t)-\lambda(t))^{m_0}$ for every $t \in I$.
%
%
Since $P(t)$ for almost every $t \in I$ is weakly associated with $A(t)$ and $\lambda(t)$ and since 
\begin{align*}
\dim P(t)X = \rk P(0)X < \infty
\end{align*}
for every $t \in I$ (which equality is due to the continuity of $t \mapsto P(t)$ and Lemma~VII.6.7 of~\cite{DunfordSchwartz}), there is a $t$-independent constant $m_0 \in \N$ -- for instance, $m_0 := \rk P(0)$ -- such that $P(t)$ is weakly $m_0$-associated with $A(t)$ and $\lambda(t)$ for almost every $t \in I$. 
In particular, it follows from Theorem~\ref{thm: typ mögl für PX und (1-P)X} that
\begin{align*}
P(t)X \subset \ker(A(t)-\lambda(t))^{m_0} \quad \text{and} \quad (1-P(t))X \subset \overline{ \rg(A(t)-\lambda(t))^{m_0} }
\end{align*} 
for almost every $t \in I$. 
%
%
We now show that the inclusion
\begin{align} \label{eq: gl 0, adsatz ohne sl}
P(t)X \subset \ker(A(t)-\lambda(t))^{m_0}
\end{align}
actually 
holds for every $t \in I$, 
in order to be able to apply Lemma~\ref{lm: lm 2 zum erw adsatz ohne sl}. 
Since $P(t)$ commutes with $A(t)$ 
and $\dim P(t)X = \rk P(0) < \infty$ for every $t \in I$, 
one has $P(t) D(A(t)) = P(t)X$ and hence $P(t)X \subset D(A(t)^{m_0})$ as well as 
\begin{align*}
(A(t)-\lambda(t))^{m_0} P(t) = \big( (A(t)-\lambda(t))P(t) \big)^{m_0}
\end{align*}
for every $t \in I$. So in order to establish~\eqref{eq: gl 0, adsatz ohne sl} it suffices to show that $t \mapsto A(t)P(t)$ is SOT-continuous, because the 
set $I \setminus N$ of those $t \in I$ where $P(t)$ is weakly associated with $A(t)$ and $\lambda(t)$ is dense in $I$. 
Analagously to~\eqref{eq: stet von A P, lm 2 zum erw adsatz ohne sl} it follows that
\begin{align*}
I \setminus N \ni t \mapsto \,\, & A(0)P(t) 
= A(0) S_{\delta}(t)  \, \sum_{k=0}^{m_0-1} \binom{m_0}{k}  \big( -\delta e^{i \vartheta(t)} \big)^{m_0-k}  \cdot \notag \\ 
&\qquad \qquad \qquad \qquad \qquad \qquad \qquad \cdot S_{\delta}(t)^{m_0-1 - k} \, \big( 1 + \delta e^{i \vartheta(t)} S_{\delta}(t) \big)^k P(t) 
\end{align*}
extends to a $W^{1,\infty}_*(I,L(X))$-regular 
map and the closedness of $A(0)$ implies that the $W^{1,\infty}_*(I,L(X))$-regular extension is given by $I \ni t \mapsto A(0)P(t)$. In other words, $t \mapsto P(t)$ (by the definition of the norm of $Y$ in Condition~\ref{cond: reg 1}) belongs to $W^{1,\infty}_*(I,L(X,Y))$. 
In particular, $t \mapsto A(t)P(t)$ is 
SOT-continuous and~\eqref{eq: gl 0, adsatz ohne sl} follows. 
\smallskip

As a second preparatory step, we solve -- in accordance 
with the proof of the adiabatic theorems with spectral gap condition -- 
a suitable (approximate) commutator equation. Inspired by~\eqref{eq: wegweiser}, we define the operators 
\begin{align}
&B_{n \, \bm{\delta}}(t) := \sum_{k=0}^{m_0-1} \Big(  \prod_{i=1}^{k+1} \overline{R}_{\delta_i}(t) \Big) Q_n(t) (\lambda(t)-A(t))^k P(t) \notag \\
&\qquad \qquad \qquad \qquad \qquad \quad + \sum_{k=0}^{m_0-1} (\lambda(t)-A(t))^k P(t) Q_n(t) \Big(  \prod_{i=1}^{k+1} \overline{R}_{\delta_i}(t) \Big)
\end{align}
for $n \in \N$, $\bm{\delta} := (\delta_1, \dots, \delta_{m_0}) \in (0,\delta_0]^{m_0}$ and $t \in I$, where
\begin{align*}
\overline{R}_{\delta}(t) := R_{\delta}(t) \overline{P}(t) \quad \text{with} \quad R_{\delta}(t) := \big( \lambda(t) + \delta e^{i \vartheta(t)} - A(t) \big)^{-1} \text{\, and \,\,} \overline{P}(t) := 1-P(t)
\end{align*}
for $\delta \in (0,\delta_0]$, and where 
\begin{align*}
Q_n(t) := \int_0^1 j_{\frac{1}{n}}(t-r) P'(r) \, dr.
\end{align*}
In other words, $Q_n$ is obtained from $P'$ by mollification, whence $t \mapsto Q_n(t)$ is SOT-continuously differentiable and $Q_n(t) \longrightarrow P'(t)$ as $n \to \infty$ w.r.t.~SOT for 
$t \in (0,1)$ and
\begin{align*}
\sup \{ \norm{Q_n(t)}: t \in I, n \in \N \} \le \sup_{t \in I} \norm{P'(t)}.
\end{align*}
We now show that the operators $B_{n \, \bm{\delta}}(t)$ satisfy the approximate commutator equation
\begin{align} \label{eq: appr commutator equation}
B_{n \, \bm{\delta}}(t)A(t) - A(t)B_{n \, \bm{\delta}}(t) + C_{n \, \bm{\delta}}(t) \subset [Q_n(t),P(t)]
\end{align}
with remainder terms $C_{n \, \bm{\delta}}(t)$ that 
will have to be suitably controlled below. Since
\begin{align*}
(\lambda-A) \Big(  \prod_{i=1}^{k+1} \overline{R}_{\delta_i} \Big) =  \Big( \prod_{1 \le i \le k} \overline{R}_{\delta_i} \Big) - \delta_{k+1} e^{i \vartheta} \Big( \prod_{i=1}^{k+1} \overline{R}_{\delta_i} \Big) 
\supset \Big(  \prod_{i=1}^{k+1} \overline{R}_{\delta_i} \Big) (\lambda-A)
\end{align*}
(the $t$-dependence being suppressed here and in the following lines for the sake of convenience), it follows that
\begin{align*}
(\lambda-A) B_{n \, \bm{\delta}} &= \sum_{k=0}^{m_0-1} \Big( \prod_{1 \le i \le k} \overline{R}_{\delta_i} \Big) Q_n (\lambda-A)^k P   +  \sum_{k=0}^{m_0-1} (\lambda-A)^{k+1} P Q_n \Big(  \prod_{i=1}^{k+1} \overline{R}_{\delta_i} \Big)
- C_{n \, \bm{\delta}}^+ \\
B_{n \, \bm{\delta}} (\lambda-A) &\subset \sum_{k=0}^{m_0-1} \Big(  \prod_{i=1}^{k+1} \overline{R}_{\delta_i} \Big) Q_n (\lambda-A)^{k+1} P   +  \sum_{k=0}^{m_0-1} (\lambda-A)^k P Q_n  \Big( \prod_{1 \le i \le k} \overline{R}_{\delta_i} \Big)
- C_{n \, \bm{\delta}}^-
\end{align*}
where we used the abbreviations
\begin{align}
&C_{n \, \bm{\delta}}^+ := \sum_{k=0}^{m_0-1} \delta_{k+1} e^{i \vartheta} \Big( \prod_{i=1}^{k+1} \overline{R}_{\delta_i} \Big) Q_n (\lambda-A)^k P,  \notag \\
&\qquad \qquad \qquad \qquad \qquad C_{n \, \bm{\delta}}^- := \sum_{k=0}^{m_0-1} (\lambda-A)^k P Q_n \, \delta_{k+1} e^{i \vartheta} \Big( \prod_{i=1}^{k+1} \overline{R}_{\delta_i} \Big).
\end{align}
Subtracting $B_{n \, \bm{\delta}} (\lambda-A)$ from $(\lambda-A)B_{n \, \bm{\delta}}$ and noticing that, by doing so, of all the summands not 
belonging to $C_{n \, \bm{\delta}}^+$, $C_{n \, \bm{\delta}}^-$ only
\begin{align*}
Q_n P - \Big( \prod_{i=1}^{m_0} \overline{R}_{\delta_i} \Big) Q_n (\lambda-A)^{m_0} P    +   (\lambda-A)^{m_0} P Q_n \Big( \prod_{i=1}^{m_0} \overline{R}_{\delta_i} \Big) - P Q_n
= [Q_n,P]  
\end{align*}
remains (remember~\eqref{eq: gl 0, adsatz ohne sl}), 
we see that 
\begin{align*}
B_{n \, \bm{\delta}} A - A B_{n \, \bm{\delta}} \subset [Q_n,P] - C_{n \, \bm{\delta}}^+ + C_{n \, \bm{\delta}}^-
\end{align*}
which is nothing but~\eqref{eq: appr commutator equation} if one defines $C_{n \, \bm{\delta}} := C_{n \, \bm{\delta}}^+ - C_{n \, \bm{\delta}}^-$.
\smallskip

As a third 
preparatory step we observe that $t \mapsto B_{n \, \bm{\delta}}(t)$ belongs to $W^{1,\infty}_*(I,L(X,Y))$ and estimate $B_{n \, \bm{\delta}}$ as well as $B_{n \, \bm{\delta}}'$. Since $t \mapsto A(t) - \lambda(t) - \delta e^{i \vartheta(t)}$ is in $W^{1,\infty}_*(I,L(Y,X))$, 
\begin{align*}
t \mapsto (A(t)-\lambda(t))^k P(t) = \big( (A(t)-\lambda(t))P(t) \big)^k
\end{align*}
is in $W^{1,\infty}_*(I,L(X))$ by Lemma~\ref{lm: lm 2 zum erw adsatz ohne sl}, 
and $t \mapsto P(t)$ is in $W^{1,\infty}_*(I,L(X,Y))$ by the first preparatory step, 
the asserted $W^{1,\infty}_*(I,L(X,Y))$-regularity of $t \mapsto B_{n \, \bm{\delta}}(t)$ follows from 
Lemma~\ref{lm: prod- und inversenregel}.
Additionally, there is a constant $c$ such that
\begin{align} \label{eq: absch B_n eps}
\sup_{t \in I} \big\| B_{n \, \bm{\delta}}(t) \big\| \le \sum_{k=1}^{m_0} c \, \Big( \prod_{i=1}^{k} \delta_i \Big)^{-1}
\end{align}
for all $\bm{\delta} \in (0,\delta_0]^{m_0}$ by the assumed resolvent estimate and Lemma~\ref{lm: lm 2 zum erw adsatz ohne sl}. 
And since 
\begin{gather*}
\norm{R_{\delta}(t)}_{X \to X} 
\le \sum_{k=0}^{m_0-1} \frac{1}{\delta^{k+1}} \norm{(A(t)-\lambda(t))^k P(t)}_{X \to X}  + \,\, \norm{\overline{R}_{\delta}(t)}_{X \to X}
\le \frac{c}{\delta^{m_0}} \\
\text{ as well as } \\
\norm{\overline{R}_{\delta}(t)}_{X \to Y} \le \norm{(A(t)-1)^{-1}}_{X \to Y} \norm{(A(t)-1) \overline{R}_{\delta}(t)}_{X \to X} \le \frac{c}{\delta}
\end{gather*}
for all $t \in I$ and all $\delta \in (0,\delta_0]$ (with another constant $c$) by the assumed resolvent estimate and Lemma~\ref{lm: lm 2 zum erw adsatz ohne sl}, it follows from Lemma~\ref{lm: prod- und inversenregel} that there is a $W^{1,\infty}_*$-derivative $\overline{R}_{\delta}'$ of $t \mapsto \overline{R}_{\delta}(t)$ such that
\begin{align} \label{eq: absch R_eps'}
\esssup_{t \in I} \big\| \overline{R}_{\delta}'(t) \big\| \le \frac{c}{\delta^{m_0+1}}
\end{align}
for all $\delta \in (0,\delta_0]$ (with yet another constant $c$) and, hence, that there is a $W^{1,\infty}_*$-derivative $B_{n \, \bm{\delta}}'$ of $t \mapsto B_{n \, \bm{\delta}}(t)$ such that
\begin{align}  \label{eq: absch B_n eps'}
\esssup_{t \in I} \norm{ B_{n \, \bm{\delta}}'(t) } \le \sum_{k=1}^{m_0} c_n \, \Big( \prod_{i=1}^{k} \delta_i \Big)^{-(m_0+1)}
\end{align}
for all $\bm{\delta} \in (0,\delta_0]^{m_0}$ and some constant $c_n \in (0,\infty)$ depending on the supremum norm $\sup_{t \in I} \norm{Q_n'(t)}$ of the SOT-derivative of $t \mapsto Q_n(t)$.
\smallskip

After these preparations we can now turn to the main part of the proof where the cases~(i) and~(ii) have to be treated 
separately. We first prove assertion~(i). 
%
%
As has already been shown in~\eqref{eq: intgl, lm 2 zum erw adsatz ohne sl}, 
\begin{align*}     
\big( V_{0\,\eps}(t)- U_{\eps}(t) \big) P(0)x &= U_{\eps}(t,s)V_{0\,\eps}(s)P(0)x \big|_{s=0}^{s=t} = \int_0^t U_{\eps}(t,s) \, P'(s) \, V_{0\,\eps}(s) P(0)x \,ds 
\end{align*}
so that, by rewriting the right hand side of this equation, we obtain 
\begin{align} \label{eq: gl 1, adsatz ohne sl}
\big( V_{0\,\eps}(t)- U_{\eps}(t) \big) P(0)x =& \int_0^t U_{\eps}(t,s)  \, (P'(s)-Q_n(s))P(s) \, V_{0\,\eps}(s)P(0)x \, ds \notag \\
&+ \int_0^t U_{\eps}(t,s)  \, [Q_n(s),P(s)] \, V_{0\,\eps}(s)P(0)x \, ds
\end{align}
for all $t \in I$, $\eps \in (0, \infty)$ and $x \in X$.
Since $Q_n(s)P(s) \longrightarrow P'(s)P(s)$ for every $s \in (0,1)$ by the SOT-convergence of $(Q_n(s))$ to $P'(s)$ for $s \in (0,1)$ and by $\rk P(s) = \rk P(0) < \infty$ for $s \in I$,
it follows by Lemma~\ref{lm: lm 2 zum erw adsatz ohne sl} and by the dominated convergence theorem that
\begin{align} \label{eq: gl 2, adsatz ohne sl}
\sup_{\eps \in (0,\infty)} \sup_{t \in I} \norm{ \int_0^t U_{\eps}(t,s)  \, (P'(s)-Q_n(s))P(s) \, V_{0\,\eps}(s)P(0) \, ds  }  
\longrightarrow 0 
\end{align}
as $n \to \infty$.
In view of~\eqref{eq: gl 1, adsatz ohne sl} we therefore have to show that for each fixed $n \in \N$ 
\begin{align} \label{eq: zwbeh 1, adsatz ohne sl}
\sup_{t \in I} \norm{  \int_0^t U_{\eps}(t,s)  \, [Q_n(s),P(s)] \, V_{0\,\eps}(s)P(0) \, ds  } \longrightarrow 0 
\end{align}
as $\eps \searrow 0$.
So let $n \in \N$ be fixed for the rest of the proof.
%
%
Since $s \mapsto B_{n \, \bm{\delta}}(s)$ is in $W^{1,\infty}_*(I,L(X,Y))$ by the third preparatory step and since $[0,t] \ni s \mapsto U_{\eps}(t,s)|_Y \in L(Y,X)$ as well as $s \mapsto V_{0\,\eps}(s) \in L(X)$ are SOT-continuously differentiable, Lemma~\ref{lm: prod- und inversenregel} yields that
\begin{align*}
[0,t] \ni s \mapsto U_{\eps}(t,s) B_{n \, \bm{\delta}}(s) V_{0\,\eps}(s) P(0)x 
\end{align*}
is the continuous representative of an element of $W^{1,\infty}([0,t],X)$ for every $x \in X$. With the help of the approximate commutator equation~\eqref{eq: appr commutator equation} of the second preparatory step, we therefore see that
\begin{gather}  
\int_0^t U_{\eps}(t,s)  \, [Q_n(s),P(s)] \, V_{0\,\eps}(s)P(0)x \, ds 
= \eps \, \int_0^t U_{\eps}(t,s)  \Big(\! -\frac 1 \eps A(s) B_{n \, \bm{\delta}}(s) \notag \\
+ \,\, B_{n \, \bm{\delta}}(s) \frac 1 \eps A(s) \Big) V_{0\,\eps}(s) P(0)x \, ds \,  + \int_0^t U_{\eps}(t,s) \, C_{n \, \bm{\delta}}^+(s) \, V_{0\,\eps}(s) P(0)x \, ds \notag \\
= \eps \, U_{\eps}(t,s) B_{n \, \bm{\delta}}(s) V_{0\,\eps}(s) P(0)x \Big|_{s=0}^{s=t} - \eps \, \int_0^t U_{\eps}(t,s)  \Big( B_{n \, \bm{\delta}}'(s) + B_{n \, \bm{\delta}}(s) [P'(s),P(s)] \Big) \notag \\
 V_{0\,\eps}(s) P(0)x \, ds + \int_0^t U_{\eps}(t,s) \, C_{n \, \bm{\delta}}^+(s) \, V_{0\,\eps}(s) P(0)x \, ds    \label{eq: gl 3, adsatz ohne sl}
\end{gather}
for all $t \in I$, $\eps \in (0,\infty)$, $x \in X$ and $\bm{\delta} \in (0,\delta_0]^{m_0}$. 
%
%
We now want to find functions $\eps \mapsto \delta_{1 \, \eps}, \dots, \delta_{m_0 \, \eps}$ defined on a small interval $(0,\delta_0']$ and converging to $0$ as $\eps \searrow 0$ in such a way that, if they are inserted in the right hand side of~\eqref{eq: gl 3, adsatz ohne sl}, the desired convergence~\eqref{eq: zwbeh 1, adsatz ohne sl} follows. 
In view of the estimates~\eqref{eq: absch B_n eps}, \eqref{eq: absch B_n eps'} and 
\begin{align}  \label{eq: absch C_n eps^+}
\int_0^1 \norm{ C_{n \, \bm{\delta}}^+(s) } \, ds \le \sum_{k=1}^{m_0} c \, \Big( \prod_{1 \le i < k} \delta_i \Big)^{-1} \, \int_0^1 \norm{ \delta_{k} \overline{R}_{\delta_{k}}(s) Q_n(s) P(s) } \, ds, 
\end{align} 
we would like the functions $\eps \mapsto \delta_{i \, \eps}$ to converge to $0$ so slowly that
\begin{gather}
\eps \, \Big( \prod_{i=1}^{k} \delta_{i \, \eps} \Big)^{-(m_0+1)} \longrightarrow 0 \quad (\eps \searrow 0)  \label{eq: wunsch 1 an eps_i} \\
\Big( \prod_{1 \le i < k} \delta_{i\,\eps} \Big)^{-1} \, \int_0^1 \norm{ \delta_{k\,\eps} \overline{R}_{\delta_{k\,\eps}}(s) Q_n(s) P(s) } \, ds \longrightarrow 0 \quad (\eps \searrow 0)   \label{eq: wunsch 2 an eps_i}
\end{gather}
for all $k \in \{ 1, \dots, m_0 \}$. 
Since 
\begin{align} \label{eq: wesentl grund für ex der eps_i}
\eta_{n}^+(\delta) := \int_0^1 \norm{\delta \overline{R}_{\delta}(s) Q_n(s) P(s) } \, ds \longrightarrow 0 \quad (\delta \searrow 0)
\end{align}
by Lemma~\ref{lm: lm 1 zum erw adsatz ohne sl}, by $\rk P(s) = \rk P(0) < \infty$ and by the dominated convergence theorem, 
such functions $\eps \mapsto \delta_{i\,\eps}$ really can be found. Indeed, define recursively 
\begin{gather*}
\delta_{m_0\,\eps} := \eps^{\frac{1}{(m_0+1)^2}}
\quad \text{and} \quad 
\delta_{m_0-l \,\eps} := \max \Big\{  \Big( \Big( \prod_{m_0-l+1 \le i < k} \delta_{i\,\eps} \Big)^{-1} \, \eta_{n}^+(\delta_{k\,\eps}) \Big)^{\frac{1}{2}}: \\
k \in \{ m_0-l+1, \dots, m_0 \}  \Big\} \cup \Big\{ \eps^{\frac{1}{(m_0+1)^2}} \Big\}
\end{gather*}
for $l \in \{ 1, \dots, m_0-1\}$. With the help of~\eqref{eq: wesentl grund für ex der eps_i} it then successively follows, by proceeding from larger to smaller indices $i$, that $\delta_{i\,\eps} \longrightarrow 0$ as $\eps \searrow 0$ for all $i \in \{1, \dots, m_0\}$ (so that, in particular, $\delta_{i\,\eps} \in (0,\delta_0]$ for small enough $\eps$ whence the expressions $\eta_{n}^+(\delta_{i\,\eps})$ used in the recursive definition 
make sense for small $\eps$ in the first place) and that~\eqref{eq: wunsch 1 an eps_i} and \eqref{eq: wunsch 2 an eps_i} are satisfied. Assertion~(i) 
now follows from~\eqref{eq: gl 1, adsatz ohne sl}, \eqref{eq: gl 2, adsatz ohne sl}, \eqref{eq: gl 3, adsatz ohne sl} by virtue of~\eqref{eq: absch B_n eps}, \eqref{eq: absch B_n eps'}, \eqref{eq: absch C_n eps^+} and Lemma~\ref{lm: lm 2 zum erw adsatz ohne sl}. 
\smallskip

We now prove assertion~(ii) and, for that purpose, additionally assume that $X$ is reflexive and $t \mapsto P(t)$ is NOT-continuously differentiable.
%
%
Analogously to~\eqref{eq: gl 1, adsatz ohne sl} we obtain 
\begin{align} \label{eq: gl 4, adsatz ohne sl}
\big( V_{\eps}(t)- U_{\eps}(t) \big) x =& \int_0^t U_{\eps}(t,s) \, [P'(s)-Q_n(s), P(s)] \, V_{\eps}(s)x \, ds \notag \\
&+ \int_0^t U_{\eps}(t,s) \, [Q_n(s),P(s)] \, V_{\eps}(s)x \, ds
\end{align}
for all $t \in I$, $\eps \in (0, \infty)$ and $x \in D(A(0)) = D$.
Since $Q_n(s) \longrightarrow P'(s)$ for every $s \in (0,1)$ by the additionally assumed NOT-continuous differentiability of $t \mapsto P(t)$,
it follows by Proposition~\ref{prop: störreihe für gestörte zeitentw} and by the dominated convergence theorem that
\begin{align} \label{eq: gl 5, adsatz ohne sl}
\sup_{\eps \in (0,\infty)}  \sup_{t \in I} \norm{ \int_0^t U_{\eps}(t,s) \, [P'(s)-Q_n(s), P(s)] \, V_{\eps}(s) \, ds  }  
\longrightarrow 0 
\end{align}
as $n \to \infty$.
In view of~\eqref{eq: gl 4, adsatz ohne sl} we therefore have to show that for each fixed $n \in \N$ 
\begin{align} \label{eq: zwbeh 2, adsatz ohne sl}
\sup_{t \in I} \norm{  \int_0^t U_{\eps}(t,s) \, [Q_n(s),P(s)] \, V_{\eps}(s) \, ds  } \longrightarrow 0 
\end{align}
as $\eps \searrow 0$.
So let $n \in \N$ be fixed for the rest of the proof.
%
%
Again completely analogously to the proof of~(i) it follows that
\begin{align*}
[0,t] \ni s \mapsto U_{\eps}(t,s) B_{n \, \bm{\delta}}(s) V_{\eps}(s)x 
\end{align*}
is the continuous representative of an element of $W^{1,\infty}([0,t],X)$ for every $x \in D(A(0)) = D$. With the help of the approximate commutator equation~\eqref{eq: appr commutator equation} of the second preparatory step, we therefore see that
\begin{gather}  
\int_0^t U_{\eps}(t,s) \, [Q_n(s),P(s)] \, V_{\eps}(s)x \, ds 
= \frac{1}{\eps} \, \int_0^t U_{\eps}(t,s)  \Big(\! -\frac 1 \eps A(s) B_{n \, \bm{\delta}}(s) \notag \\
+ \,\, B_{n \, \bm{\delta}}(s) \frac 1 \eps A(s) \Big) V_{\eps}(s) x \, ds \,  + \int_0^t U_{\eps}(t,s) \, C_{n \, \bm{\delta}}(s) \, V_{\eps}(s) x \, ds \notag \\
= \eps \, U_{\eps}(t,s) B_{n \, \bm{\delta}}(s) V_{\eps}(s) x \Big|_{s=0}^{s=t} - \eps \, \int_0^t U_{\eps}(t,s)  \Big( B_{n \, \bm{\delta}}'(s) + B_{n \, \bm{\delta}}(s) [P'(s),P(s)] \Big) \notag \\
 V_{\eps}(s) x \, ds + \int_0^t U_{\eps}(t,s) \, C_{n \, \bm{\delta}}(s) \, V_{\eps}(s) x \, ds    \label{eq: gl 6, adsatz ohne sl}
\end{gather}
for all $t \in I$, $\eps \in (0,\infty)$, $x \in D(A(0)) = D$ and $\bm{\delta} \in (0,\delta_0]^{m_0}$. 
%
%
In view of the estimates ~\eqref{eq: absch B_n eps}, \eqref{eq: absch B_n eps'}, \eqref{eq: absch C_n eps^+} and
\begin{align}  \label{eq: absch C_n eps^-}
\int_0^1 \norm{ C_{n \, \bm{\delta}}^-(s) } \, ds \le \sum_{k=1}^{m_0} c \, \Big( \prod_{1 \le i < k} \delta_i \Big)^{-1} \, \int_0^1 \norm{ P(s) Q_n(s) \delta_{k} \overline{R}_{\delta_{k}}(s) } \, ds,
\end{align}
we would now like to find functions $\eps \mapsto \delta_{1 \, \eps}, \dots, \delta_{m_0 \, \eps}$ defined on a small interval $(0,\delta_0']$ and converging to $0$ as $\eps \searrow 0$ so slowly that~\eqref{eq: wunsch 1 an eps_i}, \eqref{eq: wunsch 2 an eps_i} and
\begin{gather}
\Big( \prod_{1 \le i < k} \delta_{i\,\eps} \Big)^{-1} \, \int_0^1 \norm{ P(s) Q_n(s) \delta_{k\,\eps} \overline{R}_{\delta_{k\,\eps}}(s) } \, ds \longrightarrow 0 \quad (\eps \searrow 0)   \label{eq: wunsch 3 an eps_i}
\end{gather}
are satisfied for all $k \in \{ 1, \dots, m_0 \}$. 
Why is it possible to find such functions $\eps \mapsto \delta_{i\,\eps}$? In essence, this is 
because of~\eqref{eq: wesentl grund für ex der eps_i} and because
\begin{align} \label{eq: wesentl grund 2 für ex der eps_i}
\eta_{n}^-(\delta) := \int_0^1 \norm{P(s) Q_n(s) \delta \overline{R}_{\delta}(s) } \, ds \longrightarrow 0 \quad (\delta \searrow 0),
\end{align}
which last convergence can be seen as follows: by virtue of Proposition~\ref{prop: schwache assoziiertheit, dual}, which applies by the additionally assumed reflexivity of $X$, $P(s)^*$ is weakly $m_0$-associated with $A(s)^*$ and $\lambda(s)$ for almost every $s \in I$, and therefore Lemma~\ref{lm: lm 1 zum erw adsatz ohne sl} together with $\rk P(s)^* = \rk P(s) < \infty$ yields the convergence
\begin{align*}
\norm{ P(s) Q_n(s) \delta \overline{R}_{\delta}(s) } = \norm{ \delta \overline{R}_{\delta}(s)^* \, Q_n(s)^* P(s)^* }  \longrightarrow 0 \quad (\delta \searrow 0)
\end{align*}
for almost every $s \in I$, from which~\eqref{eq: wesentl grund 2 für ex der eps_i} follows by the dominated convergence theorem.
We now recursively define 
\begin{gather*}
\delta_{m_0\,\eps} := \eps^{\frac{1}{(m_0+1)^2}}
\quad \text{and} \quad 
\delta_{m_0-l \,\eps} := \max \Big\{  \Big( \Big( \prod_{m_0-l+1 \le i < k} \delta_{i\,\eps} \Big)^{-1} \, \eta_{n}^+(\delta_{k\,\eps}) \Big)^{\frac 1 2}, \\
\qquad \qquad \Big( \Big( \prod_{m_0-l+1 \le i < k} \delta_{i\,\eps} \Big)^{-1} \, \eta_{n}^-(\delta_{k\,\eps}) \Big)^{\frac 1 2}:  k \in \{ m_0-l+1, \dots, m_0 \}  \Big\} \cup \Big\{ \eps^{\frac{1}{(m_0+1)^2}} \Big\}
\end{gather*}
for $l \in \{ 1, \dots, m_0-1\}$. With the help of~\eqref{eq: wesentl grund für ex der eps_i} and~\eqref{eq: wesentl grund 2 für ex der eps_i} it then successively follows, by proceeding from larger to smaller indices $i$, that $\delta_{i\,\eps} \longrightarrow 0$ as $\eps \searrow 0$ for all $i \in \{1, \dots, m_0\}$ 
and that~\eqref{eq: wunsch 1 an eps_i}, \eqref{eq: wunsch 2 an eps_i} and~\eqref{eq: wunsch 3 an eps_i} are satisfied. Assertion~(ii) 
now follows from~\eqref{eq: gl 4, adsatz ohne sl}, \eqref{eq: gl 5, adsatz ohne sl}, \eqref{eq: gl 6, adsatz ohne sl} by virtue of~\eqref{eq: absch B_n eps}, \eqref{eq: absch B_n eps'}, \eqref{eq: absch C_n eps^+}, \eqref{eq: absch C_n eps^-} and Proposition~\ref{prop: störreihe für gestörte zeitentw}. 
\end{proof}

Some remarks, which in particular clarify the relation of the above theorem with the adiabatic theorem without spectral gap condition from~\cite{AvronGraf11} and~\cite{dipl}, are in order.
\smallskip

1. Clearly, the adiabatic theorem above generalizes the adiabatic theorems without spectral gap condition from~\cite{AvronGraf11} (Theorem~11) and~\cite{dipl} (Theorem~6.4) 
which cover the less general case of weakly semisimple eigenvalues $\lambda(t)$ of (not necessarily skew self-adjoint) operators $A(t): D \subset X \to X$ under less general regularity conditions.
See Section~4.3 for simple examples where the previously known adiabatic theorems cannot be applied -- and which show, moreover, 
that the adiabatic theorem above is by no means confined to spectral operators. 
%
%
%
In the special case where the eigenvalues $\lambda(t)$ from the above theorem lie on the imaginary axis $i \R$ for every $t \in I$, 
these eigenvalues are automatically weakly semisimple by the $(M,0)$-stability hypothesis of the theorem,  which by 
Joye's example (from the end of Section~3.3) cannot be essentially weakened, and by the weak associatedness hypothesis. 
(Argue as in the second remark at the beginning of Section~3.3.) 
And so, the above adiabatic theorem -- in the special case of purely imaginary eigenvalues -- essentially reduces to the adiabatic theorems without spectral gap condition from~\cite{AvronGraf11} and~\cite{dipl} -- but a general adiabatic theory without spectral gap condition should be able to cover more than just this special case, of course.
%
%
%
%
\smallskip

%
2. As can be seen from the proof of the theorem above, one would obtain the same conclusion if -- instead of requiring the weak associatedness of $P(t)$ with $A(t)$ and $\lambda(t)$ for almost every $t \in I$ -- one only required the inclusions
\begin{align} \label{eq: blosse inklusionen}
P(t)X \subset \ker(A(t)-\lambda(t))^{m_0} \quad \text{and} \quad (1-P(t))X \subset \overline{ \rg(A(t)-\lambda(t))^{m_0} }
\end{align} 
for almost every $t \in I$ 
and some $m_0 \in \N$. 
We point out, however, that one would 
nevertheless not obtain a 
truly more general theorem by thus modifying the hypotheses, because the projections $P(t)$ would then still 
be weakly associated with $A(t)$ and $\lambda(t)$ for almost every $t \in I$. 
(In order to see this, notice that, for every $t \in I$ where~\eqref{eq: blosse inklusionen} is satisfied,  $(1-P(t))\ker(A(t)-\lambda(t))^{m_0} = 0$ by Lemma~\ref{lm: lm 1 zum erw adsatz ohne sl} and an expansion similar to~\eqref{eq: ascent = 1 für erz beschr halbgr} and that $P(t) \rg(A(t)-\lambda(t))^{m_0} = 0$ so that $P(t) \, \ol{\rg(A(t)-\lambda(t))^{m_0}} = 0$ as well. And then apply the first remark after Theorem~\ref{thm: typ mögl für PX und (1-P)X}.)
Also it can be seen from the proof of the above theorem: if the finite rank hypothesis on $P(0)$ is the only one to be violated, one still has the strong convergence
\begin{align*}
\sup_{t \in I} \norm{ \big( U_{\eps}(t) - V_{0\,\eps}(t) \big) P(0)x } \longrightarrow 0 \quad (\eps \searrow 0) \quad \text{for every } x \in X, 
\end{align*}
provided $P(t)$ is even weakly $1$-associated 
with $A(t)$ and $\lambda(t)$ for almost every $t \in I$.
(In order to see this, notice that, under this extra 
condition, the inclusion $P(t)X \subset \ker(A(t)-\lambda(t))$ holds for every $t \in I$ 
by a closedness argument similar to the one in~\eqref{eq: P(t) vertauscht mit A(t) für alle t} and the $\eps$-dependence of $V_{0\,\eps}(s)P(0)$ is solely contained in a scalar factor, 
\begin{align*}
V_{0\,\eps}(s)P(0) = e^{\frac 1 \eps \int_0^s \lambda(\tau) \, d\tau} \, W(s)P(0) \quad (s \in I),
\end{align*}
where $W$ denotes the evolution system for $[P',P]$.) 
\smallskip

3. As in the case with spectral gap, the adiabatic theorem without spectral gap condition above can -- along with the approximate commutator equation method used in its proof -- be extended to several eigenvalues $\lambda_1(t)$, \dots, $\lambda_r(t)$:
if $A$, $\lambda_j$, $P_j$ for every $j \in \{1, \dots, r\}$ satisfy the hypotheses of part~(ii) of the above adiabatic theorem and if $\lambda_j(\,.\,)$ and $\lambda_l(\,.\,)$ for all $j \ne l$ fall into each other at only countably many points accumulating at only finitely many points, then the evolution system $V_{\eps}$ for $\frac 1 \eps A + K$ with $K$ as in~\eqref{eq: K mehrere sigma_j} is adiabatic w.r.t.~all the $P_j$ and well approximates the evolution system $U_{\eps}$ for $\frac 1 \eps A$ in the sense that
\begin{align*}
\sup_{t \in I} \norm{ U_{\eps}(t) - V_{\eps}(t) } \longrightarrow 0 \quad (\eps \searrow 0),
\end{align*}
provided $V_{\eps}$ exists on $D$.
In order to see this (in the technically simpler case where the $t \mapsto P_j(t)$ are twice continuously SOT-differentiable), one sets $B_{\bm{\delta}}(t) := \frac 1 2 \sum_{j=1}^{r+1} B_{j \, \bm{\delta}}(t)$ where
\begin{gather}
B_{j \, \bm{\delta}} := \sum_{k=0}^{m_0-1} \Big( \prod_{i=1}^{k+1} \ol{R}_{j \, \delta_i} \Big) P_j' P_j (\lambda_j-A)^k + \sum_{k=0}^{m_0-1} (\lambda_j-A)^k P_j P_j' \Big( \prod_{i=1}^{k+1} \ol{R}_{j \, \delta_i} \Big) \;\, (j \in \{1, \dots, r\}) \notag \\
\text{with \, } \ol{R}_{j \, \delta} := (\lambda_j + \delta e^{i \vartheta_j} - A)^{-1} (1-P_j) \text{ \, and where} \notag \\
B_{r+1 \, \bm{\delta}} := \sum_{j=1}^r B_{j \, \bm{\delta}} + \sum_{j \ne l} B_{j \, l}, \notag \\
B_{j \, l} := \sum_{i,k = 0}^{m_0-1} \binom{k+i}{i} \frac{(-1)^{i}}{(\lambda_l-\lambda_j)^{k+i+1}} \Big( (A-\lambda_j)^k P_j P_j' P_l (A-\lambda_l)^i \notag \\
+ (A-\lambda_l)^{i} P_l P_j' P_j (A-\lambda_j)^k \Big) \notag
\end{gather}
and then one shows that the approximate commutator equation $B_{\bm{\delta}}(t)A(t) - A(t)B_{\bm{\delta}}(t) \subset K(t) - C_{\bm{\delta}}(t)$ with $C_{\bm{\delta}}(t) := \sum_{j=1}^r C_{j \, \bm{\delta}}(t)$ and
\begin{align*}
C_{j \, \bm{\delta}} :=  \sum_{k=0}^{m_0-1} \delta_{k+1} e^{i \vartheta_j} \Big( \prod_{i=1}^{k+1} \ol{R}_{j \, \delta_i} \Big) P_j' P_j (\lambda_j-A)^k 
- \sum_{k=0}^{m_0-1} (\lambda_j-A)^k P_j P_j' \, \delta_{k+1} e^{i \vartheta_j} \Big( \prod_{i=1}^{k+1} \ol{R}_{j \, \delta_i} \Big)
\end{align*}
is satisfied for all points $t$ where $P_j(t)$ is weakly associated with $A(t)$, $\lambda_j(t)$ for every $j \in \{1, \dots, r\}$ and where no crossing between two curves $\lambda_j(\,.\,)$ and $\lambda_l(\,.\,)$ takes place -- 
by showing that for all such $t$ and all $j \ne l$ 
\begin{align*}
B_{j \, l}(t) A(t) - A(t) B_{j \, l}(t) \subset P_j(t)P_j'(t)P_l(t) - P_l(t)P_j'(t)P_j(t) = [P_j'(t), P_l(t)]
\end{align*}
(for the second equality, use that $P_j(t)P_l(t) = 0 = P_l(t)P_j(t)$ which is due to the weak associatedness of $P_j(t)$, $P_l(t)$ with $A(t)$, $\lambda_j(t)$ resp.~$\lambda_l(t)$ 
and to $\lambda_j(t) \ne \lambda_l(t)$) 
and by noticing that $[P_{r+1}',P_{r+1}] = [P',P]$. 
At first glance, the defining formula for 
$B_{r+1 \, \bm{\delta}}$ might seem a bit mysterious, but in fact 
it can be guessed, just like 
the formulas for $B_{1 \, \bm{\delta}}$, \dots, $B_{r \, \bm{\delta}}$, 
from the case with spectral gap: indeed, by rewriting the formula for $B_{r+1}$ from~\eqref{eq: B_{r+1}} -- in the special case of singletons $\sigma_j(t) = \{ \lambda_j(t) \}$ consisting of poles of $(\,.\,-A(t))^{-1}$ of order at most $m_0$ -- one obtains $B_{r+1} = B_{r+1 \, \bm{\delta}} \big|_{\bm{\delta} = 0}$ with the help of Cauchy's theorem. 

\subsection{A quantitative adiabatic theorem without spectral gap condition -- especially for scalar type spectral operators}

As a supplement to the qualitative adiabatic theorem above (Theorem~\ref{thm: erw adsatz ohne sl}), we note the following quantitative refinement. 
It implies 
that, if in the situation of the above theorem 
the map $t \mapsto P(t)$ is even $W^{2,\infty}_*$-regular, then the rate of convergence  (Lemma~\ref{lm: lm 1 zum erw adsatz ohne sl}!) 
of the integrals
\begin{align} \label{eq: eta_0}
&\eta^+(\delta) := \int_0^1 \norm{\delta \big( \lambda(s)+\delta e^{i \vartheta(s)} - A(s) \big)^{-1} P'(s)P(s)} \, ds, \notag \\
&\qquad \qquad \qquad \eta^-(\delta) := \int_0^1 \norm{P(s)P'(s) \delta \big( \lambda(s)+\delta e^{i \vartheta(s)} - A(s) \big)^{-1} } \, ds 
\end{align}
yields a simple upper bound on the rate of convergence of $\sup_{t \in I} \norm{ U_{\eps}(t)-V_{\eps}(t) }$ which we are interested in here. 
%
See~\cite{Teufel01} for an analogous result in the case of skew self-adjoint operators $A(t)$.

%

\begin{thm} \label{thm: erw adsatz ohne sl, quantitativ}
Suppose that 
$A(t)$, $\lambda(t)$, $P(t)$ are as in Theorem~\ref{thm: erw adsatz ohne sl} with $X$ not necessarily reflexive and that $t \mapsto P(t)$ is even in $W^{2,\infty}_*(I,L(X))$. Suppose further that $\eta: (0,\delta_0] \subset (0,1] \to (0,\infty)$ is a function such that $\eta(\delta) \longrightarrow 0$ as $\delta \searrow 0$ and 
\begin{align*}
\eta(\delta) \ge \delta \quad \text{as well as} \quad \eta^{\pm}(\delta) \le \eta(\delta)
\end{align*} 
for all $\delta \in (0,\delta_0]$ with $\eta^{\pm}$ as above. 
Then there is a constant $c$ such that
\begin{align*}
\sup_{t \in I} \norm{ U_{\eps}(t)-V_{\eps}(t) } \le c \, \tilde{\eta}^{m_0} \big( \eps^{2/(m_0(m_0+1))} \big) = c \, ( \tilde{\eta} \circ \dotsb \circ \tilde{\eta} ) \big( \eps^{2/(m_0(m_0+1))} \big)
\end{align*}
for $\eps$ sufficiently small, where $\tilde{\eta}(\delta) := \eta( \delta^{\frac{1}{2}} )$.
\end{thm}

\begin{proof}
We proceed as in the proof of the qualitative adiabatic theorem above, but now replace $Q_n$ and $Q_n'$ at any occurrence by $P'$ and $P''$. We can then conclude from~\eqref{eq: gl 4, adsatz ohne sl} and~\eqref{eq: gl 6, adsatz ohne sl} (with the replacements just mentioned) 
that there is a constant $c'$ such that
\begin{align}  \label{eq: absch interessierender ausdruck}
%
&\sup_{t \in I} \norm{ U_{\eps}(t)-V_{\eps}(t) } \le \notag \\
& \quad \,\, c' \, \bigg(    \sum_{k=1}^{m_0} \eps \Big( \prod_{j=1}^{m_0} \delta_j \Big)^{-1} 
+ \sum_{k=1}^{m_0} \eps \Big( \delta_k^{m_0+1} \, \prod_{j \ne k} \delta_j \Big)^{-1} \eta(\delta_k) 
+ \sum_{k=1}^{m_0} \Big( \prod_{1 \le j < k} \delta_j \Big)^{-1} \eta(\delta_k) \bigg)
\end{align}
for all $\delta_1, \dots, \delta_{m_0} \in (0,\delta_0]$ and $\eps \in (0,\infty)$. In this estimate the first, second, and third sum correspond to the $B_{\bm{\delta}}$-, $B_{\bm{\delta}}'$-, $C_{\bm{\delta}}$-terms in~\eqref{eq: gl 6, adsatz ohne sl}, respectively. (See~\eqref{eq: absch B_n eps} and~\eqref{eq: absch C_n eps^+}, \eqref{eq: absch C_n eps^-} for the estimation of the $B_{\bm{\delta}}$-terms and $C_{\bm{\delta}}$-terms. And to obtain the upper bound for the $B_{\bm{\delta}}'$-terms, refine the estimate on $\int_0^1 \norm{ B_{\bm{\delta}}'(s)} \, ds$ 
from the proof of the previous theorem by using -- instead of~\eqref{eq: absch B_n eps'} -- the fact that 
\begin{gather*}
\int_0^1 \norm{P(s) P'(s) \ol{R}_{\delta}(s)} \, ds, \,\, \int_0^1 \norm{(A(s)-1) \ol{R}_{\delta}(s) P'(s)P(s)} \, ds \le c \, \frac{\eta(\delta)}{\delta} 
\end{gather*}
and that $\sup_{s \in I} \norm{R_{\delta}(s)}, \sup_{s \in I} \norm{(A(s)-1)R_{\delta}(s)} \le \frac{c}{\delta^{m_0}}$ for all sufficiently small $\delta \in (0,\delta_0]$.)
We now recursively define
\begin{align*}
\delta_{m_0 \, \eps} := \eps^{\frac{1}{m_0(m_0+1)}} 
\quad \text{and} \quad 
\delta_{m_0-k \, \eps} := \big( \eta(\delta_{m_0-k+1 \, \eps}) \big)^{\frac{1}{2}}
\end{align*}
for $\eps$ so small that $\delta_{m_0-k+1 \, \eps}$ lies in $(0,\delta_0]$ and for $k \in \{1, \dots, m_0-1\}$. (It should be noticed 
that $\delta_{m_0-k+1 \, \eps} \longrightarrow 0$ as $\eps \searrow 0$ because $\eta(\delta) \longrightarrow 0$ and that $\delta_{m_0-k+1 \, \eps}$ therefore really lies in the domain $(0,\delta_0]$ of $\eta$ for sufficiently small $\eps$.) 
Since 
$\eta(\delta_{1 \, \eps}) = \tilde{\eta}^{m_0} ( \eps^{2/(m_0(m_0+1))} )$ and $\frac{1}{\delta_{k-1 \, \eps}} \eta(\delta_{k \, \eps}) = \delta_{k-1 \, \eps} \le \eta(\delta_{k-1 \, \eps})$ for $k \in \{2, \dots, m_0\}$,
it follows by induction that
\begin{align} \label{eq: absch 1}
\Big( \prod_{1 \le j < k} \delta_{j \, \eps} \Big)^{-1} \eta(\delta_{k \, \eps}) \le \tilde{\eta}^{m_0} \big( \eps^{2/(m_0(m_0+1))} \big)
\end{align}
and, in particular, $\eta(\delta_{k \, \eps}) \le \tilde{\eta}^{m_0} \big( \eps^{2/(m_0(m_0+1))} \big)$ for all $k \in \{1, \dots, m_0\}$ and sufficiently small $\eps$. Since 
$\delta_{m_0 \, \eps} \le \delta_{m_0-k+1 \, \eps} \le \delta_{m_0-k \, \eps}$ for $k \in \{1, \dots, m_0-1\}$ and small $\eps$, 
it further follows that
\begin{align} \label{eq: absch 2}
\eps \Big( \prod_{j=1}^{k} \delta_{j \, \eps} \Big)^{-1} \le \eps \Big( \delta_{k \, \eps}^{m_0+1}\prod_{j \ne k} \delta_{j \, \eps} \Big)^{-1} \eta(\delta_{k \, \eps}) 
&\le \eps \Big( \prod_{j=1}^{m_0} \delta_{m_0 \, \eps} \Big)^{-(m_0+1)} \tilde{\eta}^{m_0} \big( \eps^{2/(m_0(m_0+1))} \big) \notag \\
&= \tilde{\eta}^{m_0} \big( \eps^{2/(m_0(m_0+1))} \big)
\end{align}
for all $k \in \{1, \dots, m_0\}$ and sufficiently small $\eps$. 
Combining~\eqref{eq: absch interessierender ausdruck}, \eqref{eq: absch 1} and~\eqref{eq: absch 2} we finally obtain the assertion.
\end{proof}

1. 
Clearly, a function $\eta$ as described in the above theorem exists under the hypotheses of the qualitative adiabatic theorem (Theorem~\ref{thm: erw adsatz ohne sl}) including the reflexivity of $X$. In fact, one 
has only to define $\eta(\delta) := \max \{ \eta^+(\delta), \eta^-(\delta), \delta \}$ with $\eta_{\pm}$ as in~\eqref{eq: eta_0} and to remember that $\eta(\delta) \longrightarrow 0$ as $\delta \searrow 0$ by~\eqref{eq: wesentl grund für ex der eps_i} and~\eqref{eq: wesentl grund 2 für ex der eps_i}. 
\smallskip

2. 
An inspection of the proof above -- or, more precisely, of the arguments leading to~\eqref{eq: absch interessierender ausdruck} -- shows that one obtains the same conclusion if one drops the finite rank hypothesis on $P(0)$ and, at the same time, 
replaces the hypothesis that $P(t)$ be weakly associated with $A(t)$ and $\lambda(t)$ for almost every $t \in I$ by the condition that there exist 
an $m_0 \in \N$ such that
\begin{align*}
P(t)X \subset \ker(A(t)-\lambda(t))^{m_0}
\end{align*} 
for every $t \in I$ (while leaving all other hypotheses of the above theorem unchanged).
\bigskip

We now specialize to the case of spectral operators of scalar type which -- by what has been remarked in Section~2.1 -- automatically have weakly semisimple eigenvalues and which arise as the generic one-dimensional periodic Schrödinger operators (Remark~8.7 of~\cite{GesztesyTkachenko09}).
As a simple corollary of the adiabatic theorem above, we 
obtain the following quantitative adiabatic theorem tailored to scalar type spectral operators $A(t)$ 
whose spectral measures $P^{A(t)}$ are -- in some sense -- Hölder continuous in $t$.
It generalizes a result for skew self-adjoint $A(t)$ of Avron and Elgart (Corollary~1 in~\cite{AvronElgart99}) and a refinement of it due to Teufel (Remark~1 in~\cite{Teufel01}) and slightly improves the rates of convergence given there.

\begin{cor} \label{prop: quant adsatz für hoelderstet spektrmass}
Suppose $A(t): D \subset X \to X$ for every $t \in I$ is a spectral operator of scalar type (with spectral measure $P^{A(t)}$) such that Condition~\ref{cond: reg 1} is satisfied with $\omega = 0$ and such that $\sup_{t \in I} \sup_{E \in \mathcal{B}_{\C}} \norm{ P^{A(t)}(E) } < \infty$. 
Suppose further that $\lambda(t)$ for every $t \in I$ is an eigenvalue of $A(t)$ 
such that the open sector
\begin{align*}
\lambda(t) + \delta_0 \, S_{(\vartheta(t)-\vartheta_0, \vartheta(t) + \vartheta_0)} := \big\{ \lambda(t) + \delta e^{i \vartheta}: \delta \in (0,\delta_0), \vartheta \in (\vartheta(t)-\vartheta_0, \vartheta(t) + \vartheta_0) \big\}
\end{align*} 
of radius $\delta_0 \in (0,\infty)$ and angle $2 \vartheta_0 \in (0, \pi)$ for every $t \in I$ is contained in 
$\rho(A(t))$ 
and such that $t \mapsto \lambda(t)$, $e^{i \vartheta(t)}$ are Lipschitz continuous. 
Suppose finally that $P(t)$ for every $t \in I$ is a bounded projection in $X$ such that $P(t) = P^{A(t)}(\{ \lambda(t) \})$ for almost every $t \in I$ and $t \mapsto P(t)$ is in $W^{2,\infty}_*(I,L(X))$, and suppose that 
$P^{A(t)}$ is Hölder continuous with exponent $\alpha \in (0,1]$ uniformly in $t \in I$, that is, there is a $c_0 \in (0, \infty)$ such that 
\begin{align*}
\big\| P^{A(t)}(E)x \big\| \le c_0 \, \lambda(E)^{\frac{\alpha}{2}} \, \norm{x}
\end{align*}
for all $E \in \mathcal{B}_{\C}$ with $\lambda(E) \le 1$ and all $x \in X$ and $t \in I$.
Then there is a constant 
$c \in (0,\infty)$ such that
\begin{align*}
\sup_{t \in I} \norm{ U_{\eps}(t) - V_{\eps}(t) } \le c \; \eps^{ \frac{\alpha}{2(1+\alpha)} }
\end{align*}
for small enough $\eps \in (0, \infty)$, where $V_{\eps}$ denotes the evolution system for $\frac{1}{\eps} A + [P',P]$.
\end{cor}

\begin{proof}
We show that there exists a function $\eta: (0, \delta_0'] \to (0,\infty)$ such that $\eta(\delta) \longrightarrow 0$ as $\delta \searrow 0$ and 
\begin{align} \label{eq: gl 1, quant für normal}
\eta(\delta) \ge \delta \quad \text{as well as} \quad \norm{\delta R_{\delta}(t)} = \norm{ \delta \big( \lambda(t)+\delta e^{i \vartheta(t)} - A(t) \big)^{-1} } \le \eta(\delta)
\end{align}
for all $\delta \in (0,\delta_0']$ and $t \in I$ (with a suitable $\delta_0'$). 
In order to do so, we notice that, by the scalar type spectrality of $A(t)$ and Theorem~XVIII.2.11 of~\cite{DunfordSchwartz},  
\begin{align*}
\big| \scprd{x^*, \delta R_{\delta}(t)x } \big|
\le \int_{\sigma(A(t))} \frac{\delta}{ | \lambda(t) + \delta e^{i \vartheta(t)} - z | } \,\, d \big| P^{A(t)}_{x^*,x} \big|(z),
\end{align*}
where $\big| P^{A(t)}_{x*,x} \big|$ denotes the total variation of the complex measure $E \mapsto P^{A(t)}_{x^*,x}(E) := \scprd{x^*, P^{A(t)}(E)x}$, 
and then divide the spectrum $\sigma(A(t))$ of $A(t)$ 
into the parts
\begin{align*}
\sigma_{1 \, r_{\delta}}(t) := \sigma(A(t)) \cap U_{r_{\delta}}(\lambda(t)) \quad \text{and} \quad \sigma_{2 \, r_{\delta}}(t) := \sigma(A(t)) \cap \C \setminus U_{r_{\delta}}(\lambda(t))
\end{align*} 
of those spectral values that are close to $\lambda(t)$ resp.~far from $\lambda(t)$,
where $r_{\delta} := \delta^{\gamma}$ and $\gamma \in (0,1)$ will be chosen in~\eqref{eq: gl 4, quant für normal} below. 
Since, by Lemma~III.1.5 of~\cite{DunfordSchwartz}, 
\begin{align*}
\big| P^{A(t)}_{x^*,x} \big|(E) \le 4 \sup_{F \in \mathcal{B}_E} \big| \langle x^*, P^{A(t)}(F) P^{A(t)}(E)x \rangle \big|  
\le 4 M' \norm{x^*} \big\| P^{A(t)}(E)x \big\|
\end{align*}
for every $t \in I$ and $E \in \mathcal{B}_{\C}$ (where $M' := \sup_{t \in I} \sup_{F \in \mathcal{B}_{\C}} \norm{ P^{A(t)}(F) } < \infty$) and since, by the assumed sector condition,
\begin{align*}
\operatorname{dist}\big( \lambda(t)+\delta e^{i \vartheta(t)}, \sigma(A(t)) \big) \ge (\sin \vartheta_0) \, \delta
\end{align*}
for every $t \in I$ and $\delta \in (0,\delta_0']$ (where $\delta_0'$ is chosen small enough), there are positive constants $c_1$, $c_2$ such that
\begin{gather*}
\int_{\sigma_{1 \, r_{\delta}}(t)} \frac{\delta}{ | \lambda(t) + \delta e^{i \vartheta(t)} - z | } \,\, d \big|P^{A(t)}_{x^*,x} \big|(z) 
\le \frac{1}{\sin \vartheta_0} \, \big| P^{A(t)}_{x^*,x} \big| \big( U_{r_{\delta}}(\lambda(t)) \big) 
\le c_1 \delta^{\alpha \, \gamma} \norm{x^*} \norm{x} \\ 
\text{as well as} \\
\int_{\sigma_{2 \, r_{\delta}}(t)} \frac{\delta}{ | \lambda(t) + \delta e^{i \vartheta(t)} - z | } \,\, d \big|P^{A(t)}_{x^*,x}\big| (z)
\le \frac{\delta}{ r_{\delta}-\delta } \, \big| P^{A(t)}_{x^*,x} \big|(\C)  
\le c_2 \delta^{1-\gamma} \norm{x^*} \norm{x}  
\end{gather*}
for every $x \in X$, $x^* \in X^*$, $\delta \in (0,\delta_0']$ and $t \in I$. Consequently,
\begin{align}
\norm{\delta R_{\delta}(t) } 
\le c_1 \, \delta^{ \alpha \, \gamma} + c_2 \, \delta^{ 1-\gamma } 
\le \max \{ c_1, c_2 \} \, \delta^{ \min \{ \alpha \, \gamma ,  1-\gamma \} } = c_0' \, \delta^{ \beta(\gamma) } 
\end{align}
for every $t \in I$ and $\delta \in (0, \delta_0']$ (notice that $\beta(\gamma) := \min \{ \alpha \, \gamma, 1-\gamma \}$, for given $\gamma$, is the best -- that is, biggest -- 
possible exponent in the second inequality above). 
And as $\gamma \mapsto \beta(\gamma)$ is maximal at $\gamma_0 := \frac{1}{1+\alpha}$, we choose 
\begin{align} \label{eq: gl 4, quant für normal}
\gamma := \gamma_0, \quad \beta := \beta(\gamma_0) = \frac{\alpha}{1+\alpha}, \quad \eta(\delta) := c_0' \, \delta^{\beta} = c_0' \, \delta^{ \frac{\alpha}{1+\alpha} },
\end{align}
thereby obtaining \eqref{eq: gl 1, quant für normal}. Since $P(t)A(t) \subset A(t)P(t) = \lambda(t)P(t)$ holds for every $t \in I$ (by the scalar type spectrality of $A(t)$ and the closedness argument in the remark after Theorem~\ref{thm: handl adsatz mit nichtglm sl}), 
the desired conclusion follows by the second remark after Theorem~\ref{thm: erw adsatz ohne sl, quantitativ} with $m_0 = 1$.
\end{proof}

\subsection{Some examples}

We begin with two examples where $\lambda(t)$ is an eigenvalue of $A(t)$ that is allowed to 
be non-isolated and non-weakly-semisimple for every $t \in I$. In particular, these examples cannot be dealt with by way 
of the previously known adiabatic theorems. In the first example, $A(t)$ is a spectral operator 
whereas in the second 
it is not (by Theorem~XV.3.10 and
~XV.8.7 of~\cite{DunfordSchwartz} 
and by the spectral structure of the right shift $S_+$ (Section~2.4)). 

\begin{ex} \label{ex: A_2(t) diagb}
Suppose $A$, $\lambda$, $P$ with $A(t) = R(t)^{-1} A_0(t) R(t)$, $P(t) = R(t)^{-1} P_0 R(t)$, and $R(t) = e^{C t}$ are given as follows in $X := \ell^p(I_d) \times \ell^p(I_{\infty})$ (where $p \in [1,\infty)$ and $d \in \N$): 
\begin{align*}
A_0(t) := 
\begin{pmatrix} \lambda(t) + \alpha(t) N  & 0 \\ 0 & \operatorname{diag}\big( (\lambda_n)_{n \in \N} \big) \end{pmatrix} 
\quad \text{and} \quad
P_0 := \begin{pmatrix} 1 & 0 \\ 0 & 0 \end{pmatrix},
\end{align*} 
where $\lambda(t) \in (-\infty, 0]$, $\alpha(t)$, $N$ are such that Condition~\ref{cond: baustein mit nicht-halbeinfachem ew} is satisfied and where $(\lambda_n)_{n \in \N}$ is an enumeration of $[-1,0] \cap \Q$ such that $\lambda(t) \notin \{ \lambda_n: n \in \N \}$ for almost every $t \in I$. 
Additionally, suppose $t \mapsto \lambda(t)$ and $t \mapsto \alpha(t)$ are Lipschitz continuous and $C$ is the right shift operator on $\ell^p(I_d) \times \ell^p(I_{\infty}) \cong \ell^p(I_{\infty})$:
\begin{align*}
C(z_1, \dots, z_d, z_{d+1}, \dots ) := (0, z_1, \dots, z_{d-1}, z_d, \dots ).
\end{align*} 
Then $t \mapsto A(t)$ is in $W^{1,\infty}_*(I,L(X))$ and $t \mapsto A_0(t)$ is $(M_0,0)$-stable (by Lemma~\ref{lm: char (M,0)-stab für einfaches A}), so that $A$ is $(M,0)$-stable for some $M \in [1,\infty)$ by Lemma~\ref{lm: (M,w)-stabilität und ähnl.trf.}. Since $A_0(t)|_{P_0 X} - \lambda(t)$ 
is nilpotent of order at most $m_0 := \operatorname{dim} \ell^p(I_d) = d$ for every $t \in I$ 
and since $A_0(t)|_{(1-P_0)X}-\lambda(t)$ is injective and has dense range in $(1-P_0)X$ 
(because $\lambda(t) \notin \{ \lambda_n: n \in \N \}$) 
for almost every $t \in I$, 
$P_0$ is weakly $m_0$-associated with $A_0(t)$ and $\lambda(t)$, whence the same is true for $A(t)$, $P$ instead of $A_0(t)$ and $P_0$.
And finally, the resolvent estimate of Theorem~\ref{thm: erw adsatz ohne sl} 
is clearly fulfilled if we choose $\vartheta(t) := \frac{\pi}{2}$ for all $t \in I$. 
All 
other hypotheses of 
Theorem~\ref{thm: erw adsatz ohne sl}~(i) are obvious. 
$\blacktriangleleft$
\end{ex}

In the above example, we have chosen $C$ to be the right shift operator on $X = \ell^p(I_d) \times \ell^p(I_{\infty})$ in order to make sure that the example cannot be reduced to a finite-dimensional subspace: there is no finite-dimensional subspace $M$ of $X$ such that 
\begin{align*}
M \supset P(0)X \text{ and } 
A(t)M \subset M \text{ as well as } P(t)M \subset M 
\end{align*}
for every $t \in I$.
%
(Clearly, if 
given data $A$, $\lambda$, $P$ 
can be reduced to a finite-dimensional subspace $M$ and satisfy the hypotheses of Theorem~\ref{thm: erw adsatz ohne sl}~(i), 
it suffices to prove the respective statement for the reduced data $A^M$, $\lambda$, $P^M$ 
given by
\begin{align*}
A^M(t) := A(t)|_M \quad \text{and} \quad P^M(t) := P(t)|_M \quad (t \in I).
\end{align*}
And this, in turn, can typically already be done with the help of the adiabatic theorem with non-uniform spectral gap condition of Section~3: 
by the finite-dimensionality of $M$, $\lambda(t)$ is isolated in $\sigma(A^M(t))$ for every $t \in I$ and, by Theorem~\ref{thm: typ mögl für PX und (1-P)X} and the first remark following it, $P^M(t)$ is the projection associated with $A^M(t)$ and $\lambda(t)$ for almost every $t \in I$.)  
%
In order to see the claimed irreducibility of the example above, 
we assume, on the contrary, that $A$, $\lambda$, $P$ as given by the example can be reduced to a finite-dimensional subspace $M$ of $X$. 
Then $R(t)M$ is invariant under $A_0(t)$ for every $t \in I$ and hence (by finite-dimensional spectral theory applied to 
$A^|_0(t) := A_0(t)|_{R(t)M}$)  
\begin{align*}
%
R(t)M \quad = \bigoplus_{\lambda \in \sigma_p(A^|_0(t))} \bigcup_{k \in \N} \ker( A^|_0(t) - \lambda)^k 
\quad \subset 
\bigoplus_{\lambda \in \sigma_p(A^|_0(t))} \bigcup_{k \in \N} \ker( A_0(t) - \lambda)^k,
\end{align*}
which latter space (by the special choice of $A_0(t)$) for every $t \in I$ contains only vectors with finitely many non-zero components. Consequently, the same is also true 
for the vectors in $R(t)M$ for every $t \in I$. 
We now obtain the desired contradiction by observing that 
the vector $R(t)v = e^{C t}v$, for every $t \ne 0$ and every vector $0 \ne v \in X$ 
with only finitely many non-zero components, has infinitely many non-zero components because 
\begin{align*}
e^{C t} e_i = \big( 0, \dots, 0, 1, \frac{t}{1!}, \frac{t^2}{2!}, \frac{t^3}{3!}, \dots \big)
\end{align*}
for $i \in \N$ (where the entry $1$ is in the $i$th place).
\smallskip

\begin{ex} \label{ex: A_2(t) nicht diagb}
Suppose $A$, $\lambda$, $P$ with $A(t) = R(t)^{-1} A_0(t) R(t)$, $P(t) = R(t)^{-1} P_0 R(t)$, and $R(t) = e^{C t}$ are given as follows in $X := \ell^p(I_d) \times \ell^p(I_{\infty})$ (where $p \in (1,\infty)$ and $d \in \N$): 
\begin{align*}
A_0(t) := 
\begin{pmatrix} \lambda(t) + \alpha(t) N  & 0 \\ 0 & S_+ - 1 \end{pmatrix} 
\quad \text{and} \quad
P_0 := \begin{pmatrix} 1 & 0 \\ 0 & 0 \end{pmatrix},
\end{align*} 
where $\lambda(t) \in \partial U_1(-1)$, $\alpha(t)$, $N$ are such that Condition~\ref{cond: baustein mit nicht-halbeinfachem ew} is satisfied. Additionally, $t \mapsto \lambda(t)$ and $t \mapsto \alpha(t)$ are Lipschitz continuous and $C$ is the bounded linear map in $\ell^p(I_d) \times \ell^p(I_{\infty}) \cong \ell^p(I_{\infty})$ given by 
\begin{align*}
C (z_1, \dots, z_d, z_{d+1}, \dots ) := (0, \dots, 0, z_{d+1}, -z_d, 0, \dots),
\end{align*}
where in the vector on the right $z_{d+1}$, $-z_d$ appear in the $d$th and $(d+1)$th place. 
Since $\lambda(t) \in \partial U_1(-1) = \sigma_c(S_+ - 1)$ 
for every $t \in I$ because $p \ne 1$ (Section~2.4), 
$P_0$ is weakly associated with $A_0(t)$ and $\lambda(t)$ and therefore the same goes for $A_0(t)$, $P_0$ replaced by $A(t)$ and $P(t)$.
Also, if for every $t \in I$ we choose $\vartheta(t)$ such that $\lambda(t) = -1 + e^{i \vartheta(t)}$, 
then the resolvent estimate of Theorem~\ref{thm: erw adsatz ohne sl} holds true 
because
\begin{align*}
\norm{  \big( \lambda(t) + \delta e^{i \vartheta(t)} - A_0(t) \big)^{-1} (1-P_0) } \le \norm{  \big( 1 + \delta - e^{-i \vartheta(t)} S_+ \big)^{-1}  } \le \frac{1}{\delta}
\end{align*}
for every $t \in I$ and $\delta \in (0,\infty)$ 
(Section~2.4).
$\blacktriangleleft$
\end{ex}

Just like the first example, the example above cannot be reduced to a finite-dimensional subspace $M$. Indeed, 
assuming that $A$, $\lambda$, $P$ from above could, we obtain that $R(t)M$ for every $t \in I$ is invariant under $A_0(t)$ and
\begin{align*}
\emptyset \ne \sigma(A_0(t)|_{R(t)M}) = \sigma_p(A_0(t)|_{R(t)M}) \subset \sigma_p(A_0(t)) 
= \{ \lambda(t) \}
\end{align*}
(as $\sigma_p(S_+ - 1) = \emptyset$ (Section~2.4)),
so that (by finite-dimensional spectral theory) 
%
\begin{align*}
R(t)M = \ker(A_0(t)|_{R(t)M}-\lambda(t))^{\dim M} 
\subset \ker(A_0(t)-\lambda(t))^{\dim M} 
= P_0 X \subset M
\end{align*}
for all $t \in I$.
Consequently, $R(t)P_0X \subset R(t) M \subset P_0 X$. 
In other words, $P_0 X$ is invariant under $R(t) = e^{C t}$ for every $t \in I$, which is obviously not true by the choice of $C$ in the above example. Contradiction! 
We point out that the above example really works only in case $p \ne 1$. In fact, if $p = 1$ then $\lambda(t) \in \partial U_1(-1) = \sigma_r(S_+-1)$ (Section~2.4) from which it follows by the block structure of $A_0(t)$ that  
\begin{gather*}
%
\ker(A_0(t)-\lambda(t))^k 
\subset \ell^p(I_d) \times 0, \\ 
\overline{\rg(A_0(t)-\lambda(t))^k} 
\subset \ell^p(I_d) \times \overline{\rg(S_+ - 1 - \lambda(t))^k} 
\subsetneq \ell^p(I_d) \times \ell^p(I_{\infty}) 
\end{gather*}
and hence $\ker(A_0(t)-\lambda(t))^k + \overline{\rg(A_0(t)-\lambda(t))^k} \ne X$ for every $k \in \N$ and $t \in I$. Consequently, the adiabatic theorems proven above cannot be applied. 
\smallskip

In our last example we show that the conclusion of the adiabatic theorem without spectral gap condition may fail if the regularity hypothesis on $P$ 
is the only one to be violated. 

\begin{ex}  \label{ex: reg an P wesentl, ohne sl}
Set $A(t):= M_{f_t}$ in $X := L^p(\R)$ (for some $p \in [1,\infty)$), where 
\begin{align*}
f_t := f_0(\,.\, + t) \quad \text{with} \quad 0 \ne f_0 \in C_c^1(\R, i \R),
\end{align*}
$\lambda(t) := 0$ and $P(t) := M_{\chi_{E_t}}$ with $E_t := \{ f_t = 0 \}$. 
Then all the hypotheses of the adiabatic theorem without spectral gap condition -- in the version for projections of infinite rank (second remark after Theorem~\ref{thm: erw adsatz ohne sl}) -- are fulfilled with the sole exception that $t \mapsto P(t)$ is not SOT-continuously differentiable (by Lemma~\ref{lm: wenn P stet db, dann schon konst}). 
And indeed, the conclusion of the adiabatic theorem 
already fails: as the $A(t)$ are pairwise commuting and $t \mapsto f_t(x)$ is Riemann integrable for every $x \in \R$, one has
\begin{align*}
\big( U_{\eps}(t,s) g \big) (x) = \Big( e^{\frac 1 \eps \int_s^t A(\tau) \, d\tau} \, g \Big)(x) = e^{ \frac 1 \eps \int_s^t f_{\tau}(x) \, d\tau} \, g(x)
\end{align*}  
for almost every $x \in \R$ and therefore (by $f_0(\R) \subset i \R$)
\begin{align*}
\norm{ (1-P(t)) U_{\eps}(t) P(0) g }^p 
= \int \big| (1-\chi_{E_t}(x)) \chi_{E_0}(x) g(x) \big|^p \, dx
\end{align*}
for every $t \in I$, $\eps \in (0,\infty)$ and $g \in X$. Since the right hand side of this equation does not depend on $\eps \in (0, \infty)$ and since for every $t \in (0,1]$ there is a $g \in X$ such that this right hand side does not vanish, 
the conclusion 
of the adiabatic theorem without spectral gap 
-- more precisely, the weaker assertion that $\sup_{t \in I} \norm{(1-P(t))U_{\eps}(t)P(0)g} \longrightarrow 0$ for all $g \in X$ --
fails.
$\blacktriangleleft$
\end{ex}


It should be pointed out that the failure of both the hypotheses 
and the conclusion of the adiabatic theorems without spectral gap condition presented above is a quite typical phenomenon in the case where $A(t) = M_{f_t}$ in $X = L^p(X_0)$ for some $p \in [1,\infty)$ and some $\sigma$-finite measure space $(X_0, \mathcal{A}, \mu)$. 
Indeed, if $A(t) = M_{f_t}$ in $X = L^p(X_0)$ for measurable functions $f_t: X_0 \to \{ \Re z \le 0 \}$ such that $D(M_{f_t}) = D$ for all $t \in I$, 
if $\lambda(t)$ is an eigenvalue of $A(t)$, and if $P(t)$ for almost every $t \in I$ (with exceptional set $N$) is 
weakly associated with $A(t)$ and $\lambda(t)$, 
then 
\begin{align*}
P(t) = M_{\chi_{ \{ f_t = \lambda(t) \} }} = M_{\chi_{E_t}} \text{ for every } t \in I \setminus N 
\end{align*}
by Theorem~\ref{thm: typ mögl für PX und (1-P)X}, and therefore the following holds true.
As soon as $I \setminus N \ni t \mapsto P(t)$ is not constant, the hypotheses of the adiabatic theorem without spectral gap (Theorem~\ref{thm: erw adsatz ohne sl}) must fail (because then $I \setminus N \ni t \mapsto P(t) = M_{\chi_{E_t}}$ cannot extend to an SOT-continuously differentiable map 
by Lemma~\ref{lm: wenn P stet db, dann schon konst}). 
And as soon as, in addition, the maps $f_t$ are $i \R$-valued and $t \mapsto f_t g \in X$ is continuous for all $g \in D$, the conclusion 
of Theorem~\ref{thm: erw adsatz ohne sl}, or more precisely, of its corollary
\begin{align*}
\sup_{t \in I} \norm{(1-P(t))U_{\eps}(t)P(0)} \longrightarrow 0 \quad \text{and} \quad \sup_{t \in I} \norm{P(t)U_{\eps}(t)(1-P(0))} \longrightarrow 0,
\end{align*}
must fail as well. (In order to see this, 
one gathers from 
Theorem~2.3 of~\cite{NickelSchnaubelt98} (and its proof) that the evolution system $U_{\eps}$ for $\frac 1 \eps A$ exists on $D$ and can be strongly approximated by finite products of operators of the form $e^{M_{f_{\tau}} \, \sigma}$ with $\tau \in I$ and $\sigma \in [0,\infty)$, and infers from this that 
for arbitrary $g \in X$ 
\begin{align*}
\Big| (1-\chi_{E_t}(x)) \big( U_{\eps}(t) \chi_{E_0} g \big)(x) - \chi_{E_t}(x) \big( U_{\eps}(t) (1-\chi_{E_0}) g \big)(x) \Big| = 
\big| \chi_{E_t}(x) - \chi_{E_0}(x) \big| \big| g(x) \big|
\end{align*}
for almost every $x \in X_0$, 
whence
\begin{align*}
\big\| (1-P(t))U_{\eps}(t)P(0)g - P(t)U_{\eps}(t)(1-P(0))g \big\| = \norm{P(t)g - P(0)g}
\end{align*} 
for all $t \in I \setminus N$, $\eps \in (0,\infty)$. Since the right hand side of this equation does not depend on $\eps \in (0,\infty)$ and since $I \setminus N \ni t \mapsto P(t)$ is not constant, 
there is a $t \in (0,1]$ and a $g \in X$ such that $(1-P(t))U_{\eps}(t)P(0)g$ and $P(t)U_{\eps}(t)(1-P(0))g$ do not both converge to $0$ as $\eps \searrow 0$.)

\section{Adiabatic theorems for time-dependent domains}  \label{sect: adsätze für zeitabh domains}

In this section we extend the adiabatic theorems (with and without spectral gap condition) for time-independent domains of Section~3 and~4 
to the case of operators $A(t)$ with time-dependent domains -- by slightly modifying the proofs of the respective adiabatic theorems for time-independent domains.
%
Striving for 
such an extension is very natural because the requirement of constant domains is rather 
restrictive -- just think of differential operators $A(t)$ with (fully) time-dependent boundary conditions. 
%
We will see in Section~6 that the adiabatic theorems of this section allow one to almost effortlessly derive 
adiabatic theorems for operators $A(t) = i A_{a(t)}$ defined by symmetric sesquilinear forms $a(t)$.
%
All the theorems of this section are generalizations of the respective 
adiabatic theorems for time-independent domains if in these latter theorems all $W^{n,\infty}_*$-regulartity requirements 
are strengthened 
to $n$ times SOT-continuous differentiability requirements.
%
We will need 
the following very natural condition on $A$, which takes the role of Condition~\ref{cond: reg 1}.

\begin{cond} \label{cond: U_eps existiert und beschränkt}
$A(t): D(A(t)) \subset X \to X$ for every $t \in I$ is a densely defined closed linear map such that, for every $\eps \in (0,\infty)$, there is an evolution system $U_{\eps}$ for $\frac 1 \eps A$ on $D(A(t))$ and there is a constant $M \in [1,\infty)$ such that $\norm{U_{\eps}(t,s)} \le M$ for all $(s,t) \in \Delta$ and $\eps \in (0,\infty)$. 
\end{cond}

We point out 
that there is a huge number of papers establishing the existence of 
evolution systems $U$ for a given family $A$ of linear maps $A(t)$ on $D(A(t))$ as, for instance, \cite{Kato56}, \cite{Kisynski63}, \cite{Tanabe60}, \cite{KatoTanabe62}, \cite{FujieTanabe73}, \cite{AcquistapaceTerreni87}.
See the survey article~\cite{Schnaubelt02} for many more references. 
%
Instead of working with evolution systems on the spaces $Y_t = D(A(t))$ as in Condition~\ref{cond: U_eps existiert und beschränkt}, one could also prove adiabatic theorems operating -- as in~\cite{Kato70} or~\cite{Kato73} -- with evolution systems for $A$ on certain subspaces $Y$ of the intersection of all $D(A(t))$ (Definition~VI.9.2 of~\cite{EngelNagel}), 
but then one would have 
to impose various invariance conditions on the subspace $Y$, such as the $A(t)$-admissibiltity of $Y$, 
the invariance
\begin{align} \label{eq: sect 5, invarianz}
(z-A(t))^{-1} Y \subset Y
\end{align}
for $z \in \rg \gamma_t$ (case with spectral gap) or $z \in \{ \lambda(t) + \eps e^{i \vartheta(t)}: \eps \in (0,\eps_0] \}$ (case without spectral gap),
and the invariance of $Y$ under $P(t)$ and $P'(t)$. 
Such invariance conditions, however, are difficult 
to verify in practice:  \eqref{eq: sect 5, invarianz}, for instance, is clear only for complex numbers $z$ with sufficiently large positive real part 
(Proposition~2.3 of~\cite{Kato70}).

\subsection{Adiabatic theorems with spectral gap condition}

We will need the following condition depending on $m \in \{ 0 \} \cup  \N \cup \{\infty \}$ 
(number of points at which $\sigma(\,.\,)$ falls into $\sigma(A(\,.\,)) \setminus \sigma(\,.\,)$).

\begin{cond} \label{cond: vor adsatz mit sl}
$A(t): D(A(t)) \subset X \to X$ for every $t \in I$ is a linear map such that Condition~\ref{cond: U_eps existiert und beschränkt} is satisfied.
$\sigma(t)$ for every $t \in I$ is a compact 
subset of $\sigma(A(t))$, $\sigma(\,.\,)$ falls into $\sigma(A(\,.\,)) \setminus \sigma(\,.\,)$ at exactly $m$ points that accumulate at only finitely many points, and $I \setminus N \ni t \mapsto \sigma(t)$ is continuous, where $N$ denotes the set of those $m$ points at which $\sigma(\,.\,)$ falls into $\sigma(A(\,.\,)) \setminus \sigma(\,.\,)$. 
\begin{gather*}
J_{t_0} \ni t \mapsto (z-A(t))^{-1}  \text{ is SOT-continuously differentiable for all } z \in \rg \gamma_{t_0}, \\ \rg \gamma_{t_0} \ni z \mapsto \ddt{ (z-A(t))^{-1} } \text{ is SOT-continuous for all } t \in J_{t_0}, \\
\sup_{ (t,z) \in J_{t_0} \times \rg \gamma_{t_0} } \norm{ \ddt{ (z-A(t))^{-1} } } < \infty
\end{gather*}
for every $t_0 \in I \setminus N$, where the cycle $\gamma_{t_0}$ and the non-trivial closed interval $J_{t_0} \ni t_0$ are chosen such that $\rg \gamma_{t_0} \subset \rho(A(t))$ and $\operatorname{n}(\gamma_{t_0}, \sigma(t)) = 1$ and $\operatorname{n}(\gamma_{t_0}, \sigma(A(t)) \setminus \sigma(t)) = 0$ for every $t \in J_{t_0}$.
And finally, $P(t)$ is the 
projection associated with $A(t)$ and $\sigma(t)$ for every $t \in I \setminus N$ and $I \setminus N \ni t \mapsto P(t)$ extends to a twice SOT-continuously differentiable map on the whole of $I$.
\end{cond}

With this condition at hand, we can now formulate an adiabatic theorem with uniform spectral gap condition ($m = 0$) and non-uniform spectral gap condition ($m \in \N \cup \{\infty\}$) for time-dependent domains.

\begin{thm} \label{thm: adsatz mit sl, zeitabh}
Suppose $A(t)$, $\sigma(t)$, $P(t)$ for $t \in I$ are such that Condition~\ref{cond: vor adsatz mit sl} is satisfied with $m = 0$ or $m \in \N \cup \{ \infty\}$, respectively. Then
\begin{align*}
\sup_{t \in I} \norm{ U_{\eps}(t) - V_{\eps}(t) } = O(\eps)  \text{ resp. } o(1) \quad (\eps \searrow 0), 
\end{align*}
whenever the evolution system $V_{\eps}$ for $\frac 1 \eps A + [P',P]$ exists on $D(A(t))$ for all $\eps \in (0,\infty)$.
\end{thm}

\begin{proof}
We have only to prove the theorem in the case of a uniform spectral gap ($m = 0$), since the theorem in the case of a non-uniform spectral gap ($m \in \N \cup \{ \infty \}$) then follows in the same way as Theorem~\ref{thm: handl adsatz mit nichtglm sl} followed from Theorem~\ref{thm: handl adsatz mit glm sl}. In order to do so, 
we must only slightly modify the proof of Theorem~\ref{thm: handl adsatz mit glm sl}. 
We define the operators $B(t)$ as in the proof of that theorem (where now $\gamma_{t_0}$ and $J_{t_0}$ are given by Condition~\ref{cond: vor adsatz mit sl}), take over the first preparatory step of that proof, 
and easily show -- instead of what has been shown in the second preparatory step -- 
that $t \mapsto B(t)$ is SOT-continuously differentiable. (It has to be used for this last statement that Condition~\ref{cond: vor adsatz mit sl} implies
\begin{align*}
\sup_{ (t,z) \in J_{t_0} \times \rg \gamma_{t_0} } \norm{  (z-A(t))^{-1}  } < \infty 
\end{align*}
which can be seen as in the proof of~\eqref{eq: w.l.o.g. assumption} below.) 
We can then almost literally take over the main part of the proof of Theorem~\ref{thm: handl adsatz mit glm sl}: the only thing that has to be changed is that the $W^{1,\infty}$-regularity of 
\begin{align*}
[0,t] \ni s \mapsto U_{\eps}(t,s) B(s) V_{\eps}(s)x
\end{align*}
for $x \in D(A(0))$ can no longer be deduced from Lemma~\ref{lm: prod- und inversenregel}, but has to be inferred from Lemma~\ref{lm: prodregel rechtsseit db} and Lemma~\ref{lm: rechtsseit db und W^{1,infty}}, and that Proposition~\ref{prop: störreihe für gestörte zeitentw} has to be invoked for an $\eps$-independent bound on $V_{\eps}$.
\end{proof}

In general, the existence of the evolution system $V_{\eps}$ for $\frac 1 \eps A + [P',P]$ on $D(A(t))$ does not seem to be guaranteed under -- the fairly general -- Condition~\ref{cond: vor adsatz mit sl}. 
(In view of Proposition~\ref{prop: störreihe für gestörte zeitentw} one would, of course, like to define $V_{\eps}$ as a perturbation series and show that $[s,1] \ni t \mapsto V_{\eps}(t,s)y$ for every $y \in D(A(s))$ is a continuously differentiable solution to the initial value problem $x' = \frac 1 \eps A(t)x + [P'(t),P(t)]x$, $x(s) = y$, but this is not clear in general.)  
It is therefore good to know 
that under Condition~\ref{cond: vor adsatz mit sl} with $m=0$ 
one has at least the following statement:
\begin{align} \label{eq: zusatz zu adsatz mit sl, zeitabh}
\sup_{t \in I} \norm{ (1-P(t))U_{\eps}(t)P(0) }, \quad \sup_{t \in I} \norm{ P(t)U_{\eps}(t)(1-P(0)) } = O(\eps)  
\end{align}
as $\eps \searrow 0$, which follows from the adiabatic theorem of higher order (Theorem~\ref{thm: höherer adsatz}~(i) and~(iii) with degree of regularity $n = 1$) below. 
%
%
It should be pointed out, however, that Theorem~\ref{thm: adsatz mit sl, zeitabh} itself 
-- operating with the evolution systems for $\frac 1 \eps A + [P',P] = \frac{1}{\eps} A_{0 \, \eps} + K_{0 \, \eps} \ne \frac{1}{\eps} A_{1 \, \eps} + K_{1 \, \eps}$ as comparison evolutions --
is not contained in Theorem~\ref{thm: höherer adsatz}. 

\subsection{Adiabatic theorems without spectral gap condition}

We now prove an adiabatic theorem without spectral gap condition for time-dependent domains in which case we have to explicitly require the differentiability of the resolvent as well as an estimate on the derivative of the resolvent, which two things are no longer automatically satisfied as they were in the case of time-independent domains.

\begin{thm} \label{thm: erw adsatz ohne sl, zeitabh}
Suppose $A(t): D(A(t)) \subset X \to X$ for every $t \in I$ is a linear map such that Condition~\ref{cond: U_eps existiert und beschränkt} is satisfied. Suppose further that $\lambda(t)$ for every $t \in I$ is an eigenvalue of $A(t)$, and that there are numbers $\delta_0 \in (0,\infty)$ and $\vartheta(t) \in \R$ such that $\lambda(t) + \delta e^{i \vartheta(t)} \in \rho(A(t))$ for all $\delta \in (0,\delta_0]$ and $t \in I$ and such that $t \mapsto \lambda(t)$ and $t \mapsto e^{i \vartheta(t)}$ are continuously differentiable and $t \mapsto \big( \lambda(t) + \delta e^{i \vartheta(t)} - A(t) \big)^{-1}$ is SOT-continuously differentiable.
Suppose finally that $P(t)$ for every $t \in I$ is a bounded projection in $X$ commuting with $A(t)$ 
such that $P(t)$ for almost every $t \in I$ is weakly associated with $A(t)$ and $\lambda(t)$ and that 
\begin{align*}
P(t)X \subset \ker (A(t)-\lambda(t))^{m_0}  
\end{align*} 
for every $t \in I$ (and some $m_0 \in \N$). Additionally, suppose that there are $M_0, M_0' \in (0,\infty)$ such that 
\begin{align*}
&\norm{ \big( \lambda(t) + \delta e^{i \vartheta(t)} - A(t) \big)^{-1} (1-P(t)) } \le \frac{M_0}{\delta}, \\
& \qquad \qquad \qquad \qquad \norm{  \ddt{ \Big( \big( \lambda(t) + \delta e^{i \vartheta(t)} - A(t) \big)^{-1} } (1-P(t)) \Big)    } \le \frac{M_0'}{\delta^{m_0+1}}
\end{align*}
for all $\delta \in (0, \delta_0]$ and $t \in I$,
let $\rk P(0) < \infty$ and let $t \mapsto P(t)$ be SOT-continuously differentiable. 
Then
\begin{align*} 
\sup_{t \in I} \norm{ \big( U_{\eps}(t) - V_{0\,\eps}(t) \big) P(0) } \longrightarrow 0 \quad (\eps \searrow 0),
\end{align*}
where $V_{0\,\eps}$ for every $\eps \in (0,\infty)$ denotes the evolution system for $\frac 1 \eps  A P + [P',P]$ on $X$. If, in addition, $X$ is reflexive and $t \mapsto P(t)$ norm continuously differentiable, then  
\begin{align*}
\sup_{t \in I} \norm{ U_{\eps}(t) - V_{\eps}(t) } \longrightarrow 0 \quad (\eps \searrow 0),
\end{align*}
whenever the evolution system $V_{\eps}$ for $\frac 1 \eps A + [P',P]$ exists on $D(A(t))$ for every $\eps \in (0, \infty)$.
\end{thm}

\begin{proof}
Apart from three small changes we can take over the proof of Theorem~\ref{thm: erw adsatz ohne sl} (notice that Lemma~\ref{lm: lm 1 zum erw adsatz ohne sl} and Lemma~\ref{lm: lm 2 zum erw adsatz ohne sl} 
also apply in the present case of time-dependent domains). 
What has to be changed is the following:
first, the inclusion~\eqref{eq: gl 0, adsatz ohne sl} holds true for every $t \in I$ by assumption, 
while in Section~4 this was derived by a closedness argument. 
Second, the fact that
\begin{align*}
[0,t] \ni s \mapsto U_{\eps}(t,s) B_{n \, \bm{\delta}}(s) V_{0\,\eps}(s) P(0)x \quad \text{resp.} \quad [0,t] \ni s \mapsto U_{\eps}(t,s) B_{n \, \bm{\delta}}(s) V_{\eps}(s) x
\end{align*}
is the continuous representative of an element in $W^{1,\infty}([0,t],X)$ for all $x \in X$ resp.~all $x \in D(A(0))$ 
can no longer be deduced from Lemma~\ref{lm: prod- und inversenregel} but has to be inferred from Lemma~\ref{lm: prodregel rechtsseit db} and Lemma~\ref{lm: rechtsseit db und W^{1,infty}} -- notice that $s \mapsto B_{n \, \bm{\delta}}(s)$ is SOT-continuously differentiable (Lemma~\ref{lm: lm 2 zum erw adsatz ohne sl}) with $B_{n \, \bm{\delta}}(s) X \subset D(A(s))$ for every $s \in I$. 
And third, the 
derivative of $s \mapsto \overline{R}_{\eps}(s)$ can no longer be explicitly expressed and estimated as in~\eqref{eq: absch R_eps'} 
but the respective estimate holds true by assumption. 
\end{proof}


Similarly, one sees that 
the variants of the adiabatic theorem without spectral gap condition of Section~4 
carry over to the case of time-dependent domains as well, provided their hypotheses are adapted in a similar way as above.

\subsection{An adiabatic theorem of higher order}

In this subsection we extend the adiabatic theorem of higher order of Joye and Pfister from~\cite{JoyePfister93} to the case of general operetors $A(t)$ with possibly time-dependent domains -- mainly for the sake of completeness and in order to make clear the relation to the basic adiabatic theorem with spectral gap (Theorem~\ref{thm: adsatz mit sl, zeitabh}).
We will use the elegant iterative scheme of~\cite{JoyePfister93} which we briefly recall (in a slightly modified form). 
\smallskip

Suppose $A(t): D(A(t)) \subset X \to X$ is a densely defined closed linear map and $\gamma_t$ is a cycle in $\C$ for every $t \in J$, where $J$ is a compact interval, and let $\eps \in (0, \infty)$ and $n \in \N$. Then $A_{0 \, \eps}$, $P_{0 \, \eps}$, $K_{0 \, \eps}$ are called \emph{well-defined w.r.t.~$\gamma_t$ ($t \in J$)} if and only if $\rg \gamma_t \subset \rho(A_{0 \, \eps}(t))$ for all $t \in J$, where $A_{0 \, \eps}(t) := A(t)$, and $J \ni t \mapsto P_{0 \, \eps}(t)$ is WOT-continuously differentiable, where
\begin{align*} 
P_{0 \, \eps}(t):= \frac{1}{2 \pi i} \int_{\gamma_t} (z-A_{0 \, \eps}(t))^{-1} \, dz. 
\end{align*}
In this case $K_{0 \, \eps}$ is defined by $K_{0 \, \eps}(t) := [P_{0 \, \eps}'(t), P_{0 \, \eps}(t)]$.
And, for general $n \in \N$, $A_{n \, \eps}$, $P_{n \, \eps}$, $K_{n \, \eps}$ are called \emph{well-defined w.r.t.~$\gamma_t$ ($t \in J$)} if and only if  $A_{n-1 \, \, \eps}$, $P_{n-1  \, \eps}$, $K_{n-1 \, \eps}$ are well-defined w.r.t.~$\gamma_t$ ($t \in J$),  $\rg \gamma_t \subset \rho(A_{n \, \eps}(t))$ for all $t \in J$, where $A_{n \, \eps}(t) := A(t) - \eps K_{n-1 \, \eps}(t)$, and $J \ni t \mapsto P_{n \, \eps}(t)$ is WOT-continuously differentiable, where 
\begin{align*}
P_{n \, \eps}(t):= \frac{1}{2 \pi i} \int_{\gamma_t} (z-A_{n \, \eps}(t))^{-1} \, dz.
\end{align*}
In this case $K_{n \, \eps}$ is defined by $K_{n \, \eps}(t) := [P_{n \, \eps}'(t), P_{n \, \eps}(t)]$.
\smallskip

We will need the following conditions depending on $n \in \N \cup \{ \infty \}$ (degree of regularity) in the adiabatic theorem of higher order below. 

\begin{cond} \label{cond: vor höherer adsatz}
$A(t): D(A(t)) \subset X \to X$ for every $t \in I$ is a densely defined closed linear map. 
$\sigma(t)$ for every $t \in I$ is a compact and isolated subset of $\sigma(A(t))$, there is an $r_0 > 0$ such that $U_{r_0}(\sigma(t)) \setminus \sigma(t) \subset \rho(A(t))$ for all $t \in I$, 
and $t \mapsto \sigma(t)$ is continuous. 
For every $t_0 \in I$, there are positive constants $a_{t_0}$, $b_{t_0}$, $c_{t_0}$ such that 
\begin{gather*}
J_{t_0} \ni t \mapsto (z-A(t))^{-1}  \text{ is $n$ times WOT-continuously differentiable for all } z \in \rg \gamma_{t_0}, \\ 
\rg \gamma_{t_0} \ni z \mapsto \ddtl{ (z-A(t))^{-1} } \text{ is SOT-continuous for all } t \in J_{t_0}, l \in \{ 1, \dots, n \}, \\
\sup_{ (t,z) \in J_{t_0} \times \rg \gamma_{t_0} } \norm{ \ddtl{ (z-A(t))^{-1} } } \le a_{t_0} c_{t_0}^l \frac{l!}{(1+l)^2} \text{ for all } l \in \{ 1, \dots, n \},
\end{gather*}
where $\gamma_{t_0}$ is a cycle in $\overline{U}_{\frac{4 r_0}{7}}(\sigma(t_0)) \setminus U_{ \frac{3 r_0}{7} }(\sigma(t_0))$ with 
\begin{align*}
\operatorname{n}\big( \gamma_{t_0}, U_{ \frac{3 r_0}{7} }(\sigma(t_0)) \big) = 1 \quad \text{and} \quad \operatorname{n}\big( \gamma_{t_0}, \C \setminus \overline{U}_{ \frac{4 r_0}{7} }(\sigma(t_0)) \big) = 0
\end{align*}
and where $J_{t_0} \subset I$ is a non-trivial closed interval containing $t_0$ such that $\sigma(t) \subset U_{ \frac{r_0}{7} }(\sigma(t_0))$ and $\sigma(t_0) \subset U_{ \frac{r_0}{7} }(\sigma(t))$ for all $t \in J_{t_0}$.
And finally, $P(t)$ for every $t \in I$ is the 
projection associated with $A(t)$ and $\sigma(t)$,
$t \mapsto P(t)$ is $n+1$ times WOT-continuously differentiable and
\begin{align*}
\sup_{t \in J_{t_0}} \norm{ \ddtl{ [P'(t),P(t)] } } \le b_{t_0} c_{t_0}^l \frac{l!}{(1+l)^2} \text{ for all } l \in \{ 0, 1, \dots, n \} \text{ and } t_0 \in I.
\end{align*}
\end{cond}

In the special case of time-independent domains $D(A(t)) = D$, one easily sees -- using the remark before Lemma~\ref{lm: prod- und inversenregel} -- that the requirements on the resolvent of $A(t)$ in Condition~\ref{cond: vor höherer adsatz} are fulfilled for an $n \in \N$ if, for instance, $t \mapsto A(t)x$ is $n$ times weakly continuously differentiable for all $x \in D$. 
And they 
are fulfilled for $n = \infty$ if, for instance, there is an open neighbourhood $U_I$ of $I$ in $\C$ such that, for every $x \in D$, $t \mapsto A(t)x$ extends to a holomorphic map on $U_I$ (Cauchy inequalities!).

\begin{lm}[Joye--Pfister] \label{lm: iteration wohldef}
(i) Suppose that Condition~\ref{cond: vor höherer adsatz} is satisfied for an $n \in \N$. Then there is an $\eps^* >0$ such that $A_{n \, \eps}$, $P_{n \, \eps}$, $K_{n \, \eps}$ are well-defined w.r.t.~$\gamma_t$ ($t \in I$) for every $\eps \in (0, \eps^*]$. Furthermore, 
\begin{align*}
\sup_{t \in I} \norm{ K_{n \, \eps}(t) - K_{n-1 \, \eps}(t)} = O(\eps^n) \quad (\eps \searrow 0).
\end{align*}
(ii) Suppose that Condition~\ref{cond: vor höherer adsatz} is satisfied for $n = \infty$. Then there is an $\eps^* >0$ and for every $\eps \in (0, \eps^*]$ there is a natural number $n^*(\eps) \in \N$ such that $A_{n^*(\eps) \, \eps}$, $P_{n^*(\eps) \, \eps}$, $K_{n^*(\eps) \, \eps}$ are well-defined w.r.t.~$\gamma_t$ ($t \in I$) for every $\eps \in (0, \eps^*]$. Furthermore, there is 
a constant $g \in (0, \infty)$ such that
\begin{align*}
\sup_{t \in I} \norm{ K_{n^*(\eps) \, \eps}(t) - K_{n^*(\eps)-1 \, \eps}(t)} = O\bigl( e^{-\frac{g}{\eps}} \bigr) \quad (\eps \searrow 0).
\end{align*}
\end{lm}

\begin{proof}
We begin with some general preparatory considerations from which both part~(i) and part~(ii) will easily follow. 
Suppose (for the entire proof) that Condition~\ref{cond: vor höherer adsatz} is satisfied for $n = 1$ 
and fix $t_0 \in I$ for the moment. 
We have 
\begin{align*}
J_{t_0} \times C_{t_0} := J_{t_0} \times \overline{U}_{ \frac{5 r_0}{7} }(\sigma(t_0)) \setminus U_{ \frac{2 r_0}{7} }(\sigma(t_0)) \subset \subset \bigl\{ (t,z) \in J_{t_0} \times \C: z \in \rho(A(t)) \bigr \} =: U_{t_0}
\end{align*}
and $U_{t_0} \ni (t,z) \mapsto (z-A(t))^{-1}$ is continuous, because $J_{t_0} \ni t \mapsto A(t)$ is continuous in the generalized sense due to the WOT-continuous differentiability of $J_{t_0} \ni t \mapsto (z-A(t))^{-1}$ (Theorem~IV.3.15 of~\cite{KatoPerturbation80}). Consequently,
$\sup_{(t,z) \in J_{t_0} \times C_{t_0}} \norm{ (z-A(t))^{-1} } < \infty$,
whence we can (and will) assume w.l.o.g. that 
\begin{align} \label{eq: w.l.o.g. assumption}
\sup_{(t,z) \in J_{t_0} \times C_{t_0}} \norm{ (z-A(t))^{-1} } \le a_{t_0}.
\end{align}
We now define $\eps_{t_0}^*$ and $n_{t_0}^*(\eps)$ just like in Joye and Pfister's paper~\cite{JoyePfister93}, that is, 
\begin{align} \label{eq: def von eps^*}
\nonumber \eps_{t_0}^* &:= \max \Bigl \{ \eps \in \bigl( 0, \frac{1}{2 a_{t_0} b_{t_0}} \bigr): \sum_{k=1}^{\infty} \bigl( 2 \alpha^2 a_{t_0} b_{t_0} \frac{\eps}{1-2 a_{t_0} b_{t_0} \eps} \bigr)^k \le \alpha \Bigr \}, \\
& \qquad \qquad \quad n_{t_0}^*(\eps) := \Big \lfloor \frac{1}{ e c_{t_0} d_{t_0} \, \eps }  \Big \rfloor \text{ for } \eps \in (0,\infty),
\end{align}
where $\alpha$ and $d_{t_0}$ are defined by equation~(2.30) and equation~(2.50) of~\cite{JoyePfister93}. (In particular,  $\eps_{t_0}^*$ and $n_{t_0}^*(\eps)$ only depend on $\gamma_{t_0}$, $a_{t_0}$, $b_{t_0}$ and $c_{t_0}$.) 
We now show by finite induction over $k$: whenever Condition~\ref{cond: vor höherer adsatz} is satisfied for a certain $n' \in \N$, then the following holds true for all $\eps \in (0, \eps_{t_0}^*]$ and all $k \in \{ 1, \dots, n_{t_0}^*(\eps, n') \}$ with $n_{t_0}^*(\eps, n') := \min\{ n_{t_0}^*(\eps), n' \}$: 
\begin{itemize}
\item[(a)] $A_{k \, \eps}$, $P_{k \, \eps}$, $K_{k \, \eps}$ are well-defined w.r.t.~$\gamma_t$ ($t \in J_{t_0}$) 
and $J_{t_0} \ni t \mapsto K_{k \, \eps}(t)$ is $n_{t_0}^*(\eps, n')-k$ times WOT-continuously differentiable
\item[(b)] $\sup_{t \in J_{t_0}} \norm{ K_{k \, \eps}^{(l)}(t) - K_{k-1 \, \eps}^{(l)}(t) } \le b_{t_0} c_{t_0}^{k+l} d_{t_0}^k \eps^k \frac{(k+l)!}{(1+l)^2}$ for all $l \in \N \cup \{ 0 \}$ 
with the property that  
$k + l \le n_{t_0}^*(\eps, n')$
\item[(c)] $\sup_{t \in J_{t_0}} \norm{ K_{k \, \eps}^{(l)}(t) } \le 2 b_{t_0} c_{t_0}^l \frac{l!}{(1+l)^2}$ for all $l \in \N \cup \{ 0 \}$ with $k + l \le n_{t_0}^*(\eps, n')$.
\end{itemize}
Suppose that Condition~\ref{cond: vor höherer adsatz} is satisfied for a certain $n' \in \N$ and fix $\eps \in (0, \eps_{t_0}^*]$.
Set $k = 1$ for the induction basis. We have only to prove assertion~(a) since 
assertions~(b) and~(c) can be gathered from the proof of Proposition~2.1 of~\cite{JoyePfister93}. It is obvious that $A_{0 \, \eps}$, $P_{0 \, \eps}$, $K_{0 \, \eps}$ are well-defined w.r.t.~$\gamma_t$ ($t \in J_{t_0}$) and that $t \mapsto K_{0 \, \eps}(t) = [P'(t), P(t)]$ is $n'$ times WOT-continuously differentiable. 
Since, for $z \in C_{t_0}$ and $t \in J_{t_0}$, 
\begin{gather*}
(z-A_{1 \, \eps}(t)) = \bigl( 1+ \eps K_{0 \, \eps}(t) (z-A(t))^{-1} \bigr) (z-A(t)) \\
\text{ and } \norm{ \eps K_{0 \, \eps}(t) (z-A(t))^{-1} } \le \eps b_{t_0} \norm{ (z-A(t))^{-1} } \le \eps_{t_0}^* b_{t_0} a_{t_0} < \frac{1}{2}
\end{gather*}
(remember the estimate for $K_{0 \, \eps} = [P',P]$ from Condition~\ref{cond: vor höherer adsatz}, the estimate for the resolvent of $A$ from~\eqref{eq: w.l.o.g. assumption}, and the definition of $\eps_{t_0}^*$ in~\eqref{eq: def von eps^*}), we see that 
\begin{align*}
\rg \gamma_t \subset \overline{U}_{ \frac{4 r_0}{7} }(\sigma(t)) \setminus U_{ \frac{3 r_0}{7} }(\sigma(t)) \subset \overline{U}_{ \frac{5 r_0}{7} }(\sigma(t_0)) \setminus U_{ \frac{2 r_0}{7} }(\sigma(t_0)) = C_{t_0} \subset \rho(A_{1 \, \eps}(t))
\end{align*}
for all $t \in J_{t_0}$. 
And since 
\begin{align*}
\operatorname{n}\big( \gamma_t, U_{ \frac{2 r_0}{7} }(\sigma(t_0)) \big) = 
1 &= \operatorname{n}\big( \gamma_{t_0}, U_{ \frac{2 r_0}{7} }(\sigma(t_0)) \big), \\
\operatorname{n}\big( \gamma_t, \C \setminus \overline{U}_{ \frac{5 r_0}{7} }(\sigma(t_0)) \big) = 
0 &= \operatorname{n}\big( \gamma_{t_0}, \C \setminus \overline{U}_{ \frac{5 r_0}{7} }(\sigma(t_0)) \big)
\end{align*}
and $C_{t_0} \subset \rho(A_{1 \, \eps}(t))$ for all $t \in J_{t_0}$, the cycles $\gamma_t$ and $\gamma_{t_0}$ are homologous in $\rho(A_{1 \, \eps}(t))$ for $t \in J_{t_0}$, so that
\begin{align*}
J_{t_0} \ni t \mapsto P_{1 \, \eps}(t) &= \frac{1}{2 \pi i} \int_{\gamma_{t}} (z-A_{1 \, \eps}(t))^{-1} \, dz \\
&= \frac{1}{2 \pi i} \int_{\gamma_{t_0}} (z-A(t))^{-1} \big( 1 + \eps K_{0 \, \eps}(t) (z-A(t))^{-1} \big)^{-1} \, dz
\end{align*}
is $n'$ times WOT-continuously differentiable. 
(In order to see this, use the product rule and inverses rule for WOT-continuous differentiability from the remark before Lemma~\ref{lm: prod- und inversenregel} as well as Condition~\ref{cond: vor höherer adsatz}.)
Consequently, $A_{1 \, \eps}$, $P_{1 \, \eps}$, $K_{1 \, \eps}$ are well-defined w.r.t.~$\gamma_t$ ($t \in J_{t_0})$ and $t \mapsto K_{1 \, \eps}(t)$ is $n'-1$ times (in particular, $n_{t_0}^*(\eps, n')-1$ times) WOT-continuously differentiable.

Choose now $k \in \{2, \dots, n_{t_0}^*(\eps, n') \}$ and assume that assertions (a), (b), (c) are true for $k-1$. We then have  to show that they are also true for $k$. As above we have only to establish (a) since (b) and (c) can then be derived as in the proof of Proposition~2.1 of~\cite{JoyePfister93}, as a close inspection of that proof shows. And in order to prove (a) we can proceed essentially as above: just use assertion~(c) for $k-1$ 
to get the estimate 
\begin{align*}
\sup_{(t,z) \in J_{t_0} \times C_{t_0}} \norm{ \eps K_{k-1 \, \eps}(t) (z-A(t))^{-1} } \le 2 b_{t_0} a_{t_0} \eps_{t_0}^* < 1
\end{align*}
and continue as above, thereby concluding the inductive proof of~(a), (b), (c).
\smallskip

Choosing finitely many points $t_1, \dots, t_m \in I$ such that $J_{t_1} \cup \dots \cup J_{t_m} = I$, and setting 
\begin{align} \label{eq: def eps^* und n^*(eps), global}
\eps^* := \min \{ \eps_{t_1}^*, \dots, \eps_{t_m}^* \} \text{ and } n^*(\eps) := \min \{ n_{t_1}^*(\eps), \dots, n_{t_m}^*(\eps) \},
\end{align} 
we find -- in virtue of the above preparations -- that, for every $\eps \in (0, \eps^*]$, the following holds true: whenever Condition~\ref{cond: vor höherer adsatz} is fulfilled for an $n' \in \N$, then $A_{k \, \eps}$, $P_{k \, \eps}$, $K_{k \, \eps}$ are well-defined w.r.t.~$\gamma_t$ ($t \in I$) and
\begin{align} \label{eq: diff der K-terme}
\sup_{t \in I} \norm{K_{k \, \eps}(t) - K_{k-1 \, \eps}(t) } \le b c^k d^k \eps^k k!
\end{align}
for every $k \in \{ 1, \dots, n^*(\eps, n') \}$, where $b$, $c$, $d$ are obtained by taking the maximum of the corresponding quantities for the points $t_1, \dots, t_m$ and $n^*(\eps, n') := \min \{ n^*(\eps), n' \}$.
\smallskip

Suppose now as in~(i) that Condition~\ref{cond: vor höherer adsatz} is satisfied for an $n \in \N$. Since $n^*(\eps) \longrightarrow \infty$ as $\eps \searrow 0$, 
we can assume w.l.o.g. that $n^*(\eps,n) = n$ for all $\eps \in (0, \eps^*]$ 
and therefore assertion~(i) follows from~\eqref{eq: diff der K-terme}.
Suppose finally as in~(ii) that Condition~\ref{cond: vor höherer adsatz} is satisfied for $n = \infty$. Since for every $\eps \in (0, \eps^*]$ 
there is $n' \in \N$ such that $n^*(\eps, n') = n^*(\eps)$ and since Condition~\ref{cond: vor höherer adsatz} is satisfied, in particular, for this $n'$, assertion~(ii) follows from~\eqref{eq: diff der K-terme} 
with the help of Stirling's formula (see, for instance, the proof of Theorem~2.1 of~\cite{JoyePfister93} or of Theorem~1b of~\cite{Nenciu93}).
\end{proof}

After these preparations we can now prove the announced adiabatic theorem of higher order extending Theorem~2.1 of~\cite{JoyePfister93} where skew self-adjoint operators $A(t)$ that analytically depend on $t$ and have time-independent domains are considered.

\begin{thm} \label{thm: höherer adsatz}
Suppose $A(t)$, $\sigma(t)$, $P(t)$ for $t \in I$ are such that Condition~\ref{cond: U_eps existiert und beschränkt} is satisfied and Condition~\ref{cond: vor höherer adsatz} with WOT replaced by SOT 
is satisfied for an $n \in \N$ or $n = \infty$, respectively. Then 
\begin{itemize}
\item[(i)] $\sup_{t \in I} \norm{ P_{\eps}(t) - P(t) } = O(\eps)$ as $\eps \searrow 0$, where, for all $\eps \in (0, \eps^*]$ and $t \in I$, $P_{\eps}(t) := P_{n \, \eps}(t)$ in case $n \in \N$ and $P_{\eps}(t) := P_{n^*(\eps) \, \eps}(t)$ in case $n = \infty$ (and where $\eps^*$ and $n^*(\eps)$ are defined as in~\eqref{eq: def eps^* und n^*(eps), global} of the lemma above). 
\item[(ii)] Whenever 
the evolution system $V_{\eps}$ for $\frac{1}{\eps} A_{n \, \eps} + K_{n \, \eps}$ resp.~$\frac{1}{\eps} A_{n^*(\eps) \, \eps} +  K_{n^*(\eps) \, \eps}$ exists on $D(A(t))$ for all $\eps \in (0,\eps^*]$, then $V_{\eps}$ is adiabatic w.r.t.~$P_{\eps}$ and for a suitable constant $g \in (0,\infty)$
\begin{align*}
\sup_{t \in I} \norm{ V_{\eps}(t) - U_{\eps}(t) } = O(\eps^n)  \text{ resp. } O\bigl( e^{-\frac{g}{\eps}} \bigr) \quad (\eps \searrow 0).
\end{align*}
\item[(iii)] 
Additionally, one has -- the existence of $V_{ \frac{1}{\eps}}$ being  irrelevant here -- 
that 
\begin{align*}
&\sup_{t \in I} \norm{ (1-P_{\eps}(t)) U_{\eps}(t) P_{\eps}(0) }, \\
& \qquad \qquad \sup_{t \in I} \norm{ P_{\eps}(t)U_{\eps}(t)(1-P_{\eps}(0)) } = O(\eps^n)  \text{ resp. } O\bigl( e^{-\frac{g}{\eps}} \bigr) \quad (\eps \searrow 0).
\end{align*}
\end{itemize}
\end{thm}

\begin{proof}
(i) Set $A_{\eps}(t) := A_{n \, \eps}(t)$ and $K_{\eps}^-(t) := K_{n-1 \, \eps}(t)$ in case  $n \in \N$ and $A_{\eps}(t) := A_{n^*(\eps) \, \eps}(t)$ and $K_{\eps}^-(t) := K_{n^*(\eps)-1 \, \eps}(t)$ in case $n = \infty$ (for $t \in I$ and $\eps \in (0,\eps^*]$).
As was shown in the proof of the above lemma, the cycles $\gamma_{t}$ and $\gamma_{t_i}$ are homologous in $\rho(A_{\eps}(t))$ for every $t \in J_{t_i}$ (where $t_1, \dots, t_m$ are points of $I$ chosen as in the definition of $\eps^*$ and $n^*(\eps)$ in~\eqref{eq: def eps^* und n^*(eps), global}) and every $\eps \in (0,\eps^*]$, whence 
\begin{align*}
P_{\eps}(t) - P(t) &= \frac{1}{2 \pi i} \int_{\gamma_{t_i}} (z-A_{\eps}(t))^{-1} - (z-A(t))^{-1} \, dz \\
&= - \frac{1}{2 \pi i} \int_{\gamma_{t_i}} (z-A_{\eps}(t))^{-1} \, \eps K_{\eps}^-(t) \, (z-A(t))^{-1} \, dz
\end{align*} 
for all $t \in J_{t_i}$ and $\eps \in (0,\eps^*]$. 
Also, it was shown in the proof of the above lemma 
that for all $\eps \in (0, \eps^*]$ and all $i \in \{1, \dots, m\}$ one has $\sup_{(t,z) \in J_{t_i} \times \rg \gamma_{t_i}} \norm{ (z-A(t))^{-1} } \le a_{t_i}$, $\sup_{t \in J_{t_i}} \norm{ K_{\eps}^-(t) } \le 2 b_{t_i}$, and 
\begin{align*}
\norm{ (z-A_{\eps}(t))^{-1} } &\le  \norm{ (z-A(t))^{-1} } \norm{ \big( 1 + \eps K_{\eps}^-(t) (z-A(t))^{-1} \big)^{-1} } \\
&\le a_{t_i} \sum_{m=0}^{\infty} ( \eps 2 b_{t_i} a_{t_i} )^m \le \frac{a_{t_i}}{1-2 a_{t_i} b_{t_i} \eps_{t_i}^*} < \infty 
\end{align*}   
for all $(t,z) \in J_{t_i} \times \rg \gamma_{t_i}$. Assertion~(i) is now clear (notice that for this assertion Condition~\ref{cond: U_eps existiert und beschränkt} is not needed -- only Condition~\ref{cond: vor höherer adsatz} in its original WOT version is used). 
\smallskip

(ii) Set $K_{\eps}^+(t) := [P_{\eps}'(t), P_{\eps}(t)]$ for $t \in I$ and $\eps \in (0,\eps^*]$ 
and suppose that the evolution system $V_{\eps}$ for $\frac{1}{\eps} A_{\eps} + K_{\eps}^+$ exists on $D(A(t))$. 
Since for every $x \in D(A(0))$ the map $[0,t] \ni s \mapsto U_{\eps}(t,s) V_{\eps}(s) x$ is continuous and right differentiable (by Lemma~\ref{lm: prodregel rechtsseit db}) and since the right derivative $s \mapsto U_{\eps}(t,s) \big( K_{\eps}^+(s)-K_{\eps}^-(s) \big) V_{\eps}(s) x$ is bounded, 
it follows from Lemma~\ref{lm: rechtsseit db und W^{1,infty}} that
\begin{align} \label{eq: höherer adsatz, intdarst}
V_{\eps}(t)x - U_{\eps}(t)x &= U_{\eps}(t,s) V_{\eps}(s) x \big |_{s=0}^{s=t} \nonumber \\
&= \int_0^t  U_{\eps}(t,s) \big( K_{\eps}^+(s)-K_{\eps}^-(s) \big) V_{\eps}(s) x \, ds
\end{align}
for all $t \in I$. And from this, in turn, we conclude the desired estimates -- using the estimates for $K_{\eps}^+-K_{\eps}^-$ from Lemma~\ref{lm: iteration wohldef} and applying Proposition~\ref{prop: störreihe für gestörte zeitentw}~(ii). 
It remains to show that $V_{\eps}$ is adiabatic w.r.t.~$P_{\eps}$. As, by assumption, Condition~\ref{cond: vor höherer adsatz} is satisfied with WOT replaced by SOT (up to now, the unaltered Condition~\ref{cond: vor höherer adsatz} was sufficient), $t \mapsto P_{\eps}(t)$ is continuously differentiable not only w.r.t.~WOT but also w.r.t.~SOT. 
%
And therefore, the adiabaticity of $V_{\eps}$ follows from Propostion~\ref{prop: intertwining relation}.
\smallskip

(iii) Arguing as in the adiabaticity proof above, we get for every $x \in D(A(0))$ and every $t \in I$ that
\begin{align*}
P_{\eps}(t) U_{\eps}(t)x - U_{\eps}(t)P_{\eps}(0)x &= U_{\eps}(t,s) P_{\eps}(s) U_{\eps}(s)x \big |_{s=0}^{s=t} \\
&= \int_0^t U_{\eps}(t,s) \Big( P_{\eps}'(s) - \frac{1}{\eps} \big( A(s) P_{\eps}(s) - P_{\eps}(s) A(s) \big) \Big) U_{\eps}(s)x \, ds.
\end{align*}
Since $A_{\eps}(s)$ commutes with $P_{\eps}(s)$ for $s \in I$ and since $A = A_{\eps} + \eps K_{\eps}^-$, we have
\begin{align*}
P_{\eps}'(s) &- \frac{1}{\eps} \big( A(s)P_{\eps}(s) - P_{\eps}(s)A(s) \big) \subset P_{\eps}'(s) - [K_{\eps}^-(s), P_{\eps}(s)] \\
&= P_{\eps}'(s) - [K_{\eps}^+(s), P_{\eps}(s)] + [K_{\eps}^+(s) - K_{\eps}^-(s), P_{\eps}(s)] = [K_{\eps}^+(s) - K_{\eps}^-(s), P_{\eps}(s)]
\end{align*}
for every $s \in I$, and the desired conclusion follows with the help of Lemma~\ref{lm: iteration wohldef}. 
\end{proof}

It is obvious from the definition of Joye and Pfister's iterative scheme that $P_{\eps}(t) = P(t)$ for all $t$ in the (possibly empty) set $I \setminus \supp P'$, and therefore it follows from Theorem~\ref{thm: höherer adsatz}~(iii) that 
\begin{align*}
&\sup_{t \in I \setminus \supp P'} \norm{ (1-P(t))U_{ \frac{1}{\eps} }(t)P(0) }, \\
& \qquad \qquad \sup_{t \in I \setminus \supp P'} \norm{ P(t)U_{  \frac{1}{\eps} }(t)(1-P(0)) } = O(\eps^n)  \text{ resp. } O\bigl( e^{-\frac{g}{\eps}} \bigr) \quad (\eps \searrow 0).
\end{align*}
%

A result similar to Theorem~\ref{thm: höherer adsatz} could have been proved with the help of a method developed by Nenciu in~\cite{Nenciu93} -- this can  easily be gathered from 
the exposition in Section~7 of~\cite{dipl}. We have chosen Joye and Pfister's method since it is 
easier to remember 
and effortlessly transferred to  
the case of several compact isolated subsets $\sigma_1(t), \dots, \sigma_m(t)$ of $\sigma(A(t))$ where each is uniformly isolated in $\sigma(A(t))$ and uniformly isolated from each of the others.
\smallskip

We finally comment on a recent 
theorem by Joye from~\cite{Joye07} dealing with time-independent domains and several spectral subsets $\sigma_i(t)$. It allows for a generalization 
of Condition~\ref{cond: U_eps existiert und beschränkt} at the cost of a specialization of Condition~\ref{cond: vor höherer adsatz}
and states the following (where we confine ourselves, for the sake of notational simplicity, to the case of only one spectral subset $\sigma_i(t) = \sigma(t)$):
if -- and what follows is a special case of Condition~\ref{cond: vor höherer adsatz} -- there is an open neighbourhood $U_I$ of $I$ such that $t \mapsto A(t)x$ for every $x \in D$ extends to a holomorphic map on $U_I$ and if $\sigma(t) = \{ \lambda(t) \}$ for every $t \in I$ for a uniformly isolated spectral value $\lambda(t)$ of $A(t)$ of finite algebraic multiplicity (hence an eigenvalue) such that 
$t \mapsto \lambda(t)$ is continuous, 
then it suffices for the conclusion of Theorem~\ref{thm: höherer adsatz} to hold that -- instead of 
Condition~\ref{cond: U_eps existiert und beschränkt} -- $\lambda(t)$ lie in the left closed complex half-plane and 
$A(t) \overline{P}(t)$ generate a contraction semigroup on $X$ for every $t \in I$ (where $\overline{P} := 1-P$).
So, in the above-mentioned special case of Condition~\ref{cond: vor höherer adsatz} 
the boundedness requirement on $U_{\eps}$ 
from Condition~\ref{cond: U_eps existiert und beschränkt} is not necessary for assertions~(i), (ii) and~(iii) of Theorem~\ref{thm: höherer adsatz}. 
It 
is, however, necessary for the convergences 
\begin{align*}
\sup_{t \in I} \norm{ (1-P(t))U_{\eps}(t)P(0) }, 
\quad \sup_{t \in I} \norm{ P(t)U_{\eps}(t)(1-P(0)) }  \longrightarrow 0  \quad (\eps \searrow 0) 
\end{align*}
with the originally given projections $P(t)$, 
which we are primarily interested in here. 
See the example at the end of Section~1 of~\cite{Joye07} for a proof of 
this necessity statement. 
Also, it should be remarked that the above-mentioned special requirements (analyticity and finite algebraic multiplicity) of Joye's theorem from~\cite{Joye07} are really essential for the proof in~\cite{Joye07}. 
Indeed, this proof essentially rests upon the following estimate for the evolution system $V_{0 \, \eps}$ for $\frac{1}{\eps} A_{0 \, \eps} + K_{0 \, \eps} = \frac{1}{\eps} A + [P',P]$ on $D$
\begin{align} \label{eq: entsch absch joye}
\sup_{(s,t) \in \Delta} \norm{ V_{0 \, \eps}(t,s) } \le c \, e^{c/ \eps^{\beta}} \quad (\eps \in (0,\eps^*])
\end{align}
with constants $\beta \in (0,1)$ and $c \in (0,\infty)$ (Proposition~6.1 of~\cite{Joye07}), which then -- by the usual perturbation argument (Proposition~\ref{prop: störreihe für gestörte zeitentw}) -- 
yields the estimates
\begin{align} \label{eq: absch 2 joye}
\sup_{(s,t) \in \Delta} \norm{ U_{\eps}(t,s) }, \, \, \sup_{(s,t) \in \Delta} \norm{ V_{\eps}(t,s) } \le c' \, e^{c'/ \eps^{\beta}} \quad (\eps \in (0,\eps^*])
\end{align}
from which, in turn, by the integral representation~\eqref{eq: höherer adsatz, intdarst} and the exponential decay of $K_{\eps}^+-K_{\eps}^-$ from Lemma~\ref{lm: iteration wohldef} (analyticity requirement!), the conclusion of Theorem~\ref{thm: höherer adsatz} finally follows. And the fundamental estimate~\eqref{eq: entsch absch joye} 
rests upon  a result on the growth (in $\eps$) of the evolution system for analytic families $\frac{1}{\eps} N$ of nilpotent operators $N(t)$ on finite-dimensional spaces (Proposition~4.1 of~\cite{Joye07}), which proposition (by the analyticity and finite algebraic multiplicity requirement!) can be applied to the 
nilpotent endomorphisms 
\begin{align*}
N(t) := W(t)^{-1} (A(t)-\lambda(t)) W(t) \big|_{P(0)X}
\end{align*}
of the finite-dimensional space $P(0)X$ that analytically depend on $t$. 
$W$ denotes the evolution system for $[P',P]$ on $X$ exactly intertwining the subspaces $P(t)X$.
%
%
%

\subsection{An example with time-dependent domains}

We confine ourselves to an example illustrating the adiabatic theorem without spectral gap condition. In this example, a differential operator of the simplest kind occurs, namely $B: W^{1,p}(\R) \subset L^p(\R) \to L^p(\R)$ with $B f := \partial f$ (weak derivative). Since $B$ is the generator of the (left) translation group $T$ on $L^p(\R)$ (which is given by $T(t) f := f(\,.\,+t)$ for $t \in \R$), one has $\sigma(B) \subset i \R$, and since for every $\lambda \in i \R$ the function $g$, defined by
\begin{align*}
g(t) := \frac{e^{\lambda t}}{t^{\alpha}} \, \chi_{[1,\infty)}(t) \quad (t \in \R)
\end{align*}
with arbitrary $\alpha \in (\frac{1}{p}, 1+\frac{1}{p}]$, belongs to $L^p(\R)$ but not to the range of $B - \lambda$, one even has $\sigma(B) = i \R$ for $p \in [1, \infty)$. Additionally, since $B_q^* = - B_{q^*}$ for every $q \in [1,\infty)$ with dual exponent $q^*$ and since $\sigma_p(B_q) = \emptyset$ for $q \in [1,\infty)$ and $\sigma_p(B_q) = i \R$ for $q = \infty$, one obtains the following fine structure of the spectrum of $B$: 
\begin{gather*}
\sigma_p(B) = \emptyset, \quad \sigma_c(B) = \emptyset, \quad \sigma_r(B) = i \R \quad (p = 1) \\
\sigma_p(B) = \emptyset, \quad \sigma_c(B) =  i \R, \quad \sigma_r(B) = \emptyset \quad (p \in (1,\infty)).
\end{gather*}

\begin{ex} \label{ex: ablop}
Suppose $A$, $\lambda$, $P$ with $A(t) = R(t)^{-1} A_0(t) R(t)$, $P(t) = R(t)^{-1} P_0 R(t)$, and $R(t) = e^{C t}$ are given as follows in $X := \ell^p(I_d) \times L^p(\R)$ (where $p \in (1,\infty)$ and $d \in \N$): 
\begin{align*}
A_0(t) := 
\begin{pmatrix} \lambda(t) + \alpha(t) N  & 0 \\ 0 & B \end{pmatrix} 
\quad \text{and} \quad
P_0 := \begin{pmatrix} 1 & 0 \\ 0 & 0 \end{pmatrix},
\end{align*} 
where $\lambda(t) \in \{ \Re z \le 0 \}$, $\alpha(t)$, $N$ are such that Condition~\ref{cond: baustein mit nicht-halbeinfachem ew} is satisfied and where $B$ is the differentiation operator on $L^p(\R)$ defined above, so that, in particular, 
\begin{align*}
D(A_0(t)) = D := \ell^p(I_d) \times W^{1,p}(\R).
\end{align*}
Additionally, suppose $t \mapsto \lambda(t)$, $\alpha(t)$ are continuously differentiable 
and $C$ is the bounded linear map in $\ell^p(I_d) \times L^p(\R)$ given by 
\begin{align*}
C := \begin{pmatrix} 0 & 0 \\ C_0 & 0 \end{pmatrix}
\quad \text{with} \quad
C_0 (x_1, \dots, x_d) := x_d f_0 
\end{align*}
for an arbitrary fixed $0 \ne f_0 \in L^p(\R)$.
Since (by $p \ne 1$) the spectrum of $A_0(t)|_{(1-P_0)D} = B$ is purely continuous 
for every $t \in I$, 
$P_0$ is weakly associated with $A_0(t)$ and $P_0$, and hence the same is true for $A(t)$ and $P(t)$ instead of $A_0(t)$, $P_0$. 
Since, moreover, $B$ 
generates a contraction group (not only a semigroup) in $L^p(\R)$, the resolvent estimates of Theorem~\ref{thm: erw adsatz ohne sl, zeitabh} are satisfied with $\vartheta(t) := \pi$. 
$\blacktriangleleft$
\end{ex}

It follows in the same way as after 
Example~\ref{ex: A_2(t) nicht diagb} that $A$, $\lambda$, $P$ of the above example cannot be reduced to a finite-dimensional subspace and that our adiabatic theorem without spectral gap condition does not apply if $p = 1$ in the example above. 
%
%
%
%
%
Also, it should be noticed that the domains $D(A(t)) = e^{-C t} D$ of the above $A(t)$ really are time-dependent -- more precisely, one 
has $D(A(t_1)) \ne D(A(t_2))$ for $t_1 \ne t_2$ -- if only $f_0$ is chosen to lie not in $W^{1,p}(\R)$.
Indeed, 
if under this condition on $f_0$ one has the twofold representation
\begin{align*}
(x, f - t_1 x_d \, f_0) = e^{-C t_1} (x,f) = e^{-C t_2} (y,g) = (y, g - t_2 y_d \, f_0)
\end{align*} 
for some $(x,f)$, $(y,g) \in \ell^p(I_d) \times W^{1,p}(\R) = D$ with $x_d \ne 0$, 
then 
$t_1$ must be equal to $t_2$.


\section{Adiabatic theorems for operators defined by symmetric sesquilinear forms}  \label{sect: adsätze für A(t) von sesquilinearformen}

After having established general adiabatic theorems for time-dependent domains in Section~5, we now 
apply these theorems to obtain -- as simple corollaries -- 
adiabatic theorems for operators $A(t) = i A_{a(t)}$ defined by densely defined closed symmetric sesquilinear forms $a(t)$ with time-independent form domain -- such as, for instance, Schrödinger operators~$A(t)$ with time-dependent potentials $V(t)$ belonging to the Rollnik class. In particular, the theorem of Section~6.3 contains the adiabatic theorem of Bornemann from~\cite{Bornemann98} as a special case.

\subsection{Some notation and preliminaries}

We start by recording 
the basic conditions (depending on a regularity parameter $n \in \N \cup \{\infty\}$) that shall be imposed on the sesquilinear forms $a(t)$ in the adiabatic theorems of this section.

\begin{cond} \label{cond: vor an a(t)}
$a(t)$ for every $t \in I$ is a symmetric 
sesquilinear form on the Hilbert space $H^+$ (with norm $\norm{\,.\,}^+$ and scalar product $\scprd{\,.\,,\,..\,}^+$) which is densely and continuously embedded into $H$ (with norm $\norm{\,.\,}$ and scalar product $\scprd{\,.\,,\,..\,}$). There is a number $m \in (0,\infty)$ such that
\begin{align*}
\scprd{ \,.\, , \,..\, }_t^+ := a(t)(\,.\, , \,..\, ) + m \scprd{ \,.\, , \,..\, } 
\end{align*} 
is a scalar product on $H^+$ and such that the induced norm $\norm{\,.\,}_t^+$ is equivalent to $\norm{\,.\,}^+$ for every $t \in I$. 
And finally, $t \mapsto a(t)(x,y)$ is $n$ times continuously differentiable for all $x, y \in H^+$. 
\end{cond} 

In Condition~\ref{cond: vor an a(t)}, the requirement that $\scprd{ \,.\, , \,..\, }_t^+$ be a scalar product on $H^+$ whose norm $\norm{\,.\,}_t^+$ is equivalent to $\norm{\,.\,}^+$ for every $t \in I$ could be reformulated in an equivalent way by saying that there is $m \in (0,\infty)$ such that $a(t)(\,.\, , \,..\, ) + m \scprd{ \,.\, , \,..\, }$ is $\norm{\,.\,}^+$-bounded and $\norm{\,.\,}^+$-coercive. 
It is well-known that under Condition~\ref{cond: vor an a(t)} there is, for every $t \in I$, a unique self-adjoint operator $A_{a(t)}: D(A_{a(t)}) \subset H \to H$ such that 
\begin{align*}
D(A_{a(t)}) \subset H^+ \quad \text{and} \quad \scprd{x,A_{a(t)}y} = a(t)(x,y)
\end{align*}
for every $x \in H^+$ and $y \in D(A_{a(t)})$ (Theorem~VI.2.1 and Theorem~VI.2.6 of~\cite{KatoPerturbation80} or Theorem~10.1.2 of~\cite{BirmanSolomjak}).
As usual, we denote -- in the situation of Condition~\ref{cond: vor an a(t)} -- by $H^-$ the space of $\norm{\,.\,}^+$-continuous conjugate linear functionals $H^+ \to \C$, which obviously is a Hilbert space w.r.t.~the norms 
\begin{align*}
f \mapsto \norm{f}^- := \sup_{\norm{x}^+ = 1} |f(x)| 
\quad \text{and} \quad 
f \mapsto \norm{f}_t^- := \sup_{\norm{x}_t^+ = 1} |f(x)| \quad (t \in I).
\end{align*}
And by $j: H \to H^-$ we denote the injective continuous linear map with $j(x) := \scprd{\,.\,,x} \in H^-$ for $x \in H$. 
%
\smallskip

We continue by citing the 
fundamental theorem of Kisy\'{n}ski (Theorem~8.1 of~\cite{Kisynski63}) giving sufficient conditions for the well-posedness of the initial value problems corresponding to $A$ on $D(A(t))$, where $A(t) = i A_{a(t)}$ with symmetric sesquilinear forms $a(t)$ with constant 
domain. 
Similar theorems on well-posedness can be proved for the case 
of operators $A(t) = -A_{a(t)}$ defined by sectorial sesquilinear forms $a(t)$ with time-independent form 
domain. See, for instance, Fujie and Tanabe's article~\cite{FujieTanabe73} (Theorem~3.1) or Kato and Tanabe's article~\cite{KatoTanabe62} (Theorem~7.3). 

\begin{thm} [Kisy\'{n}ski] \label{thm: Kisynski}
Suppose $a(t)$ for every $t \in I$ is a sesquilinear form such that Condition~\ref{cond: vor an a(t)} is satisfied with $n = 2$ and set $A(t) := i A_{a(t)}$ for $t \in I$. Then there is a unique evolution system $U$ for $A$ on $D(A(t))$ and $U(t,s)$ is unitary in $H$ for every $(s,t) \in \Delta$.
\end{thm}

In particular, this theorem guarantees that the basic Condition~\ref{cond: U_eps existiert und beschränkt} of the general adiabatic theorems for time-dependent domains is satisfied if Condition~\ref{cond: vor an a(t)} is with $n = 2$. 
When it comes to verifying the other conditions of the general adiabatic theorems discussed in Section~\ref{sect: adsätze für zeitabh domains}, the following lemma will be important.

\begin{lm} \label{lm: allg lm}
Suppose that Condition~\ref{cond: vor an a(t)} is satisfied for a certain $n \in \N$ and, for every $t \in I$, denote by $\tilde{A}_0(t)$ the bounded linear map $H^+ \to H^-$ extending $A_0(t):= A_{a(t)}$, that is, $\tilde{A}_0(t)x := a(t)(\,.\,,x)$ for $x \in H^+$.
Then the following holds true:
\begin{itemize}
\item[(i)] $t \mapsto \tilde{A}_0(t)$ is $n$ times WOT-continuously differentiable.
\item[(ii)] If for a certain $z \in \C$ the operator $A_0(t)-z: D(A_0(t)) \subset H \to H$ is bijective for all $t \in J_0$ (a non-trivial subinterval of $I$), then so is $\tilde{A}_0(t)-z j: H^+ \to H^-$ and
\begin{align*}
(A_0(t)-z)^{-1} x = (\tilde{A}_0(t)- z j)^{-1} j(x)
\end{align*} 
for all $t \in J_0$ and $x \in H$.
In particular, $J_0 \ni t \mapsto (A_0(t)-z)^{-1}$ is $n$ times WOT-continuously differentiable.
\end{itemize}
\end{lm}

\begin{proof}
(i) We have only to show that $t \mapsto F(\tilde{A}_0(t)x)$ is $n$ times continuously differentiable for every $x \in H^+$ and every $F \in (H^-)^*$. 
Since the canonical conjugate linear map 
\begin{align*}
H^+ \ni y \mapsto i(y) \in (H^-)^* 
\text{ \, with \, } i(y)(f) := f(y) \text{ \, for } f \in H^- 
\end{align*}
is surjective by the reflexivity of $H^+$, the claim 
is obvious from the continuous differentiability requirement in Condition~\ref{cond: vor an a(t)}. 
\smallskip   

(ii) We fix $t \in I$ and show that
\begin{align} \label{eq: beh}
\rho(A_0(t)) \subset \rho(A_0^-(t)), 
\end{align}
where $A_0^-(t): j(H^+) \subset H^- \to H^-$ is defined by $A_0^-(t) j(x) = \tilde{A}_0(t)x$ for $x \in H^+$. 
Since $A_0^-(t)$ is self-adjoint in  $(H^-, \norm{\,.\,}_t^-)$, it follows that $\C \setminus \R \subset \rho(A_0^-(t))$ and that
\begin{align} \label{eq: formel für resolv}
(A_0^-(t)-z)^{-1} j(x) = j\big( (A_0(t)-z)^{-1} x \big)
\end{align}
for $z \in \C \setminus \R$ and $x \in H$. 
It therefore remains to prove that $\rho(A_0(t)) \cap \R \subset \rho(A_0^-(t))$. 
So let $z \in \rho(A_0(t)) \cap \R$. Then there is $\delta > 0$ such that 
$(z- 2 \delta, z+ 2 \delta) \subset \rho(A_0(t))$, from which it follows by Stone's formula 
(applied to both $A_0(t)$ and $A_0^-(t)$) and by~\eqref{eq: formel für resolv} that
\begin{align*}
0 = j\Big( P_{(z-\delta,z+\delta)}x + \frac{1}{2} P_{ \{ z-\delta, z+\delta\} }x \Big) = \Big( P_{(z-\delta,z+\delta)}^- + \frac{1}{2} P_{ \{ z-\delta, z+\delta\} }^- \Big) j(x)
\end{align*}
for all $x \in H$, where $P$ and $P^-$ denote the spectral measure of $A_0(t)$ and $A_0^-(t)$, respectively. It follows (by the density of $j(H)$ in $H^-$) that $P_{(z-\delta,z+\delta)}^- = 0$ and hence $z \in \rho(A_0^-(t))$.
So~\eqref{eq: beh} is established and the desired conclusion ensues. 
\end{proof}

\subsection{Adiabatic theorems with spectral gap condition}

We will need the following condition depending on a parameter $m \in \{0\} \cup \N \cup \{ \infty \}$ for the adiabatic theorem with spectral gap condition below.

\begin{cond} \label{cond: vor adsatz mit sl für A(t)=iA_{a(t)}}
$A(t) = i A_{a(t)}$ for $t \in I$, where the sesquilinear forms $a(t)$ satisfy Condition~\ref{cond: vor an a(t)} with $n = 2$. 
$\sigma(t)$ for every $t \in I$ is a compact subset of $\sigma(A(t))$, 
$\sigma(\,.\,)$ falls into $\sigma(A(\,.\,)) \setminus \sigma(\,.\,)$ at exactly $m$ points that accumulate at only finitely many points, and $I \setminus N \ni t \mapsto \sigma(t)$ is continuous, where $N$ denotes the set of those $m$ points at which $\sigma(\,.\,)$ falls into $\sigma(A(\,.\,)) \setminus \sigma(\,.\,)$.
$P(t)$ for every $t \in I \setminus N$ is the projection associated with $A(t)$ and $\sigma(t)$ and $I \setminus N \ni t \mapsto P(t)$ extends to a twice SOT-continuously differentiable map (again denoted by $P$) on the whole of $I$.
\end{cond}


In view of Lemma~\ref{lm: allg lm} it is now very easy to derive the following adiabatic theorem with uniform ($m =0$) or non-uniform ($m \in \N \cup \{ \infty \}$) spectral gap condition from the corresponding general adiabatic theorem with spectral gap condition 
(Theorem~\ref{thm: adsatz mit sl, zeitabh}).

\begin{thm} \label{thm: adsatz mit sl}
Suppose $A(t)$, $\sigma(t)$, $P(t)$ for $t \in I$ are as in Condition~\ref{cond: vor adsatz mit sl für A(t)=iA_{a(t)}} with $m = 0$ or $m \in \N \cup \{\infty\}$, respectively. Then
\begin{align*}
\sup_{t \in I} \norm{ U_{\eps}(t) - V_{\eps}(t) } = O(\eps)  \text{ resp. } o(1) \quad (\eps \searrow 0), 
\end{align*}
whenever the evolution system $V_{\eps}$ for $\frac 1 \eps A + [P',P]$ exists on $D(A(t))$ for every $\eps \in (0, \infty)$.
\end{thm}

\begin{proof}
Choose, for every $t_0 \in I \setminus N$, non-trivial closed intervals $J_{t_0}$ and cycles $\gamma_{t_0}$ as in Condition~\ref{cond: vor adsatz mit sl} (which is possible by the relative openness of $I \setminus N$ in $I$). 
In virtue of Lemma~\ref{lm: allg lm} it is then clear that Condition~\ref{cond: vor adsatz mit sl} is fulfilled, and the assertion follows from Theorem~\ref{thm: adsatz mit sl, zeitabh}.
\end{proof}

If the existence of the evolution $V_{\eps}$ for $\frac 1 \eps A +[P',P]$ cannot be ensured, 
one still has 
the remark after Theorem~\ref{thm: adsatz mit sl, zeitabh}.
In the case of uniform spectral gap, 
the existence of $V_{\eps}$ is guaranteed if, for instance, Condition~\ref{cond: vor an a(t)} is 
fulfilled with $n = 3$, since then $I \ni t \mapsto P(t)$ is thrice WOT-continuously differentiable (by Lemma~\ref{lm: allg lm}~(ii)) 
so that the symmetric sesquilinear forms $\frac 1 \eps a(t) + b(t) = \frac 1 \eps a(t) -i \scprd{\,.\,,[P'(t),P(t)]\,..\,}$ corresponding to $\frac{1}{\eps}A(t) + [P'(t),P(t)]$ satisfy Condition~\ref{cond: vor an a(t)} with $n = 2$ and Theorem~\ref{thm: Kisynski} can be applied. 
\smallskip

We finally note conditions under which the general adiabatic theorem of higher order (Theorem~\ref{thm: höherer adsatz}) can be applied to the case  of operators $A(t)$ defined by symmetric sesquilinear forms.

\begin{cond} \label{cond: adsatz höherer ord}
Suppose that $A(t) = iA_{a(t)}$ for $t \in I$ where the sesquilinear forms $a(t)$ satisfy Condition~\ref{cond: vor an a(t)} with a certain $n \in \N \setminus \{1\}$ 
or with $n = \infty$, respectively. 
In the latter case suppose further that there is an open neighbourhood $U_I$ of $I$ in $\C$ and for each $w \in U_I$ there is a $\norm{\,.\,}^+$-bounded sesquilinear form $\tilde{a}(w)$ on $H^+$ such that $\tilde{a}(t) = a(t)$ for $t \in I$ and that $U_I \ni w \mapsto \tilde{a}(w)(x,y)$ is holomorphic for every $x,y \in H^+$. 
Suppose moreover that $\sigma(t)$ for every $t \in I$ is an isolated compact subset of $\sigma(A(t))$, that $\sigma(\,.\,)$ at no point falls into $\sigma(A(\,.\,)) \setminus \sigma(\,.\,)$, 
and that $t \mapsto \sigma(t)$ is continuous. 
And finally, suppose $P(t)$ for every $t \in I$ is the projection associated with $A(t)$ and $\sigma(t)$ and $t \mapsto P(t)$ is $n+1$ times 
times SOT-continuously differentiable.
\end{cond}



It is not difficult (albeit a bit technical) to show that  under Condition~\ref{cond: adsatz höherer ord} the hypotheses of Theorem~\ref{thm: höherer adsatz} are really satisfied.
(In the case $n = \infty$ define $\tilde{A}_0(w)$ by $\tilde{A}_0(w)x := \tilde{a}(w)(\,.\, ,x)$ for $x \in H^+$. Then $\tilde{A}_0(w)$ is a bounded linear map $H^+ \to H^-$ and  
$U_I \ni w \mapsto \tilde{A}_0(w) \in L(H^+,H^-)$ is WOT-holomorphic and hence NOT-holomorphic. A simple perturbation argument and Cauchy's inequality (in conjunction with the formula in Lemma~\ref{lm: allg lm}~(ii)) then yield estimates of the desired kind.)

%

\subsection{An adiabatic theorem without spectral gap condition}

In the adiabtic theorem without spectral gap condition below, the following condition will be used. 

\begin{cond} \label{cond: vor adsatz ohne sl}
$A(t) = i A_{a(t)}$ for $t \in I$ where the sesquilinear forms $a(t)$ satisfy Condition~\ref{cond: vor an a(t)} with $n = 2$. 
$\lambda(t)$ for every $t \in I$ is an eigenvalue of $A(t)$ such that $t \mapsto \lambda(t)$ is continuous. 
And $P(t)$ for every $t \in I$ is an orthogonal projection in $H$ such that $P(t)$ is weakly associated with $A(t)$ and $\lambda(t)$ for almost every $t \in I$, $\rk P(0) < \infty$ and $t \mapsto P(t)$ is SOT-continuously differentiable.
\end{cond}

While in the case with spectral gap Lemma~\ref{lm: allg lm} was sufficient, 
we need another -- well-expected -- lemma in the case without spectral gap. 

\begin{lm} \label{lm: lm für adsatz ohne sl}
Suppose that Condition~\ref{cond: vor adsatz ohne sl} is satisfied and that, in addition, $t \mapsto \lambda(t)$ is continuously differentiable. 
Then $t \mapsto (\lambda(t) + \delta - A(t))^{-1}$ is SOT-continuously differentiable for every $\delta \in (0,\infty)$ and there is an $M_0' \in (0,\infty)$ such that
\begin{align*}
\norm{ \ddt{ (\lambda(t) + \delta - A(t))^{-1} }  } \le \frac{M_0'}{\delta^2}
\end{align*}
for all $t \in I$ and $\delta \in (0,1]$.
\end{lm}

\begin{proof}
Set $A_0(t) := A_{a(t)} = -i A(t)$ and $\lambda_0(t) := -i \lambda(t)$ and let $\tilde{A}_0(t): H^+ \to H^-$ be the bounded extension of $A_0(t)$. 
Since by Lemma~\ref{lm: allg lm} $t \mapsto \tilde{A}_0(t)$ is twice WOT- and, in particular, once SOT-continuously differentiable and $t \mapsto \lambda_0(t)$ is continuously differentiable, it follows that 
\begin{align*}
t \mapsto \big( A_0(t) - (\lambda_0(t) - i\delta) \big)^{-1} = \big( \tilde{A}_0(t) - (\lambda_0(t) - i \delta) j \big)^{-1} j
\end{align*}
is SOT-continuously differentiable for every $\delta \in (0,\infty)$ and that
\begin{align} \label{eq: darst der abl}
&\ddt{ \big( A_0(t) - (\lambda_0(t) - i\delta) \big)^{-1} } \notag \\
&\qquad \quad = \big( \tilde{A}_0(t) - (\lambda_0(t) - i \delta) j \big)^{-1} \, \big( \lambda_0'(t)j -\tilde{A}_0'(t) \big) \, \big( \tilde{A}_0(t) - (\lambda_0(t) - i \delta) j \big)^{-1} j
\end{align}
for $t \in I$ and $\delta \in (0,\infty)$. 
We therefore show 
that there is a constant $c_0' \in (0,\infty)$ such that 
\begin{align} \label{eq: wesentl zwbeh}
\norm{ \big( \tilde{A}_0(t) - (\lambda_0(t) - i \delta)j \big) x }_t^- \ge \frac{ \delta}{c_0'} \norm{x}_t^+
\end{align}
for all $x \in H^+$, $t \in I$ and $\delta \in (0,1]$. In order to do so we observe the following simple fact: if instead of $j$ the natural isometric isomorphism 
\begin{align*}
j_t^+: (H^+, \norm{\,.\,}_t^+) \to (H^-, \norm{\,.\,}_t^-) \text{\, with \,} j_t^+(x) := \scprd{\,.\,,x}_t^+ \text{ for } x \in H^+ 
\end{align*}
occurred in~\eqref{eq: wesentl zwbeh}, this assertion would be trivial. We are therefore led to express $j$ in terms of $j_t^+$: by the definition of the scalar product $\scprd{ \,.\,,\,..\, }_t^+$ in Condition~\ref{cond: vor an a(t)}, we have 
\begin{align*}
j = \frac{1}{m} \big( \tilde{A}_0(t) - j_t^+ \big)
\end{align*} 
for all $t \in I$, so that 
\begin{align*}
\tilde{A}_0(t) - (\lambda_0(t) - i \delta)j = \frac{m+\lambda_0(t)-i \delta}{m} \Big( \tilde{A}_0(t) - \frac{ \lambda_0(t) -i\delta}{m + \lambda_0(t) -i\delta} j_t^+ \Big).
\end{align*}
Since for all $x \in H^+$ with $\norm{x}_t^+ = 1$
\begin{align*}
\norm{  \Big( \tilde{A}_0(t) - \frac{ \lambda_0(t) -i\delta}{m + \lambda_0(t) -i\delta} j_t^+ \Big)x  }_t^- &\ge \Big| a(t)(x,x) - \frac{ \lambda_0(t) -i\delta}{m + \lambda_0(t) -i\delta} \big( j_t^+(x) \big)(x) \Big| \\
&\ge \Big| \Im \Big( \frac{ \lambda_0(t) -i\delta}{m + \lambda_0(t) -i\delta} \Big) \Big| = \frac{m \delta}{ |m + \lambda_0(t) -i\delta|^2 },
\end{align*}
it follows that
\begin{align*}
\norm{ \big( \tilde{A}_0(t) - (\lambda_0(t) - i \delta)j \big)x }_t^- \ge \Big| \frac{m+\lambda_0(t)-i \delta}{m} \Big| \, \frac{m \delta}{ |m + \lambda_0(t) -i\delta|^2 } \norm{x}_t^+ \ge \frac{\delta}{c_0'} \norm{x}_t^+  
\end{align*}
for all $x \in H^+$ and all $t \in I$, $\delta \in (0,1]$, where $c_0' := m + \norm{\lambda}_{\infty} + 1$. So \eqref{eq: wesentl zwbeh} is proven and it follows that
\begin{align} \label{eq: absch inverse}
\norm{ \big( \tilde{A}_0(t) - (\lambda_0(t) - i \delta)j \big)^{-1} }_{H^- \to H^+} \le \frac{c_0'}{ \delta}
\end{align}
for all $t \in I$ and $\delta \in (0,1]$, because the equivalence of the norms $\norm{\,.\,}_t^+$ with $\norm{\,.\,}$ required in Condition~\ref{cond: vor an a(t)} is uniform w.r.t.~$t$ by Lemma~7.3 of~\cite{Kisynski63}.  
In view of~\eqref{eq: darst der abl} and~\eqref{eq: absch inverse} 
the asserted estimate is now clear. 
\end{proof}

With this lemma at hand, it is now simple to derive the following adiabatic theorem without spectral gap condition which 
generalizes an adiabatic theorem of Bornemann (Theorem~IV.1 of~\cite{Bornemann98}).

\begin{thm} \label{thm: adsatz ohne sl für A(t)=iA_{a(t)}}
Suppose $A(t)$, $\lambda(t)$, $P(t)$ for $t \in I$ are such that Condition~\ref{cond: vor adsatz ohne sl} is satisfied.
Then
\begin{align*}
\sup_{t \in I} \norm{ \big( U_{\eps}(t) - V_{0\,\eps}(t) \big) P(0) } \longrightarrow 0 \quad \text{and} \quad \sup_{t \in I} \norm{ P(t) \big( U_{\eps}(t) - V_{0\,\eps}(t) \big) } \longrightarrow 0 
\end{align*}
as $\eps \searrow 0$, where $V_{0\,\eps}$ denotes the evolution system for $\frac{1}{\eps}AP + [P',P] = \frac 1 \eps \lambda P + [P',P]$ for every $\eps \in (0,\infty)$. 
If, in addition, $t \mapsto P(t)$ is thrice WOT-continuously differentiable, then the evolution system $V_{\eps}$ for $\frac 1 \eps A + [P',P]$ exists on $D(A(t))$ for every $\eps \in (0,\infty)$ and
\begin{align*}
\sup_{t \in I} \norm{U_{\eps}(t) - V_{\eps}(t) } \longrightarrow 0 \quad (\eps \searrow 0).
\end{align*}
\end{thm}

\begin{proof}
We have to verify the hypotheses of the general adiabatic theorem without spectral gap condition for time-dependent domains (Theorem~\ref{thm: erw adsatz ohne sl, zeitabh}) with $m_0 = 1$. In view of Lemma~\ref{lm: lm für adsatz ohne sl} it remains to establish two small things, namely the continuous differentiability of $t \mapsto \lambda(t)$ (from Theorem~\ref{thm: erw adsatz ohne sl, zeitabh} and from Lemma~\ref{lm: lm für adsatz ohne sl}) and the inclusion $P(t)H \subset \ker(A(t)-\lambda(t))$ for every $t \in I$ (from Theorem~5.4). 
We know by assumption that $P(t)H = \ker(A(t)-\lambda(t)) = \ker(A_0(t)-\lambda_0(t))$ for almost every $t \in I$ so that $P(t)H \subset D(A_0(t)) \subset H^+$ and
\begin{align*}
0 = j\big( (A_0(t)-\lambda_0(t))P(t)x \big) 
= \big( A_0^-(t)-\lambda_0(t) \big) j(P(t)x)
\end{align*} 
for all $x \in H$ and almost every $t \in I$ (where $A_0(t)$, $\lambda_0(t)$ are defined as in the proof of Lemma~\ref{lm: lm für adsatz ohne sl} and where $A_0^-(t)$ is the self-adjoint operator in $(H^-, \norm{\,.\,}_t^-)$ from the proof of Lemma~\ref{lm: allg lm}). Applying the  closedness argument 
after Theorem~\ref{thm: handl adsatz mit nichtglm sl}
to the closed 
operator $i A_0^-(t): j(H^+) \subset H^- \to H^-$ (with time-independent domain!), 
we see that $j(P(t)H) \subset j(H^+)$ 
and 
\begin{align*}
0 = \big( A_0^-(t)-\lambda_0(t) \big) j(P(t)x) = a(t)(\,.\,,P(t)x) - \lambda_0(t) \scprd{\,.\,,P(t)x}
\end{align*}  
for all $x \in H$ and every (not only almost every) $t \in I$. In particular, for every $t \in I$, 
\begin{align*}
0= a(t)(y,P(t)x) - \lambda_0(t) \scprd{y,P(t)x} = \scprd{(A_0(t)-\lambda_0(t))y, P(t)x}
\end{align*}
for $y \in D(A_0(t))$ and $x \in H$, so that 
\begin{align} \label{eq: inklusion für alle t}
P(t)H \subset \ker(A_0(t)-\lambda_0(t))^* = \ker(A(t)-\lambda(t))
\end{align}
for every $t \in I$, as desired. Since for every $t_0 \in I$ there is a neighbourhood $J_{t_0} \subset I$ and an $x_0 \in H$ such that $P(t)x_0 \ne 0$ for $t \in J_{t_0}$, 
it follows from~\eqref{eq: inklusion für alle t} that
\begin{align*}
\frac{1}{\lambda(t)-1} = \frac{   \scprd{P(t)x_0, (A(t)-1)^{-1} P(t)x_0}  }{   \scprd{P(t)x_0, P(t)x_0}   } 
\end{align*} 
for every $t \in J_{t_0}$, from which in turn it follows (by Lemma~\ref{lm: allg lm}) 
that $t \mapsto \lambda(t)$ is continuously differentiable, as desired.
\smallskip

According to what has been said at the beginning of the proof, it is now clear that Lemma~\ref{lm: lm für adsatz ohne sl} can be applied and that the hypotheses of the first part of Theorem~\ref{thm: erw adsatz ohne sl, zeitabh} are satisfied. 
Since the evolution system $U_{\eps}$ is unitary (by Theorem~\ref{thm: Kisynski}) and $V_{0\,\eps}$ is unitary as well, we see by obviously modifying the proof of Theorem~\ref{thm: erw adsatz ohne sl, zeitabh} that
\begin{align} \label{eq: aussage adsatz für unitäre zeitentw}
\sup_{(s,t)\in I^2} \norm{ \big( U_{\eps}(t,s) - V_{0\,\eps}(t,s) \big) P(s) } \longrightarrow 0 \quad (\eps \searrow 0),
\end{align} 
where $U_{\eps}(t,s) := U_{\eps}(s,t)^{-1} = U_{\eps}(s,t)^*$ and $V_{0\,\eps}(t,s) := V_{0\,\eps}(s,t)^{-1} = V_{0\,\eps}(s,t)^*$ for $(s,t) \in I^2$ with $s > t$. Since 
\begin{align*}
\norm{ P(t) \big( U_{\eps}(t) -V_{0\,\eps}(t) \big) } = \norm{ \big(U_{\eps}(0,t)-V_{0\,\eps}(0,t)\big)P(t) }
\end{align*}
for $t \in I$ (take adjoints), the first two of the asserted convergences follow from~\eqref{eq: aussage adsatz für unitäre zeitentw}. 
\smallskip

Suppose finally that $t \mapsto P(t)$ is thrice WOT-continuously differentiable. Then the symmetric sesquilinear forms $\frac 1 \eps a(t) + b(t) = \frac 1 \eps a(t) -i \scprd{\,.\,,[P'(t),P(t)]\,..\,}$ corresponding to the operators $\frac{1}{\eps}A(t) + [P'(t),P(t)]$ satisfy Condition~\ref{cond: vor an a(t)} with $n = 2$ and therefore the evolution system $V_{\eps}$ for $\frac 1 \eps A + [P,P]$ exists on $D(A(t))$ for every $\eps \in (0,\infty)$ by Theorem~\ref{thm: Kisynski}. Also, $t \mapsto P(t)$ is obviously norm continuously differentiable and so the hypotheses of the second part of Theorem~\ref{thm: erw adsatz ohne sl, zeitabh} are satisfied, which gives the last convergence.
\end{proof}

What are the differences between the above theorem and Bornemann's adiabatic theorem of~\cite{Bornemann98}? 
While in Theorem~IV.1 of~\cite{Bornemann98} $\lambda(t)$ is required to belong to the discrete spectrum of $A(t)$ (and hence to be an isolated eigenvalue) for every $t \in I$, 
in the above theorem it is only required that $\lambda(t)$ has finite multiplicity for almost every $t \in I$: the eigenvalues $\lambda(t)$ are allowed to have infinite multiplicity on a set of measure zero and, moreover, they are allowed to be non-isolated in $\sigma(A(t))$ for every $t \in I$.
Also, the regularity conditions on $A$ and $P$ of the above theorem are slightly weaker than those of Theorem~IV.1: for instance, $t \mapsto \tilde{A}_0(t)$ is required to be twice NOT-continuously differentiable in~\cite{Bornemann98} whereas above it is only required that $t \mapsto a(t)(x,y)$ be twice continuously differentiable for $x, y \in H^+$ (or equivalently (Lemma~\ref{lm: allg lm}), that $t \mapsto \tilde{A}_0(t)$ be twice WOT-continuously differentiable). 
%
And finally, the statement of the theorem above is 
a bit more general than the conclusion of Theorem~IV.1 in~\cite{Bornemann98} which says that, for all $x \in H^+$ (and hence for all $x \in H$)  
and uniformly in $t \in I$,
\begin{align*}
& \scprd{ U_{\eps}(t)x, P(t)U_{\eps}(t)x } 
 = \scprd{ U_{\eps}(t)x, P(t)U_{\eps}(t)x - U_{\eps}(t)P(0)x } + \scprd{ x, P(0)x } \\
& \qquad \qquad \qquad \qquad \quad \longrightarrow  \scprd{x, P(0)x} \quad (\eps \searrow 0). 
\end{align*} 

\vspace{2,2cm}

Acknowledgement: I would like to thank Professor Marcel Griesemer very much for useful discussions and for having directed my interest towards adiabatic theory in the first place.

\vspace{1,2cm}


\begin{thebibliography}{99}
\bibitem{AbouSalemFröhlich05} W. Abou Salem, J. Fröhlich: \emph{Adiabatic theorems and reversible isothermal processes.} Lett. Math. Phys. \textbf{72} (2005), 153-163.
\bibitem{AbouSalem07} W. Abou Salem: \emph{On the quasi-static evolution of nonequilibrium steady states.} Ann. Henri Poincaré \textbf{8} (2007), 569-596.
\bibitem{AbouSalemFröhlich07} W. Abou Salem, J. Fröhlich: \emph{Adiabatic theorems for quantum resonances.} Comm. Math. Phys. \textbf{237} (2007), 651-675.
\bibitem{AcquistapaceTerreni87} P. Acquistapace, B. Terreni: \emph{A unified approach to abstract linear nonautonomous parabolic equations.} Rend. Sem. Mat. Univ. Padova \textbf{78} (1987), 47-107.
\bibitem{Amann95} H. Amann: \emph{Linear and quasilinear parabolic problems I. Abstract linear theory. Birkhäuser, 1995.}
\bibitem{ArendtBatty00} W. Arendt, C. Batty, M. Hieber, F. Neubrander: \emph{Vector-valued Laplace transforms and Cauchy problems. Birkhäuser, 2001.}
\bibitem{AvronElgart99} J.E. Avron, A. Elgart: \emph{Adiabatic theorem without a gap condition.} Commun. Math. Phys. \textbf{203} (1999), 445-463.
\bibitem{AvronGraf11} J.E. Avron, M. Fraas, G.M. Graf, P. Grech: \emph{Adiabatic theorems for generators of contracting evolutions.} ArXiv preprint, 2011.
\bibitem{AvronSeilerYaffe87} J.E. Avron, R. Seiler, L.G. Yaffe: \emph{Adiabatic theorems and applications to the quantum Hall effect.} Commun. Math. Phys. \textbf{110} (1987), 33-49. (In conjunction with the corresponding erratum of 1993.)
\bibitem{BirmanSolomjak} M.S. Birman, M.Z. Solomjak: \emph{Spectral theory of self-adjoint operators in Hilbert space.} Kluwer, 1987.
\bibitem{BlankExnerHavlicek} J. Blank, P. Exner, M. Havl\'{i}\v{c}ek: \emph{Hilbert space operators in quantum physics.} American Institute of Physics Press, 1994.
\bibitem{BornFock} M. Born, V. Fock: \emph{Beweis des Adiabatensatzes.} Z. Phys. \textbf{51} (1928), 165-180.
\bibitem{Bornemann98} F. Bornemann: \emph{Homogenization in time of singularly perturbed mechanical systems.} Lecture Notes in Mathematics \textbf{1687}, Springer 1998.
\bibitem{BrouderPanatiStoltz10} C. Brouder, G. Panati, G. Stoltz: \emph{Gell-Mann and Low formula for degenerate unperturbed states.} Ann. Henri Poincaré \textbf{10} (2010), 1285-1309.
\bibitem{Conway:fana} J. B. Conway: \emph{A Course in Functional Analysis.} 2nd edition. Springer, 1990.
\bibitem{DaleckiiKrein50} J. L. Daleckii, S. G. Krein: \emph{On differential equations in Hilbert space.} Ukrain. Mat. Z. \textbf{2} (1950), 71-91.
\bibitem{Dunford58} N. Dunford: \emph{A survey of the theory of spectral operators.} Bull. Amer. Math. Soc. \textbf{64} (1958), 217-274.
\bibitem{DunfordSchwartz} N. Dunford, J. T. Schwartz: \emph{Linear operators I-III.} Wiley, 1958, 1963, 1971.  
\bibitem{ElgartHagedorn11} A. Elgart, G. A. Hagedorn: \emph{An adiabatic theorem for resonances.} Comm. Pure Appl. Math. \textbf{64} (2011), 1029-1058.
\bibitem{EngelNagel} K.-J. Engel, R. Nagel: \emph{One-parameter semigroups for linear evolution equations.} Springer, 2000.
\bibitem{FujieTanabe73} Y. Fujie, H. Tanabe: \emph{On some parabolic equations of evolution in Hilbert space.} Osaka J. Math. \textbf{10} (1973), 115-130.
\bibitem{GesztesyTkachenko09} F. Gesztesy, V. Tkachenko: \emph{A criterion for Hill operators to be spectral operators of scalar type.} J. Anal. Math. \textbf{107} (2009), 287-353.
\bibitem{GohbergGoldbergKaashoek} I. Gohberg, S. Goldberg, M. A. Kaashoek: \emph{Classes of linear operators I-II.} Birkhäuser, 1990, 1993.
\bibitem{JoyePfister93} A. Joye, C.-E. Pfister: \emph{Superadiabatic evolution and adiabatic transition probability between two non-degenerate levels isolated in the spectrum.} J. Math. Phys. \textbf{34} (1993), 454-479.
\bibitem{Joye07} A. Joye: \emph{General adiabatic evolution with a gap condition.} Commun. Math. Phys. \textbf{275} (2007), 139-162. 
\bibitem{Kato50} T. Kato: \emph{On the Adiabatic Theorem of Quantum Mechanics.} J. Phys. Soc. Japan \textbf{5} (1950), 435-439.
\bibitem{Kato56} T. Kato: \emph{On linear differential equations in Banach spaces.} Comm. Pure Appl. Math. \textbf{9} (1956), 479-486.
\bibitem{KatoTanabe62} T. Kato, H. Tanabe: \emph{On the abstract evolution equation.} Osaka Math. J. \textbf{14} (1962), 107-133.
\bibitem{Kato70} T. Kato: \emph{Linear evolution equations of "hyperbolic" type.} J. Fac. Sci. Univ. Tokyo \textbf{17} (1970), 241-258.
\bibitem{Kato73} T. Kato: \emph{Linear evolution equations of "hyperbolic" type II.} J. Math. Soc. Japan \textbf{25} (1973), 648-666.
\bibitem{KatoPerturbation80} T. Kato: \emph{Perturbation theory for linear operators.} 2nd edition. Springer, 1980.
\bibitem{Kato85} T. Kato: \emph{Abstract differential equations and nonlinear mixed problems.} Lezioni Fermiane, Accademia Nazionale dei Lincei, Scuola Normale  Superiore, Pisa (1985), 1-89.
\bibitem{KelerTeufel12} J. v. Keler, S. Teufel: \emph{Non-adiabatic transitins in a massless scalar field.} ArXiv preprint, 2012.
\bibitem{Kisynski63} J. Kisy\'{n}ski: \emph{Sur les opérateurs de Green des problèmes de Cauchy abstraits.} Stud. Math. \textbf{23} (1963), 285-328.
\bibitem{Krein71} S. G. Krein: \emph{Linear differential equations in Banach space.} Transl. Math. Monographs, American Mathematical Society, 1971.
\bibitem{Lay70} D. C. Lay: \emph{Spectral analysis using ascent, descent, nullity and defect.} Math. Ann. \textbf{184} (1970), 197-214.
\bibitem{LumerPhillips61} G. Lumer, R. S. Phillips: \emph{Dissipative operators in a Banach space.} Pacific J. Math. \textbf{11} (1961) 679-698.
\bibitem{Nenciu80} G. Nenciu: \emph{On the adiabatic theorem of quantum mechanics.} J. Phys. A: Math. Gen. \textbf{13} (1980), 15-18.
\bibitem{NenciuRasche92} G. Nenciu, G. Rasche: \emph{On the adiabatic theorem for non-self-adjoint Hamiltonians.} J. Phys. A: Math. Gen. \textbf{25} (1992), 5741-5751. 
\bibitem{Nenciu93} G. Nenciu: \emph{Linear adiabatic theory. Exponential estimates.} Commun. Math. Phys. \textbf{152} (1993), 479-496.
\bibitem{NickelSchnaubelt98} G. Nickel, R. Schnaubelt: \emph{An extension of Kato's stability condition for non-autonomous Cauchy problems.} Taiw. J. Math. \textbf{2} (1998), 483-496.
\bibitem{Nickel00} G. Nickel: \emph{Evolution semigroups and product formulas for nonautonomous Cauchy problems.} Math Nachr. \textbf{212} (2000), 101-115.
\bibitem{Pazy83} A. Pazy: \emph{Semigroups of linear operators and applications to partial differential equations.} Springer, 1983.
\bibitem{Phillips53} R. S. Phillips: \emph{Perturbation theory for semi-groups of linear operators.} Trans. Amer. Math. Soc. \textbf{74} (1953), 199-221.
\bibitem{ReedSimon} M. Reed, B. Simon: \emph{Methods of modern mathematical physics I-IV.} Academic Press, 1980, 1975, 1979, 1978.
\bibitem{dipl} J. Schmid: \emph{Adiabatensätze mit und ohne Spektrallückenbedingung.} Diplomarbeit Universität Stuttgart. 2010. (Available on arXiv.)
\bibitem{evol} J. Schmid: 
\emph{A note on the well-posedness of non-autonomous linear evolution equations.} ArXiv preprint, 2012.
\bibitem{Schnaubelt02} R. Schnaubelt: \emph{Well-posedness and asymptotic behaviour of non-autonomous linear evolution equations.} Progr. Nonlinear Differential Equations Appl. \textbf{50} (2002), 311-338.
\bibitem{Tanabe60} H. Tanabe: \emph{On the equation of evolution in a Banach space.} Osaka Math. J. \textbf{12} (1960), 363-376. 
\bibitem{Taylor58} A. E. Taylor: \emph{Introduction to functional analysis.} Wiley and Sons, 1958.
\bibitem{TaylorLay80} A. E. Taylor, D. C. Lay: \emph{Introduction to functional analysis.} 2nd edition. Wiley, 1980.
\bibitem{Teufel01} S. Teufel: \emph{A note on the adiabatic theorem without gap condition.} Lett. Math. Phys. \textbf{58} (2001), 261-266.
\bibitem{Teufel03} S. Teufel: \emph{Adiabatic perturbation theory in quantum dynamics.} Lecture Notes in Mathematics \textbf{1821}. Springer, 2003.
\bibitem{Yosida80} K. Yosida: \emph{Functional analysis.} 6th edition. Springer, 1980.
\end{thebibliography}
\end{document}